\documentclass[11pt]{article}
\usepackage{geometry}
\usepackage{lineno}
\usepackage{amssymb}
\usepackage[bottom]{footmisc}
\usepackage{algorithm, algpseudocode}
\usepackage{algcompatible}
\usepackage{amsmath}
\usepackage{amsthm}
\usepackage{epsfig}
\usepackage{booktabs}
\usepackage{graphicx}
\usepackage{graphics}
\usepackage{float}
\usepackage{multirow}
\usepackage{color}
\usepackage{caption}
\usepackage{subcaption}
\usepackage{fullpage}
\usepackage{tikz-qtree}
\usetikzlibrary{positioning}
\usepackage[normalem]{ulem} 
\usepackage{makeidx}
\usepackage{xspace}
\usepackage{wrapfig}
\usepackage{tikz}
\usepackage{hyperref}
\usepackage{colortbl}
\usepackage{upgreek}
\hypersetup{
colorlinks=true,
citecolor=blue,
linkcolor=blue,
filecolor=blue,      
urlcolor=blue,
}
\usepackage{setspace}
\usepackage[color=orange!20]{todonotes} \presetkeys{todonotes}{inline}{}

\usepackage{natbib}
\definecolor{ForestGreen}{RGB}{34,139,34}

\newtheorem{proposition}{Proposition}
\newtheorem{assumption}{Assumption}

\makeatletter
\newtheorem*{rep@theorem}{\rep@title}
\newcommand{\newreptheorem}[2]{%
\newenvironment{rep#1}[1]{%
 \def\rep@title{#2 \ref{##1}}%
 \begin{rep@theorem}}%
 {\end{rep@theorem}}}
\makeatother

\newreptheorem{proposition}{Proposition}
\newcommand{\Nik}{\mathcal{N}_{ik}}

\newcommand{\dik}{{N}_{ik}}
\newcommand{\vW}{\mathbf{W}}
\newcommand{\vw}{\mathbf{w}}

\newcommand{\vX}{\mathbf{X}}
\newcommand{\vXik}{\mathbf{X}_{ik}}
\newcommand{\vx}{\mathbf{x}}
\newcommand{\Yik}{{Y}_{ik}}
\newcommand{\Yjk}{{Y}_{jk}}
\newcommand{\Wik}{{W}_{ik}}
\newcommand{\Gik}{{G}_{ik}}
\newcommand{\Wjk}{{W}_{jk}}
\newcommand{\Gjk}{{G}_{jk}}

\newcommand{\E}{\mathbb{E}}
\newcommand{\Var}{\mathbb{V}}

\usepackage{dsfont}
\newcommand{\I}{\mathds{1}}

\usepackage[affil-it]{authblk}

 \makeatletter
\let\@fnsymbol\@alph
\makeatother

\newcommand{\review}[1]{{\color{black} #1}}
\newcommand{\movedtext}[1]{{\color{black} #1}}

\begin{document}

\begin{titlepage}
\title{Heterogeneous Treatment and Spillover Effects \\ Under Clustered Network Interference
\footnotetext{We are grateful for conversations with Julie Josse, Fabrizia Mealli, Karim Lounici, Michael Lechner, and Hyun Min Kang and for comments to participants at the ``IMS International Conference on Statistics and Data Science'' (ICSDS), at the ``American Causal Inference Conference'' (ACIC), at the ``Causal Inference Workshop'' at the SAMSI Institute, at the ``Causal Machine Learning Workshop" organized by the University of St.Gallen, at the reading groups on ``Causal Inference and Machine Learning" at Harvard University and on ``Machine Learning and Networks in Economics" at IMT School for Advanced Studies Lucca, and to members of the \textit{Missing Values and Causality Research Group} at the \'{E}cole Polytechnique. \\ Bargagli-Stoffi and Tortù are equally contributing co-first authors. \\ Corresponding author: laura.forastiere@yale.edu } }
\author[1]{Falco J. Bargagli-Stoffi}
\author[2]{Costanza Tort\'u}
\author[3]{Laura Forastiere}
\affil[1]{Harvard University}
\affil[2]{Sant'Anna School for Advanced Studies}
\affil[3]{Yale University}
\date{}

\maketitle
\vspace{-1cm}
\begin{abstract}
The bulk of causal inference studies rule out the presence of interference between units. However, in many real-world scenarios, units are interconnected by social, physical, or virtual ties, and the effect of the treatment can spill from one unit to other connected individuals in the network. In this paper, we develop a machine learning method that uses tree-based algorithms and a Horvitz-Thompson estimator to assess the heterogeneity of treatment and spillover effects with respect to individual, neighborhood, and network characteristics in the context of \review{clustered networks and neighborhood interference within clusters.} The proposed Network Causal Tree (NCT) algorithm has several advantages. \review{First, it allows the investigation of the treatment effect heterogeneity, avoiding potential bias due to the presence of interference.} Second, understanding the heterogeneity of both treatment and spillover effects can guide policy-makers in scaling up interventions, designing targeting strategies, and increasing cost-effectiveness. \review{We investigate the performance of our NCT method using a Monte Carlo simulation study, and we illustrate its application to assess the heterogeneous effects of information sessions on the uptake of a new weather insurance policy in rural China.}
\noindent \\
\vspace{0.5cm}

\bigskip
\end{abstract}
\setcounter{page}{0}
\thispagestyle{empty}
\end{titlepage}
\pagebreak \newpage

\maketitle

\linespread{1}\selectfont
\section{Introduction} \label{sec: intro}

\subsection{Motivation}

According to \cite{cox1958planning}, there is \textit{interference} between different units if the outcome of one unit is affected by the treatment assignment of other units. In the case of policy interventions or socio-economic programs, interference may arise due to social, physical, or virtual interactions. For instance, in the case of infectious diseases, unprotected individuals can still benefit from public health measures undertaken by the rest of the population, such as vaccinations or preventive behaviors, because these interventions reduce the reservoir of infection \citep{bridges2000effectiveness}, the vector of transmission \citep{howard2000evidence} and the number of susceptible individuals \citep{kissler2020projecting}. In the labor market, job placement assistance can affect job seekers using this service, but it can also have an effect on other job seekers who are competing in the same labor market \citep{mckenzie2015direct}. In education, learning programs may spill over to untreated peers through knowledge transmission paths \citep{forastiere2019museums}. In marketing, the exposure to an advertisement might directly affect the consuming behavior of the exposed individuals and indirectly affect other individuals that are influenced by the consumer choices of people in their social network \citep{parshakov2020spillover}. If the exposure to the advertisement takes place in social media, the spillover effect might be also explained by the propagation of the advertisement to non-exposed users that are virtually connected to the targeted ones \citep{chae2017spillover}. In the presence of interference, the effect of the treatment status of other units on one's outcome is usually referred to as \textit{spillover effect}. 

Spillovers are a crucial component in understanding the full impact of an intervention at the population level. In fact, the scale-up phase of a policy requires knowledge about the mechanism of spillover, and how this would affect the population once the intervention is rolled out. Information about spillovers of public policies can also support decisions about how best to deliver interventions and can be used to guide public funds allocation. Indeed, the presence of beneficial spillover effects allows for treating a lower percentage of the population, because the untreated individuals would still benefit from the treatment provided to the targeted sample. The use of spillover effects to save resources could be further improved if we were able to target those individuals who would benefit the most from indirect exposure to the intervention. 

A targeting (or seeding) strategy aims at delivering the intervention to certain individuals such that the impact on the overall population is maximized \citep[see, e.g.,][]{Valente2012,kim2015,networkmultipliers2019,targeting2020}. Typically, seeding strategies are designed in settings where either an element of the intervention (e.g., information, flyers, or coupons provided during the intervention) or the outcome (e.g., the adoption of a behavior or a product) diffuses through the network. In these settings, the goal is the identification of the set of nodes in the network that, if targeted, would maximize contagion or diffusion cascades. To do so, seeding strategies are designed using the information on the network structure and the dynamics of contagion or diffusion. This \textit{influence maximization} problem is NP-hard and computer scientists have developed approximate algorithms that usually rely on simplified contagion processes \citep{kempe2003}. Indeed, it is typically assumed that susceptibility to direct treatment and to others, as well as the influential power, are homogeneous across individuals. 

Here we take a different perspective. First, we investigate the spillover effects of a unit's treatment on other units' outcomes without specifying the mechanism through which this might take place. Second, we focus on the assessment of the heterogeneity of susceptibility to direct and indirect treatment. Understanding these heterogeneities can guide the design of targeting strategies aimed at making the interventions more cost-effective. When spillovers are not present, this can be achieved by targeting those with the highest treatment effect. In fact, in the field of personalized medicine, it is well known that individuals with different characteristics might respond differently to the treatment \citep[see, e.g.,][]{murphy2003,murphy2014,laber2019,kent2020predictive}. In the presence of interference, we also have that different people might be more or less susceptible to the treatment received by other units. This means that not only the treatment effect but also spillover effects are heterogeneous. An assessment of the heterogeneity of treatment and spillover effect is crucial not only for designing a cost-effective roll-out of the intervention within the targeted population, but it can also allow a generalization of its impact to other populations.

\subsection{Related Work}

Recently, the causal inference literature has seen a growing interest in the mechanism of interference, leading to (i) the assessment of bias for causal effects estimated under the no-interference assumption \citep{forastiere2016identification, sobel2006}, (ii) the design of experiments to either avoid or assess interference \citep{viviano2020experimental, angelucci2015program,baird2018optimal,
kang2016peer}, and (iii) the estimation of causal effects under interference. Estimators for treatment and spillover effects have first been developed under the assumption of partial interference, allowing interference within groups but not across different groups \citep{hudgens2008toward, tchetgen2012causal, liu2014large, Liu2016, forastiere2016clusters, basse2018analyzing, forastiere2019museums}. However, the assumption of group-level interference is often invalid, too broad, or not applicable. Hence, several works focus on the estimation of causal effects in the context of units interconnected on networks, both in randomized experiments \citep{aronow2017estimating, AtheyEcklesImbens2018, leung2020} and in observational studies \citep{ogburn2022causal, Sofrygin2017, forastiere2016identification, 
forastiere2018estimating, Tortu2023, Lee2023}. In the context of social networks, even in randomized experiments where the treatment is randomized at the unit level, exposure to other units' treatment is not and depends on the network structure.

In parallel to this field of research on interference, in recent years, thanks to the availability of increased computing power and large data sets, researchers have started to think about advanced data-driven methods to assess the heterogeneity of treatment effects with respect to large numbers of features. In this regard, there has been a large number of contributions in causal inference on subgroup analysis to investigate heterogeneous effect \citep[e.g.,][]{assmann2000subgroup, cook2004subgroup, rothwell2005subgroup, su2009subgroup, varadhan2013estimation, ratkovic2017sparse}. However, the standard methods for subgroup analysis have several drawbacks. In particular, (1) they strongly rely on the subjective decisions on the specific variables defining the heterogeneous sub-populations; (2) they fail to discover heterogeneities other than the ones that are \textit{a priori} defined by the researchers. In addition, a  data-driven method avoids potential problems related to cherry-picking the subgroups with extremely high/low treatment effects \citep{assmann2000subgroup, cook2004subgroup}. Hence, many data-driven algorithms for the estimation of heterogeneous causal effects have been proposed in recent years \citep{hill2011bayesian, foster2011subgroup, su2012facilitating, athey2016recursive,  wager2018estimation, athey2019generalized, lechner2018modified, hahn2020bayesian}.\footnote{For a review of these methods the reader can refer to \cite{athey2019machine} and \cite{dominici2021controlled}.} 

The underlying aim of these methods is to detect `causal' rules defining subsets of the study population where the treatment effect for that subgroup deviates from the average treatment effect. This is done by selecting the most important features and their values that define a partition of the covariate space (the tree) where the treatment effect is `significantly' heterogeneous.
Among these algorithms, many rely on extensions of the standard Classification and Regression Trees (CART) algorithm \citep{friedman1984classification} and Random Forest (RF) algorithm \citep{breiman2001random} and are adapted to different settings \citep{athey2016recursive, zhang2017mining, lee2018discovering, bargagli2018estimating, guber2018instrument, johnson2019detecting, bargagli2019causal, bargagli2022heterogeneous}. In particular, in their seminal contribution \cite{athey2016recursive} develop the honest causal tree (HCT) methodology. HCT is a causal decision tree algorithm that is aimed at discovering the heterogeneity in causal effects through a single binary tree. HCT is constructed using a criterion function aimed at maximizing the heterogeneity in causal effects at each split while penalizing splits with higher variance in the estimated conditional effects.\footnote{The criterion function is the main difference between HCT and CART. Indeed, CART's criterion function is aimed at minimizing the empirical predictive error at each split.} In addition to a modification of the criterion function, \cite{athey2016recursive} introduce the concept of \textit{honest splitting}. The authors propose to split the overall learning sample into a \textit{discovery} (or training) set and an \textit{estimation} set. The former is used to discover the heterogeneity in treatment effects, while the latter is used to estimate these effects in the discovered sub-populations. The different role of the two sets avoids potential spurious heterogeneity discovery due to overfitting of the learning sample. Building on the HCT, \cite{wager2018estimation} and \cite{athey2019generalized} introduce the causal forest (CF) and generalized random forest (GRF). Both CF and GRF are ensemble methods where multiple trees are built and combined to improve inference on the heterogeneous causal effects.

HCT and similar tree-based methodologies have already been applied to various scenarios for the discovery of heterogeneous effects of air pollution \citep{lee2018discovering}, employment incentives \citep{bargagli2019causal}, job training programs \citep{cockx2019priority}, development finance projects \citep{zhao2017quantifying}, cardiovascular surgeries \citep{wang2017personalized}, cancer treatments \citep{zhang2017mining}, and health insurance \citep{johnson2019detecting}. The wide usage of tree-based algorithms is due, in particular, to their ability to deal with non-parametric settings in an efficient and interpretable way. Indeed, these algorithms do not assume any specific shape of the treatment effect function. HCT and similar methodologies built on CART have been widely employed because of various attractive features: i.e., (i) they can deal with a large number of variables that are potentially responsible for the heterogeneity; (ii) they are simple to understand and visualize, easy to interpret, computationally scalable; and (iii) they can deal with non-linear relationships in the covariates without the need of data pre-processing.

Nevertheless, tree-based methods for the estimation of heterogeneous causal effects have been developed ruling out the presence of spillover effects by assuming no interference between units.  On the other hand, the growing literature on spillover effects has focused on the estimation of population average spillover effects, neglecting potential heterogeneous spillover effects.

There have been few articles dealing with different types of heterogeneity in spillover effects. \cite{forastiere2016clusters, forastiere2019museums} estimated the heterogeneity of spillover effects with respect to principal strata defined by the compliance behaviors in response to a cluster randomized treatment. However, the latent nature of these strata makes it difficult to effectively use the detected heterogeneity to design targeting strategies or to generalize the conclusions to a different population. Observed heterogeneity is instead studied in \cite{arduini2020identification} and \cite{arduini2019treatment} where the focus is on the heterogeneity of peer effects from other units' outcomes across two specific groups, and the estimation relies on linear-in-means models and two-stage least squares. To the best of our knowledge, there are no studies dealing with the heterogeneity of spillover effects on networks.

\subsection{Contributions}

In this paper, we bridge the gap between the aforementioned two bodies of causal inference literature by proposing a new algorithm for the discovery and estimation of heterogeneous treatment and spillover effects with respect to a large number of characteristics, including individual and network features. Our method is designed for randomized experiments affected by the presence of clustered network interference, that is, units are organized in a clustered structure, with no interactions between clusters and a network of connections within clusters (e.g., friendship networks within schools). In addition, randomization is assumed at the individual level, resulting in treated and untreated units in the same cluster. Under this setting, spillover effects are confined within clusters and are assumed to take place on network interactions.
\review{Our research tackles this gap in the literature by introducing a novel methodology capable of simultaneously discovering and estimating heterogeneity in both treatment and spillover effects. Through the introduction of a novel composite splitting function, our approach facilitates a comprehensive exploration of heterogeneity and enables the simultaneous identification of regions in the covariate space where there is heterogeneity in both treatment and spillover effects, informing policy-makers on sub-populations that can be more or less prioritized due to their response to their own treatment or to the treatment of their social ties. As far as we are aware, no prior study has put forth such a comprehensive framework.}

Our proposed method, \textit{network causal tree} (NCT), builds upon the \textit{causal trees} proposed by \cite{athey2016recursive}, by modifying the splitting criterion to target treatment and spillover effects under interference.  Splits are made so as to maximize the heterogeneity of the targeted causal effect(s), treatment, and/or spillover effects, across the population.  This criterion relies on the unbiasedness of the estimator of the effect(s) within each subset of the population. We first contribute to the existing literature on interference by developing an unbiased estimator for conditional treatment and spillover effects. We extend the Horvitz-Thompson estimator in \cite{aronow2017estimating} to conditional causal effects under clustered network interference and we prove its consistency under the clustered network setting. This estimator is then used in our NCT algorithm to choose the binary splits that maximize the heterogeneity and to finally estimate the heterogeneous causal effects within the selected sub-populations. \review{Our work further extends the theoretical properties of the Horvitz-Thompson developed by \cite{aronow2017estimating} to encompass conditional leaf-specific effects in a clustered network setting. This extension constitutes a significant contribution, as it enables the consideration of novel network structures.}

In order to use the selected partition of the covariate space and the estimated treatment and spillover effects to guide policies, the heterogeneous sub-populations should be identified based on the causal effects that will be part of the decision rule. For instance, a policy that assigns the treatment to those who benefit the most from it requires the discovery of the subsets of the populations with heterogeneous treatment effects. Alternatively, a policy that is designed to target those who will respond to both their own treatment and the neighbors' treatment must identify sub-populations with high treatment and spillover effects. Hence, the proposed NCT methodology is optimized to detect the heterogeneity in treatment and spillover effects (i) either \textit{simultaneously} or (ii) \textit{separately}. Indeed, by reworking the criterion function of the seminal causal trees algorithm to account for interference, we allow the algorithm to detect heterogeneity in treatment and/or spillover effects. 

On the one side, the discovery of the causal rules (the variables and their values defining heterogeneous sub-populations) representing the heterogeneity of one causal effect (treatment or spillover) is achieved by using a splitting criterion that maximizes the heterogeneity of that specific causal effect across the partition. On the other side, if our goal is to identify a partition of the covariate space that can explain the heterogeneity of multiple causal effects, we propose the use of a composite splitting function that is designed to simultaneously maximize the heterogeneity in all the effects. This flexibility allows scholars and policy-makers to customize their investigations depending on their targeting goals. For instance, if a policy-maker wants to target individuals with the highest treatment effect and the lowest spillover effect (with the motivation that those with higher spillover effects can benefit from the treatment received by others), the NCT algorithm should be implemented with a composite splitting function aimed at detecting subsets of the population where both treatment and spillover effects are heterogeneous and the decision criterion can be applied. Conversely, if a targeting strategy is designed to target just individuals who would benefit the most from receiving the treatment, regardless of other people's assignment, a tree would be built using a single splitting criterion targeted to maximize the treatment effect heterogeneity. Similarly, a single criterion targeted to a spillover effect would be used in the case of targeting rules only involving that spillover effect.

It is important to note that the use of our algorithm to design implementation strategies is possible thanks to its high level of interpretability. NCT provides interpretable inference on heterogeneous treatment and spillover effects by discovering a set of causal rules that can be represented through a binary tree. As argued by \cite{lee2018discovering} and \cite{lee2020causal}, it is important to provide interpretable information on simple causal rules that can be targeted to improve policy effectiveness and to ensure that stakeholders and policy-makers understand (and, in turn, trust) the functionality of these models. \cite{valdes2016mediboost} claim that a learning algorithm is interpretable if one can explain its classification by a conjunction of conditional statements (i.e., if-then rules). In this regard, tree-based algorithms based on if-then rules, such as the proposed NCT, are optimal for interpretability.

To assess the performance of the proposed NCT algorithm, we run a series of Monte Carlo simulations. In particular, we investigate the performance of the proposed algorithm with respect to two dimensions: its ability (i) to correctly identify the actual heterogeneous sub-populations and, (ii) to precisely estimate the conditional treatment and spillover effects. While the latter performance assessment is quite standard in the literature, the former is critical for interpretable algorithms for heterogeneous causal effects \citep{bargagli2022heterogeneous}. Finally, we apply the proposed NCT algorithm to a randomized experiment conducted in China to assess the impact of information sessions on the purchase of a new weather insurance policy \citep{cai2015social}. Besides estimating the population average treatment and spillover effects (as already investigated in \cite{cai2015social}), our aim is to detect the strata of the population where one or both effects are heterogeneous and estimate these effects within these strata. \review{The proposed NCT is implemented in the \texttt{NetworkCausalTree} open-source \texttt{R} package, which can be found at \href{https://github.com/fbargaglistoffi/NetworkCausalTree}{\url{https://github.com/fbargaglistoffi/NetworkCausalTree}}.}

The remainder of the paper is organized as follows. 
\review{In Section \ref{sec:motivating_application} we introduce the motivating application of the paper and the empirical setting.} In Section \ref{sec:notation} we introduce the notation, setting, and assumptions that we employ throughout the paper.  In Section \ref{sec:cace} we define the conditional causal effects in a general partition of the covariate space and develop a Horvitz-Thompson estimator. Section \ref{sec: ctrees} presents the proposed NCT algorithm, which is based on effect-specif or composite splitting functions for causal effects under interference. We then conduct a simulation study to assess the performance of the algorithm and estimator under different scenarios in Section \ref{sec:simulations} and we illustrate the application of the network causal tree on a randomized experiment in Section \ref{sec:application}. Section \ref{sec:conclusions} concludes the paper with a discussion of the proposed algorithm and directions for further research. 

\section{Motivating Application}\label{sec:motivating_application}
\review{
\subsection{Agricultural Insurance Policies against Extreme Weather Events}}

\review{In 2021 alone, worldwide, there were 432 disasters related to extreme weather events that killed 10,492 people, affected 102 million others, and incurred nearly US\$ 252 billion in damages \citep{crunch2022}. Climate-related disasters are expected to increase in frequency and severity in the future due to global warming \citep{ebi2021extreme}, posing an increasing burden on vulnerable communities \citep{lal2011socio, rogers2015resettlement, huq2015climate}.\footnote{The number of disasters, affected people and costs have steadily gone up in the last few years---i.e., disaster-related costs have surged by US\$ 122 in a two years span between 2019 and 2021 \citep{crunch2020}.} Asia is often the most severely impacted continent: in the same year, it suffered 40\% of all world's disasters and accounted for 49\% of the total number of deaths and 66\% of the total number of people affected \citep{crunch2022}. Among Asian countries, China is particularly exposed to weather hazards \citep{zhao2020extreme}. Within China, agricultural and rural communities have suffered the highest costs: in the past decades, weather hazards and disasters have affected about a quarter to one-third of arable land in China \citep{liu2010analysis}. 

Against this backdrop, agricultural insurance policies play a key role in risk mitigation strategies that can reduce agricultural production risks and provide economic support to farmers \citep{barnett2007weather}. The centrality of agricultural insurance policies has been highlighted by the enactment of individual and institutional level weather insurance policies in several countries in the last two decades \citep{collier2009weather}. In 2012, the Chinese government made an explicit proposal to expand agricultural insurance policies and extend their coverage in rural China \citep{ye2017farmers, jin2016farmers}. 

As governments have taken action to enhance weather insurance uptake and effectiveness in rural areas, the importance of understanding the factors and mechanisms influencing farmers' decisions on purchasing weather insurance has become critical. Recent studies have looked at either (i) the determinants behind farmers' insurance purchase decisions \citep{gine2008patterns, gaurav2011randomized, cole2013barriers, cai2017disaster, cai2015social, sibiko2020weather, dercon2014offering}, or at (ii) the effectiveness of training sessions on insurance uptake in a broad spectrum of rural environments and countries \citep{sibiko2020weather, dercon2014offering}. Connected to the latter literature, \cite{cai2015social} has also investigated the spillover effect of providing a training session to some individuals on the insurance uptake of their social ties through peer influence and the spread of information in China. However, none of these studies has produced a comprehensive, data-driven identification of the determinants of the heterogeneous effectiveness both at the individual level for those receiving the intervention and at the community level for those exposed to the information through spillover from their peers. Nonetheless, this task is crucial as it enables a deeper understanding of the policies and provides room for targeted interventions aimed at maximizing their cost-effectiveness.

The methodology proposed in this paper addresses this shortcoming.
By identifying the variables driving the heterogeneity of both the direct and spillover effects and estimating these effects for each subgroup of individuals with heterogeneous effects, our NCT will be able to identify those who are more likely to respond to training sessions and purchase weather insurance policies as a result of participating in these sessions, as well as to those who are more likely to respond to the influence and information provided by their friends who have attended the training sessions.
}

\review{
\subsection{Empirical Setting: Randomized Experiment in Rural China}

In order to tackle these research questions, we use data from a randomized experiment designed to assess the impact of intensive information sessions to promote the uptake of a new weather insurance policy among rice farmers in rural China \cite{cai2015social}. The promoted policy is designed to protect farmers from extreme weather events that would leave them without an income for potentially long periods of time.

The experiment, conducted in rural villages in the Jiangxi province (located in south-east China), had a factorial design with three factors: i) intensive vs simple information sessions, ii) time of attendance (round 1 or 2), and iii) additional information provided on previous purchase decisions of other village members.\footnote{More details about the randomization design can be found in \cite{cai2015social} and in Section \ref{sec:application}.} 
In the main analysis, we focus on the effects of the intensive information session at round 1, including the direct effect on those who receive it and the spillover effect of having friends who attended the intensive session at round 1. At baseline, rice farmers participating in the study were also asked to identify up to five of their friends among other participants. This network collection gives rise to a binary and directed network in each village, forming a \emph{clustered network}.\footnote{\citep{cai2015social} do not explicitly exclude friendship links between households living in different villages (households are asked to declare their friends living either in their same village or in a different village). However, for the purpose of our analysis, we exclude the few links among households in different villages (these links represent the $0.0007\%$ of the whole possible connections).} 

In this setting, interference is likely to arise. Indeed, households receiving intensive information sessions on the insurance policy may share what they have learned with untreated households in their friendship network, indirectly encouraging them to adopt the policy. Consequently, untreated households might benefit from these sessions through interactions with treated households in the same village. Similarly, households receiving the intensive information session may be more responsive to it if their social ties have received the same information.
}

\section{Clustered Network Interference and Unit-Level Randomization} 
\label{sec:notation}

\label{sec: Eff.interference}
\subsection{Notation and Setting} \label{subsec: RCM.int}

Let us consider a sample $V$ of $N$ units organized in $K$ separate clusters. Let $k\in \mathcal{K}=[1, \ldots, K]$ be the cluster indicator, and let $i=1,\ldots, n_k$ be the unit indicator in each cluster $k$. Let us now consider a connection structure such that units belonging to the same cluster might share a link whereas units belonging to different clusters are not connected. This network structure is represented by the graph $G=(V, \mathcal{E})$, where $V$ defines the set of nodes and $\mathcal{E}$ defines the set of edges, that is, the collection of links between each connected pair of nodes. A clustered network $G$ is in turn an ensemble of $K$ disjoint sub-graphs: $G_{k}=(V_k, \mathcal{E}_k), \; \; k=1, \dots, K$. The adjacency matrix $\boldsymbol{A}$ corresponding to graph $G$ is a block-diagonal matrix with $K$ blocks, $\boldsymbol{A}_k, \; \; k:1, \dots, K$, where each element $a_{ij, k}$ is equal to 1 if there is a link between unit $i$ and unit $j$ in cluster $k$, that is, if the edge $\epsilon_{ij,k}\in \mathcal{E}_k$. Elements in $\boldsymbol{A}$ off the $K$ blocks are equal to zero, indicating no links between units belonging to different clusters. \review{In our weather insurance application, the sub-graph $G_{k}$ represents the friendship network among farmers in village $k$ (see Figure \ref{fig: villages}).}

Let $\Wik \in \{0,1\}$ be a binary variable representing the treatment assigned to unit $i$ in cluster $k$ and let $Y_{ik}$ be the observed outcome. We denote by $\mathbf{W}_k$ and $\mathbf{Y}_k$ the treatment and outcome vectors in each cluster $k$. Similarly, $\mathbf{W}$ and $\mathbf{Y}$ denote the treatment and outcome vectors in the whole sample. \review{In our empirical application, 
the outcome $\Yik$ is the indicator for the insurance uptake by household $i$ in village $k$. Given the factorial design, the definition of the binary treatment variable $\Wik$ depends on the factorial combination of interest. In the main analysis, presented in Section \ref{sec:application}, the treatment variable $\Wik$ equals 1 if household $i$ in village $k$ received an intensive information session at round 1 and 0 otherwise.}
Moreover, for each unit $ik$ we observe a vector $\textbf{X}_{ik}$ of $P$ covariates (or pre-treatment variables) that are not influenced by the treatment assignment. The vector of covariates might include individual characteristics (e.g., age, education, sex, socio-economic status, ...), cluster-level characteristics (e.g., cluster size, location, ...), as well as network characteristics representing aggregated individual characteristics (e.g., average age or proportion of males and females, ...) or the network topology (e.g., degree, centrality, transitivity, ...).

Figure \ref{fig:clustinterference} provides a graphical intuition on the clustered network structure and treatment assignments at the unit level. Edges indicate links between units, within each cluster. Colors refer to the individual treatment assignment: grey-colored nodes represent treated units, while white-colored vertices indicate units assigned to the control group.

\begin{figure}[t]
\centering
		\includegraphics[width=70mm, height=70mm]{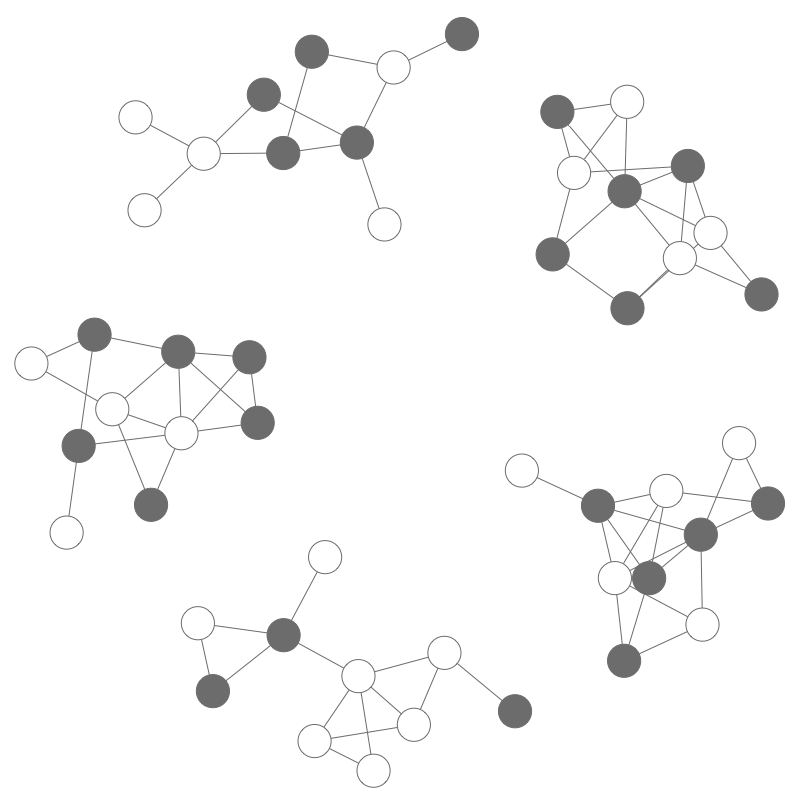}
				\caption{Clustered Network Structure}
		\label{fig:clustinterference}
\end{figure}
\subsection{Clustered Network Interference} \label{subsec: clustint}

Following the potential outcome framework \citep{rubin1974estimating,holland1986statistics}, we denote by $Y_{ik}(\mathbf{w})$ the potential outcome that unit $i$ in cluster $k$ would experience if the treatment vector $\mathbf{W}$ in the whole sample were $\mathbf{w}$, with $\mathbf{w} \in \{0,1\}^N$. Under the assumption of \textit{no interference}, the potential outcome could be indexed only by the individual treatment assignment $W_{ik}$, that is, $Y_{ik}(W_{ik}=w)$. In combination with the assumption of \textit{consistency} (see Assumption \ref{ass:consistency} below), this assumption is known as Stable Units Treatment Value Assumption (SUTVA)\citep{rubin1986comment}. \review{The no-interference assumption is clearly violated in many real-world scenarios, as is the case in our empirical application.}

\review{In this paper, we focus on a particular type of interference: \emph{clustered network interference}.} As we will show later, focusing on this type of interference is critical to ensure asymptotic properties of the estimator for conditional causal effects as well as to allow the network causal tree to divide the sample into training (or discovery) set and estimation set (and testing set, if applicable).\footnote{Alternatively, network causal trees could also be extended to the case of one single network as long as the amount of dependency is limited to ensure the consistency of the estimator. Furthermore, to divide the sample into the different sets needed for the causal tree algorithm, a community detection algorithm could be used to identify separate and densely connected communities. We leave this possible extension to further research.}

Clustered network interference implies that: i) interference is restricted to units of the same cluster and interference between clusters is ruled out, that is, one's outcome is only affected by the treatment received by units belonging to the same cluster; ii) one's outcome is affected by a weighted function of the treatment status of potentially all units in their own cluster, with weights depending on the presence and possibly the value of network links.

Let $W_{k/i}$ be the vector collecting the treatment status of all units in cluster $k$ except unit $i$. Let $g(\cdot): \{0,1\}^{n_k-1}\longrightarrow \Delta_{ik}$ be a function that maps a cluster assignment vector $\vW_{k/i}$ to an exposure value. Without loss of generality, we define it as a function of the dot product between the cluster assignment vector and a vector of weights $\boldsymbol{\delta}_i(A_k, \mathbf{X}_k)$, which in turn depends on the adjacency matrix $A_k$ and the covariate matrix $\mathbf{X}_k$, i.e., $g(\vW_{k/i}, \boldsymbol{\delta}_i(A_k, \mathbf{X}_k))=f(\vW_{k/i} \cdot \boldsymbol{\delta}_i(A_k, \mathbf{X}_k))$. For instance, the function $g(\cdot)$ could result in the number or proportion of treated units in a cluster. In this case, the weight vector would be equal to $\boldsymbol{\delta}_i(A_k, \mathbf{X}_k)=\mathbf{1}_{n_k-1}$ or $\boldsymbol{\delta}_i(A_k,\mathbf{X}_k)=(\frac{\mathbf{1}}{n_k-1})_{n_k-1}$, respectively. Alternatively, we could use the adjacency matrix to compute the geodesic distance $d(ik,jk)$ between each pair of nodes in cluster $k$ and let $g(\vW_{k/i}, \boldsymbol{\delta}_i(A_k, \mathbf{X}_k)=\sum_{j\neq i} \frac{W_{jk}}{d(ik,jk)}$. The function $g(\cdot)$ is similar to the `effective treatments' function in \cite{manski2013} and the `exposure mapping' function in \cite{aronow2017estimating}, although it applies to the cluster treatment vector only. To ease notation, throughout we will omit the weight vector $\boldsymbol{\delta}_i(A_k,\mathbf{X}_k)$ in the function $g(\cdot)$. 

We can now formalize the clustered network interference assumption as follows.
\begin{assumption}[Clustered Network Interference]
\label{ass:cni}
Given a function $g(\cdot): \{0,1\}^{n_k-1}\longrightarrow \Delta_{ik}$, $\forall k \in \mathcal{K}$, $\forall i \in V_k$, and $\forall \,\,\,\vW, \vW'\in \{0,1\}^N$ such that $ W_{ik}=W'_{ik}, \,\, g(\vW_{k/i})=g(\vW'_{k/i})$, the following equality holds: $\Yik(\vW)=\Yik(\vW')$.
\end{assumption}
\noindent Assumption \ref{ass:cni} states that the outcome of unit $i$ in cluster $k$ depends on the \textit{individual treatment} $W_{ik}$ and a function of the treatment status of the other members of cluster $k$, i.e., $g(\vW_{k/i})$, regardless of the specific treatment status of each member.  This assumption can be viewed as an intermediate assumption between (i) assuming no interference and (ii) making no assumptions about the nature of interference. In a way, it is similar to the \textit{partial interference} or the \textit{stratified interference} in \cite{hudgens2008toward}, which are special cases of the clustered network interference assumption, with $g(\vW_{k/i})=\vW_{k/i}$ and $g(\vW_{k/i})=\sum_{j\neq i} W_{jk}$.
Let $G_{ik}=g(\vW_{k/i})$, referred to as \textit{network exposure} throughout. Under Assumption \ref{ass:cni}, each unit has $|\Delta_{ik}| \times 2$ potential outcomes, which we can write in terms of the individual treatment and the network exposure as $Y_{ik}(w,g)$, representing the potential outcome of unit $ik$ under $W_{ik}=w$ and $G_{ik}=g(\vW_{k/i})=g$.

We also assume the following consistency assumption: 
\begin{assumption}[Consistency]
\label{ass:consistency}
$$Y_{ik}(W_{ik}, G_{ik})=Y_{ik}.$$
\end{assumption}
\noindent This assumption rules out different versions of the treatment and different ways in which a value of the network exposure can affect the outcome of a particular unit.
Under a `finite sample perspective', we assume the potential outcomes of each unit to be fixed but unknown, except for the observed $Y_{ik}(W_{ik}, G_{ik})$. Therefore, the only source of randomness in the potential outcomes is given by the random assignment to the treatment and the random network exposure induced by the random cluster assignment.

Assumptions \ref{ass:cni} and \ref{ass:consistency} together are alternatives to SUTVA when interference is present and is limited to within clusters. When the weight function $\boldsymbol{\delta}_i(A_k, \mathbf{X}_k)$ is such that elements ${\delta}_{ij}(A_k, \mathbf{X}_k)=0$ if $j\in V_k: a_{ij,k}=0$, that is, units that are not directly connected to unit $i$ receive a weight equal to zero, then interference is limited to the neighborhood $\Nik$ of each unit, with $\Nik=\{j \in V_k: a_{ij,k}=1\}$. In this case, Assumptions \ref{ass:cni} and \ref{ass:consistency} correspond to the SUTNVA Assumption in \cite{forastiere2016identification}. 
We denote by $\mathcal{N}_{ik}^g$ the set of units defining the network exposure, that is, $\mathcal{N}_{ik}^g=\{j\in V_k: \text{if}\,\, W'_{jk}=W_{jk}\,\, \text{then}\,\, g(\vW'_{k/i})= g(\vW_{k/i}), \,\, \forall \vW'_{k/i}\neq \vW_{k/i} \}=\{j\in V_k: {\delta}_{ij}(A_k, \mathbf{X}_k)\neq 0 \}$. In most of the literature on spillover effects, this set is either the cluster $k$ \citep{hudgens2008toward} or the neighborhood of unit $i$ \citep{forastiere2016identification}. Alternative specifications are also possible and might involve higher-order neighbors. For the purpose of assessing effect heterogeneity using tree-based methods, we will further 
\review{consider a discrete exposure mapping function by making} the following assumption:

 \begin{assumption}[Discrete Network Exposure]
 \label{ass: bing}
  \review{There exists a discrete exposure mapping function $g(\cdot): \{0,1\}^{n_k-1}\longrightarrow \Delta_{ik}\subset \mathbb{Z}$ such that Assumption \ref{ass:cni} holds and $g(\cdot)$ is known and well-specified.}
\end{assumption}
  \noindent $\mathbb{Z}$ is the set of integers. This assumption implies that the network exposure $G_{ik}$ is a discrete variable.
\noindent  For instance, we can define a binary network exposure
based on a threshold function applied to the number of treated neighbors:
\begin{equation}
\label{eq: threshold}
    G_{ik}=\mathds{1}\bigg(\bigl(\sum_{j\in \Nik} W_{jk}\bigl) \geq q \bigg),
\end{equation}
 where $q$ is a threshold. \review{Hence, the network exposure $\Gik$ is equal to 1 if the number of treated network neighbors exceeds a certain threshold $q$ (e.g., at least one treated neighbor is treated, the majority of the neighbors are treated, ...).}
 In our simulation study as well as in the main analysis of the application we have chosen the following definition: $G_{ik}=\mathds{1}(\bigl(\sum_{j\in \Nik} W_{jk}) \geq 1 \bigr)$, that is, the network exposure is 1 if at least one network neighbor is treated. \review{In the empirical application, we also vary the threshold $q$, to assess the robustness of the results.} As a consequence, both the individual treatment and the network exposure are defined as binary variables, $\Wik \; \in \{ 0,1 \}$ and $\Gik \; \in \{ 0,1 \} $. It follows that the support of the joint treatment variable $(\Wik, \Gik)$ is finite and comprises four possible realizations, given by the combination of the two marginal domains. Hence, $(\Wik, \Gik) \; \in \{ (w,g)=(0,0),(1,0),(0,1),(1,1)\} $. 

A discrete network exposure is crucial for our causal tree algorithm, at least in the version proposed in this paper. Indeed, the algorithm relies on the presence of enough observations for each treatment and exposure value to allow the estimation of the causal effects. Depending on the stopping rule which might rely on the accuracy of the estimation of conditional effects or on the number of observations (see Section \ref{sec: ctrees}), if the sample size is not large enough with respect to the number of categories of the network exposure and/or its distribution is non-uniform and highly skewed, the network causal tree algorithm might result in a tree with low depth and low granularity, that is, with highly heterogeneous causal effects even within the terminal leaves. Therefore, the maximum number of categories allowed for the network exposure depends on the sample size, the number of covariates and their nature, as well as on the extent of the heterogeneity in the causal effects.

\subsection{Unit-Level Randomization and Induced Marginal and Joint Distributions}
\label{subsec: bernoulli}

In this work, we consider an experimental design with a unit-level randomization of the treatment, which is independent between clusters but might be dependent within them. Therefore, the treatment vector $\vW$ is a random vector with probability distribution $P(\vW=\vw)$ and the following assumption holds.
\begin{assumption}[Independent treatment allocation between clusters]
\label{ass:ita}
$$
P(\vW=\vw)=\prod_{k=1}^K P(\vW_k=\vw_k)
$$
where $\vW_k$ is the treatment vector in each cluster $k$.
\end{assumption}
We denote by $\pi^W_{ik}$ the unit-level probability that $W_{ik}$ is equal to 1, under the experimental design in place. In a randomized experiment, $\pi^W_{ik}$ is known. In the case of a Bernoulli trial, where each unit is independently assigned to the individual treatment, $\pi^W_{ik}$ is constant and equal to $\alpha$.\footnote{The unit-level assignment probability could also vary across clusters as in two-stage randomization \citep{hudgens2008toward}.} An example of a design with randomization independent between clusters but dependent within clusters is that of a completely randomized experiment taking place in each cluster. In this case, $\pi^W_{ik}$ would be equal to $m_k/n_k$, where $m_k$ is the fixed number of treated units, and the treatment assignment for each unit does depend on the treatment status of other units. 

Since the network exposure is a deterministic function $g(\cdot)$ of the cluster assignment vector $\vW_{k/i}$, then the randomization distribution $P(\vW=\vw)$ induces, together with the definition of the function $g(\cdot)$, a probability distribution of the vector of network exposures $\mathbf{G}$ in the whole sample.  Hence, the probability for a unit of being exposed to a specific value of the network exposure $G_{ik}=g$ given the individual treatment $w$, denoted by $\pi_{ik}^{G|W}(g|w)$, is known and can in principle be computed from the probability distribution $P(\vW)$. Note that, we can drop the dependency from the individual treatment and write $\pi_{ik}^G(g)$ when the randomization is independent between units.

Let $\Delta=\{0,1\} \times \bigcup_{ik \in V} \Delta_{ik}$ be the domain of the joint individual and network treatment status, that is, $(w,g)\in \Delta$. Let $\pi_{ik}(w,g)$ denote the \emph{marginal probability} for unit $ik$ of being assigned to individual treatment $w$ and being exposed to the network status $g$. This is equal to the expected proportion of assignment vectors $\vw$ inducing an individual treatment $w$ and a network exposure $g$:
\begin{equation}
   \pi_{ik}(w,g)=\sum_{\vw\in\{0,1\}^N} \mathds{1}(W_{ik}=w, G_{ik}=g)P(\vW=\vw)=(\pi_{ik}^W)^w(1-\pi_{ik}^W)^{1-w} \times\pi_{ik}^{G|W}(g|w).
\end{equation}
This marginal probability is a crucial component of the Horvitz-Thompson estimator for causal effects under network interference. For instance, if the experimental design is a Bernoulli trial with unit-level probability $\alpha$ and the network exposure is defined by a threshold function on the neighborhood as in Equation \ref{eq: threshold}, then the joint probability could be computed as follows:
\begin{equation}
\label{eq:ber_pi}
\pi_{ik}(w,g)=\alpha^w(1-\alpha)^{1-w} \times \biggl[1- \sum_{l=0}^{h-1} \binom{\dik}{l} p^l(1-p)^{\dik-l}\biggr]^g \biggl[\sum_{l=0}^{q-1} \binom{\dik}{l} p^l(1-p)^{\dik-l}\biggr]^{1-g}
\end{equation}
where $\dik$ is the number of neighbors (`degree') of unit $ik$.

To deal with well-defined potential outcomes, we must assume that each unit has a non-zero probability of being exposed to each $(w,g)$: 
\begin{assumption}[Positivity\label{ass: positivity}]
$\pi_{ik}(w,g)>0 \quad \forall i \in \mathcal{N}, \; k \in \mathcal{K} \; \; \text{and} \; \; \forall (w,g) \in \Delta$.
\end{assumption}
\noindent When $\pi_{ik}(w,g)=0$ for some units, then the average potential outcomes and causal
effects involving these values $z$ and $g$ must be restricted to the subset of units for whom $\pi_{ik}(w,g)>0$.  For instance, if the network exposure is defined as in Equation \ref{eq: threshold}, then the positivity assumption is violated for units that cannot be exposed to a value $g$, that is, those with a degree $\dik$ lower than the threshold $q$. Consequently, the analysis must be restricted only to the subset of the population satisfying the positivity criterion. 

The estimator that we propose below also requires the so-called \emph{pairwise exposure probabilities}, which describe the joint probability for pairs of units of being exposed to a given individual treatment and network status. Hence, given specific exposure conditions $ (w,g)\; \text{and} \; (w',g')$ , a pairwise exposure probability, denote by $\pi_{ikjh}(w,g;w',g')$, quantifies the probability that the two events ($\Wik=w,\Gik=g$) and ($W_{jh}=w',G_{jh}=g'$)  occur---i.e., $\pi_{ikjh}(w,g;w',g')=P(\Wik=w,\Gik=g,W_{jh}=w',G_{jh}=g')$. In general, this can be written as:
 \begin{equation}
   \pi_{ikjk'}(w,g;w',g')=\sum_{\vw\in\{0,1\}^N} \mathds{1}(W_{ik}=w, G_{ik}=g, W_{jk'}=w', G_{jk'}=g')P(\vW=\vw).
\end{equation}
Under the event of both units being exposed to the same condition $(w,g)$ we denote the pairwise exposure probability by $\pi_{ikjh}(w,g)$.
 
In the case of an experimental design assigning treatment independently between clusters, under the clustered network interference the two events ($\Wik=w,\Gik=g$) and ($W_{jh}=w',G_{jh}=g'$), with $k\neq h$, are independent and the pairwise exposure probability equals the product of the two joint probabilities:
 $
 \pi_{ikjh}(w,g;w',g')=\pi_{ik}(w,g) \times \pi_{jh}(w',g')
 $.
In the case of a Bernoulli trial and the network exposure defined on the neighborhood only, this is also true for units $i$ and $j$ belonging to the same cluster, i.e., $k=h$,  but are not connected and do not share any neighbors.

In general, let $\mathcal{N}_{ik}^{wg}=\mathcal{N}_{ik}^g\cup ik$. Even under independent treatment assignment, if $\mathcal{N}_{ik}^{wg} \cap \mathcal{N}_{jk'}^{wg}\neq 0$ the joint treatment of the units $ik$ and $jk'$ will be dependent and $\pi_{ikjk'}(w,g;w',g')\neq\pi_{ik}(w,g)\pi_{jk'}(w',g')\neq 0$.
Note that if $\pi_{ik}(w,g)$ or $\pi_{jh}(w',g')=0$ $\Rightarrow \pi_{ikjh}(w,g;w',g')=0$, but not the reverse. Indeed, the joint probability of the two events ($\Wik=w,\Gik=g$) and ($W_{jh}=w',G_{jh}=g'$) might be zero if the network exposures $G_{ik}$ and $G_{jh}$ are defined on two subsets of units that coincide or include $jh$ and $ik$, respectively. For example, if the network exposure is defined as in Equation \ref{eq: threshold} with a threshold equal to 1 (i.e., having at least one treated neighbor) then, if unit $ik$ is treated, with $W_{ik}$ and belongs to the neighborhood $\mathcal{N}_{jh}$ of unit $jh$, the network exposure $G_{jh}$ cannot be 0. 
 
\section{Conditional Treatment and Spillover Effects and Horvitz-Thompson Estimator}
\label{sec:cace}
\subsection{Conditional Treatment and Spillover Effects} \label{subsec: cace}

Our ultimate goal is to \emph{detect} the regions of the covariate space exhibiting a high level of heterogeneity in the causal effects and \emph{estimate} the causal effects of interest in these heterogeneous regions. In this section, we will focus on the definition and estimation of conditional treatment and spillover effects and we will assume that the heterogeneous regions that we want to investigate have already been identified, either a priori according to subject-matter knowledge or thanks to data-driven methods. 

Let us denote with $\Pi$ a partition of the covariate space $\mathcal{X}$ into $M$ non-overlapping regions: $\Pi=\{\ell_{1}, \dots ,\ell_{M}\}$, where $\bigcup^{M}_{m=1}\ell_{m}=\mathcal{X}$, and with $\ell(\vx,\Pi): \mathcal{X} \rightarrow \Pi $ a function that maps each vector $\vx$ of the covariate space into a region. Let $N(\ell_m)$ be the size of each region $\ell_m$, with $m=1,\dots,M$, and let $V_k(\ell_m)$ be the subset of units belonging to region $\ell_m$ in cluster $k$, with $k=1,\ldots,K$. In the machine learning literature on CART, these non-overlapping regions are referred to as \emph{leaves}. For consistency, throughout we will use this terminology, regardless of whether the partition $\Pi$ has been a priori defined or is the result of a tree-based algorithm. In addition, to ease notation, we will drop the reference to the partition $\Pi$ from the mapping function $\ell(\cdot)$.

When units are organized in a network, it is worth noting that a partition $\Pi$ of the covariate space divides the sample units into sub-populations according to similarities in their characteristics, regardless of their network distance. Hence, two connected units might belong to different regions of the partition. However, in a homophilous network, where the probability of forming a link depends on the similarity in certain features and, hence, connected units are likely to share similar characteristics, a partition of the covariates' space is also likely to cluster connected units together (see Figure \ref{fig: partition}).

\begin{figure}
    \centering
    \includegraphics[width=0.6\textwidth]{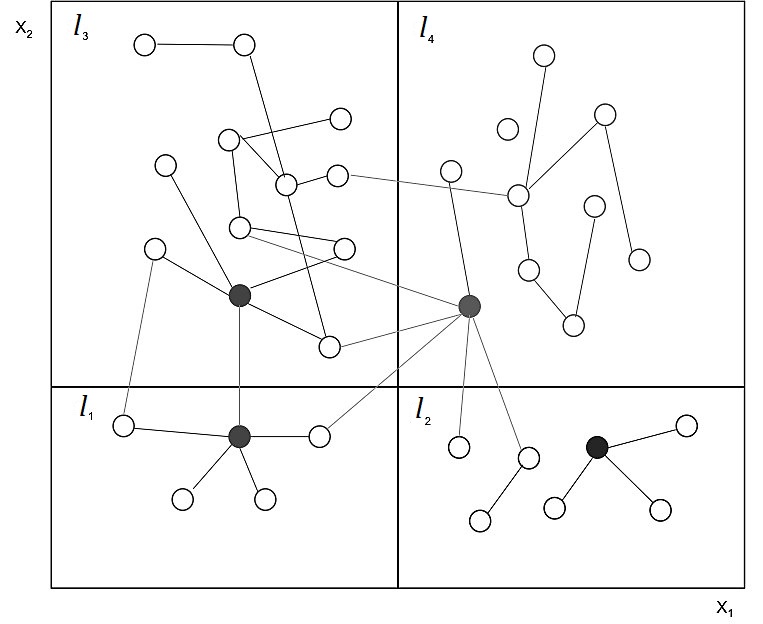}
    \caption{Partition of the covariate space with connected units.}
    \label{fig: partition}
\end{figure}

Given a partition $\Pi$, we now define conditional average potential outcomes under each individual treatment and network exposure condition $(w,g)\in \Delta$. For the subset of units with covariate vectors $\vx\in \mathcal{X}$ that are mapped to the same region by the function $\ell(\vx)$, we define the leaf-specific average potential outcome under treatment and exposure condition $(w,g)\in \Delta$ as follows:
\begin{equation}
    \mu_{(w,g)}(\ell(\vx))=\frac{1}{N(\ell(\vx))}\sum_{k=1}^K\sum_{i=1}^{n_k}\Yik(w,g)\mathds{1}(\vXik\in \ell(\vx)).
\end{equation}
Note that $\mu_{(w,g)}(\ell(\vx))$ is a sample average, that is, it is the average potential outcome for all units in sample $V$ with a covariate vector mapped to the same region $\ell(\vx)$. 

Leaf-specific conditional average causal effect (CACE) can be defined by comparing average potential outcomes under two different conditions: 
\begin{small}
\begin{align}
\tau_{(w,g;w',g')}(\ell(\vx))& =  \frac{1}{N(\ell(\vx))}\sum_{k=1}^K\sum_{i=1}^{n_k}\Yik(w,g)\mathds{1}(\vXik\in \ell(\vx)) -  \frac{1}{N(\ell(\vx))}\sum_{k=1}^K\sum_{i=1}^{n_k}\Yik(w',g')\mathds{1}(\vXik\in \ell(\vx)) \nonumber \\
    & =    \mu_{(w,g)}(\ell(\vx))- \mu_{(w',g')}(\ell(\vx)). 
\end{align}
\end{small}
We denote by $\mathcal{T}$ the set of possible contrasts we are interested in. For instance, if both the individual treatment and the network exposure are binary, then $\mathcal{T}\subseteq \{(1,0;0,0), (1,1,0,1),(0,1;0,0),(1,1,1,0),(1,1,0,0)\}$.  We define as \emph{leaf-specific treatment effects} causal contrast $\tau_{(w,g;w',g)}(\ell(\vx))$ that keep the network exposure fixed at a level $g$ while changing the individual treatment from $w'$ to $w$, that is when $g=g'$. These represent the causal effects of receiving the treatment while the treatment status of all other units is kept fixed or mapped to the same network exposure $g$. \review{For instance, in the weather insurance application $\tau_{(1,0;0,0)}(\ell(\vx))$ may represent, for farmers with similar characteristics (e.g., older than 50 years old and educated), the average effect on insurance uptake  of directly receiving the intensive information session at round 1 vs receiving the simple session or any type of session at round 2, while no friend has received the intensive session at round 1.}

On the contrary, we define as \emph{leaf-specific spillover effects} causal contrasts $\tau_{(w,g;w,g')}(\ell(\vx))$ that keep the individual treatment fixed at a level $w$ while changing the network exposure from $g'$ to $g$, that is, when $w=w'$. These spillover effects can be seen as causal effects of a change in the treatment status of other units such that the network exposure also changes, while the individual treatment status is kept fixed.
\review{For instance, in our empirical application $\tau_{(0,1;0,0)}(\ell(\vx))$ may represent, for farmers with similar characteristics (e.g., older than 50 years old and educated) and who have received the simple information session at round 1 or any type of session at round 2, the average spillover effect on insurance uptake of having at least one friend who has received the intensive session at round 1.}

It should be emphasized that $\mu_{(w,g)}(\ell(\vx))$ corresponds to a unit-level intervention setting the treatment and network exposure of each unit to specific values. The focus on these types of average potential outcomes, as opposed to the ones based on population-level hypothetical interventions as in \cite{hudgens2008toward}, is due to our purpose of investigating heterogeneous responses to the individual treatment and network status across units with different characteristics.
If one is interested in assessing the heterogeneity of the average response to the network exposure resulting from a hypothetical treatment allocation, our approach could be extended to marginalized causal effects as the ones in \cite{forastiere2016identification}.

\subsection{Estimator for leaf-specific CACE} \label{subsec: caceestim}

Here we develop a Horvitz-Thompson estimator for leaf-specific conditional average causal effects.
The derivation of the proposed estimator builds upon the estimator for average causal effects under network interference proposed by \cite{aronow2017estimating}. 

Following \cite{horvitz1952generalization} and \cite{aronow2017estimating}, a design-based estimator
for the \textit{leaf-specific average potential outcome} under individual treatment $w$ and network exposure $g$, $\mu_{(w,g)}(\ell(\vx))$, can be expressed as:
\begin{equation}
\widehat{\mu}_{(w,g)}(\ell(\vx))=\frac{1}{N(\ell(\vx))}\sum_{k=1}^K\sum_{i=1}^{n_k}\frac{Y_{ik}}{\pi_{ik}(w,g)}\mathds{1}(W_{ik}=w, G_{ik}=g, \mathbf{X}_{ik}\in \ell(\vx))    
\end{equation}
where $\pi_{ik}(w,g)$ denotes the probability of a given unit $ik$, that belongs to the leaf $\ell(\vx)$ (in the partition $\Pi$), to be exposed to the treatment condition $(w,g)$. 

The variance estimator of $\widehat{\mu}_{(w,g)}(\ell(\vx)$ can be expressed as:
\begin{small}
\begin{align}
\label{eq:var}
    \widehat{\Var}\Big(\widehat{\mu}_{(w,g)}(\ell(\vx))\Big)=&\frac{1}{N(\ell(\vx))^2}\sum_{k=1}^K\sum_{i=1}^{n_k} \I(\Wik=w,\Gik=g, \vXik \in \ell(\vx))[1-\pi_{ik}(w,g)]\biggl[\frac{\Yik}{\pi_{ik}(w,g)}\biggr]^2 \nonumber \\&+ \frac{1}{N(\ell(\vx))^2}\sum_{k=1}^K\sum_{i=1}^{n_k}\sum_{j\neq i} \I(\Wik=w,\Gik=g, \vXik \in \ell(\vx))\I(\Wjk=w,\Gjk=g, \mathbf{X}_{jk} \in \ell(\vx)) \nonumber \\ &\times\frac{\pi_{ikjk}(w,g)-\pi_{ik}(w,g)\pi_{jk}(w,g)}{\pi_{ikjk}(w,g)}\frac{\Yik}{\pi_{ik}(w,g)}\frac{\Yjk}{\pi_{jk}(w,g)}.
    \end{align}
\end{small}
This expression extends the variance estimator derived in \cite{aronow2017estimating} (Equation 7) to the case of conditional average potential outcomes and clustered network interference. In fact, the second term in \eqref{eq:var} includes the covariance between the individual treatment and network exposure of two units belonging to the same leaf $\ell(\vx)$. Under an experimental design with independent treatment allocation between clusters and under clustered interference such covariance between two units belonging to different clusters is zero and the second term should be restricted to units $j$ in the same cluster as $i$. In addition, the covariance between the joint treatment of two units is non-zero if the set of units defining the network exposure -- e.g., the whole cluster or the unit's neighborhood -- is shared between them or includes them. Formally, if $\mathcal{N}_{ik}^{wg} \cap \mathcal{N}_{jk'}^{wg}\neq 0$ the joint treatment of the units $ik$ and $jk'$ will be dependent, that is, $\pi_{ikjk'}(w,g)-\pi_{ik}(w,g)\pi_{jk'}(w,g)\neq 0$ . Hence, two units belonging to the same leaf are more likely to have intersecting sets $\mathcal{N}_{ik}^{wg}$ and $\mathcal{N}_{jk'}^{wg}$ (e.g., shared neighbors) if the sets are homogeneous, that is, units belonging to these sets share similar characteristics. Settings with homophilous networks are investigated in the appendix.

An estimator for the leaf-specific conditional average causal effect of the exposure condition $(w,g)$ compared to the configuration $(w',g')$ can be written as:
\begin{equation}
\widehat{\tau}_{(w,g;w',g')}(\ell(\vx))=\widehat{\mu}_{(w,g)}(\ell(\vx))-\widehat{\mu}_{(w',g')}(\ell(\vx)).
\end{equation}
The estimated variance of the estimator $\widehat{\tau}_{(w,g;w',g')}(\ell(\vx))$ can be decomposed as follows:
\begin{align}
\widehat{\mathbb{V}}\Big(\widehat{\tau}_{(w,g;w',g')}(\ell(\vx))\Big)= &   \widehat{\mathbb{V}}\Big(\widehat{\mu}_{(w,g)}(\ell(\vx))\Big) + \widehat{\mathbb{V}}\Big(\widehat{\mu}_{(w',g')}(\ell(\vx))\Big) \nonumber \\ &-2\Big[\widehat{\mathbb{C}}\Big(\widehat{\mu}_{(w,g)}(\ell(\vx)),\widehat{\mu}_{(w',g')}(\ell(\vx)) \Big) \Big]
\end{align}
with the covariance estimator taking the following expression for the case when $\pi_{ikjk}(w,g;w',g')>0 \,\, \forall i,j,k$:
\begin{footnotesize}
\begin{eqnarray}
\widehat{\mathbb{C}}\Big(\widehat{\mu}_{w,g}(\ell(\vx)),\widehat{\mu}_{w',g'}(\ell(\vx))\Big)&=& \frac{1}{N(\ell(\vx))^2}
\sum_{k=1}^K\sum_{i=1}^{n_k}\sum_{j \neq i}\frac{\I(\Wik=w,\Gik=g, \mathbf{X}_{ik} \in \ell(\vx))\I(\Wjk=w',\Gjk=g', \mathbf{X}_{jk} \in \ell(\vx))}{\pi_{ikjk}(w,g;w',g')} \nonumber\\
&\times&[\pi_{ikjk}(w,g;w',g')-\pi_{ik}(w,g)\pi_{jk}(w',g')]\frac{\Yik}{\pi_{ik}(w,g)}\frac{\Yjk}{\pi_{jk}(w',g')} \nonumber  \\
 &-& \frac{1}{N(\ell(\vx))^2} \sum_{k=1}^K\sum_{i=1}^{n_k}\sum_{j \neq i}\Big[\frac{\I(\Wik=w,\Gik=g, \mathbf{X}_{ik} \in \ell(\vx))\Yik^2}{2\pi_{ik}(w,g)} \nonumber \\ &+&\frac{\I(\Wik=w',\Gik=g', \mathbf{X}_{ik} \in \ell(\vx))\Yik^2}{2\pi_{ik}(w',g')}\Big]. 
\end{eqnarray}
\end{footnotesize}
Further details about the variance estimator of leaf-specific CACE can be found in Appendix \ref{appendix_varcace}.

\subsubsection{Properties of the Horvitz-Thompson Estimator}
\label{subsubsec: properties}
Here we will describe the properties of the Horvitz-Thompson estimator of leaf-specific causal effects. Asymptotic results will rely on a growth process that is commonly assumed with cluster data. In particular, we consider a sequence of nested samples $V$ of size $N$, where $V$ consists of $K$ separate clusters $V_k$ of size $n_k$, $k=1,\ldots,K$. We let the sample size $N\longrightarrow  \infty$ by letting the number of clusters go to infinity, i.e., $K\longrightarrow  \infty$,  while the cluster size $n_k$, $k=1,\ldots,K$ remains fixed.

\begin{proposition}[Unbiasedness]
\label{prop:unbiaseness} \review{Under Assumptions \ref{ass:cni}, \ref{ass:consistency}, \ref{ass: bing}, \ref{ass:ita}, and \ref{ass: positivity}, then}
\begin{equation}
    E\Big[\widehat{\mu}_{(w,g)}(\ell(\vx))\Big]=\mu_{w,g}(\ell(\vx)) \nonumber
\end{equation}
and
\begin{equation}
    E\Big[\widehat{\tau}_{(w,g, w'g')}(\ell(\vx))\Big]=\tau_{w,g;w',g'}(\ell(\vx)). \nonumber
\end{equation}
\begin{proof}
Proof in Appendix \ref{app:proofs}.
\end{proof}
\end{proposition} 
\noindent The unbiasedness of the estimator of leaf-specific CACE is conditional on the partition $\Pi$ and the function $\ell(\cdot)$. When building causal trees to assess the heterogeneity of causal effects, we will rely on this property to derive the splitting criterion and to estimate leaf-specific causal effects.\footnote{The unbiasedness of the estimator $\widehat{\tau}_{w,g, w',g'}(\ell(\vx)$ does not ensure the identification of subsets with the highest heterogeneity. The performance of the causal tree in identifying heterogeneous regions depends on the splitting criterion, the algorithm and the sample.}

\begin{proposition}[The variance estimator of $\widehat{\mu}_{(w,g)}$ is unbiased]
\review{Under Assumptions \ref{ass:cni}, \ref{ass:consistency}, \ref{ass: bing}, \ref{ass:ita}, and \ref{ass: positivity},}
if $\pi_{ikjk}(w,g)>0 \,\, \forall i,j,k$ then 
$$\E\Big[\widehat{\Var}\Big(\widehat{\mu}_{(w,g)}(\ell(\vx))\Big)\Big]= \Var\Big(\widehat{\mu}_{(w,g)}(\ell(\vx)\Big).$$
\end{proposition} 
\noindent The proof follows directly from the unbiasedness of the Horvitz-Thompson estimator. A conservative estimator for the case when $\pi_{ikjk}(w,g)=0$ for some units can be found in Appendix \ref{appendix_varcace}.
\begin{proposition}[The variance estimator of $\widehat{\tau}_{(w,g;w',g')}$ is conservative]
\label{prop: conservative}
\review{Under Assumptions \ref{ass:cni}, \ref{ass:consistency}, \ref{ass: bing}, \ref{ass:ita}, and \ref{ass: positivity}, then}
$$\E\Big[\widehat{\Var}\Big(\widehat{\tau}_{(w,g;w',g')}(\ell(\vx))\Big)\Big]\geq \Var\Big(\widehat{\tau}_{(w,g;w',g')}(\ell(\vx))\Big).$$
\end{proposition} 
\noindent A proof follows from \cite{aronow2017estimating}. 

\begin{proposition}[Consistency \review{of the estimator}]
\label{prop: consistency}
Consider the asymptotic regime where the number of clusters K go to infinity, i.e., $K \longrightarrow \infty$, while the cluster size remains bounded, i.e., $n_k \leq B(\ell(\vx))\leq B$ for some constant B. In addition, assume that $|Y_{ik}(w,g)|/\pi_{ik}(w,g)\leq C<1$, $\forall i,k, w,g$. Then, \review{under Assumptions \ref{ass:cni}, \ref{ass:consistency}, \ref{ass: bing}, \ref{ass:ita}, and \ref{ass: positivity},} as $K \longrightarrow \infty$
$$\widehat{\tau}_{(w,g;w',g')}(\ell(\vx)) \overset{p}{\longrightarrow} \tau_{(w,g;w',g')}(\ell(\vx)).$$
\end{proposition} 
\begin{proof}
See Appendix \ref{app:proofs}.
\end{proof}
\noindent Note that cluster network interference (Assumption \ref{ass:cni}) and independent treatment allocation between clusters ensure that the amount of dependence across units is limited. This limited independence is the condition required to ensure consistency \citep{aronow2017estimating}.\footnote{Note that for the variance of the estimator to go to zero as $N\longrightarrow \infty$ one must have that:
$$\sum_{ik}\sum_{jk'}  \I(\vXik \in \ell(\vx),\mathbf{X}_{jk} \in \ell(\vx))\I\bigl(\pi_{ikjk'}(w,g)-\pi_{ik}(w,g)\pi_{jk'}(w,g)\neq 0\bigl)=
o(N(\ell(\vx))^2).$$ 
\noindent This is guaranteed given that the joint treatment is independent between units in different clusters, that is, $\pi_{ikjk'}(w,g)=\pi_{ik}(w,g)\pi_{jk'}(w,g), \forall k\neq k'$, and given that the cluster size is bounded.}
 
\begin{proposition}[Asymptotic Normality]
\label{prop:asy}
\review{Given Assumptions \ref{ass:cni}, \ref{ass:consistency}, \ref{ass: bing}, \ref{ass:ita}, and \ref{ass: positivity},} then:
$$\sqrt{N(\ell(\vx))}\Big(\widehat{\tau}_{(w,g;w',g')}(\ell(\vx)) -\tau_{(w,g;w',g')}(\ell(\vx))\Big) \overset{d}{\longrightarrow} N\bigg(0, \Var\Big(\widehat{\tau}_{(w,g;w',g')}(\ell(\vx))\Big)\bigg).$$
\end{proposition}
An independent treatment allocation between clusters (Assumption \ref{ass:ita}) and the clustered network interference (Assumption \ref{ass:cni}) ensures the limited dependence condition required in \cite{aronow2017estimating}. This condition allows us to rely on a central limit theorem derived via Stein's method \citep{chenshao2004} to achieve the asymptotic normality of the estimator. 
 The variance estimators will depend on the size of the sample belonging to leaf $\ell(\vx)$ in each cluster, i.e., $n_k(\ell(\vx))\leq B(\ell(\vx))\leq B, k=1=\dots,K$, and the maximum conditional degree:
 \begin{equation*}
     D(\ell)=\text{max}_{ik\in V: \vXik\in \ell(\vx)}\mathcal{N}^g_{ik}(\ell(\vx)) \:\:\:\: \text{where} \:\:\:\: \mathcal{N}^g_{ik}(\ell)=\mathcal{N}^g_{ik} \cap V_k(\ell(\vx)).
 \end{equation*}
 Given that these quantities are bounded, we can show that $\Var\Big(\Var\Big(\widehat{\tau}_{(w,g;w',g')}(\ell(\vx))\Big)\Big)=O(1/N(\ell(\vx)))$, that is, the rate of convergence will be $1/\sqrt{N(\ell(\vx))}$, with $N(\ell(\vx))\leq KB(\ell(\vx))$ (the proof follows the one in \cite{aronow2017estimating}).

It is worth noting that, thanks to the asymptotic growth assumed here, Proposition \ref{prop:asy} would still hold if the covariates' vector $\vX$ would include network covariates defined both at the cluster-level or at the individual-level. This result would allow us to investigate the heterogeneity of causal effects with respect to network characteristics, including variables defining a cluster network structure or the structure of the network neighborhood around a node. 


\section{Network Causal Trees for Heterogeneous Causal Effects under Clustered Network Interference}  
\label{sec: ctrees}

In the previous section, we have introduced and developed an estimator for causal effects conditional on sub-populations of units defined by a partition $\Pi$ of the covariate space $\mathcal{X}$. Here we develop a data-driven machine learning algorithm to identify the partition $\Pi$ aimed at investigating the heterogeneity in the effects of interest. 

Our proposed algorithm, named \emph{Network Causal Tree} (NCT), builds upon the Causal Tree (CT) algorithm introduced by \cite{athey2016recursive}, which in turn finds its roots in the Classification and Regression Tree (CART) algorithm \citep{friedman1984classification}.
CART is a widely used nonparametric method to partition the feature space. It relies on a tree-based algorithm that recursively splits the sample. In particular, trees are constructed by recursively partitioning the observations from the \textit{root} (that contains all the observations in the learning samples) into two \textit{child nodes}. This procedure is repeated until the tree reaches the final nodes which are called \textit{leaves}. Because each node is always split into two sub-nodes, these trees are called \textit{binary trees}.

Binary trees are called \textit{regression trees} when the outcome is a continuous variable, while they are called \textit{classification trees} when the outcome is either a discrete or a binary variable. The aim of the tree construction is to identify heterogeneities in the relationship between the observed outcome and the features to best predict the outcome variable. Therefore, splits are made with the aim of minimizing the prediction error. With this aim, different splitting criteria could be specified. For additional details on CART, we refer to the seminal paper by \cite{friedman1984classification}. Figure \ref{fig:CART} illustrates an example of binary partitioning in a simple case with just two predictors $x_1 \in [0,1]$ and $x_2 \in [0,1]$.

\begin{figure}
    \centering
    \begin{subfigure}[b]{0.48\textwidth}
		\centering
		\begin{tikzpicture}[level distance=80pt, sibling distance=50pt, edge from parent path={(\tikzparentnode) -- (\tikzchildnode)}]
		\tikzset{every tree node/.style={align=center}}
		\Tree [.\node[draw]{$x_1<0.6$}; \edge node[auto=right,pos=.6]{No}; \node[circle,draw]{$l_1$}; \edge node[auto=left,pos=.6]{Yes};[.\node[draw]{$x_2 > 0.2$}; \edge node[auto=right,pos=.6]{No}; \node[circle,draw]{$l_2$}; \edge node[auto=left,pos=.6]{Yes}; \node[circle,draw]{$l_3$};  ]]
		\end{tikzpicture}
		\label{fig:example_tree}
    \end{subfigure}
    \begin{subfigure}[b]{0.48\textwidth}
    \includegraphics[width=\textwidth]{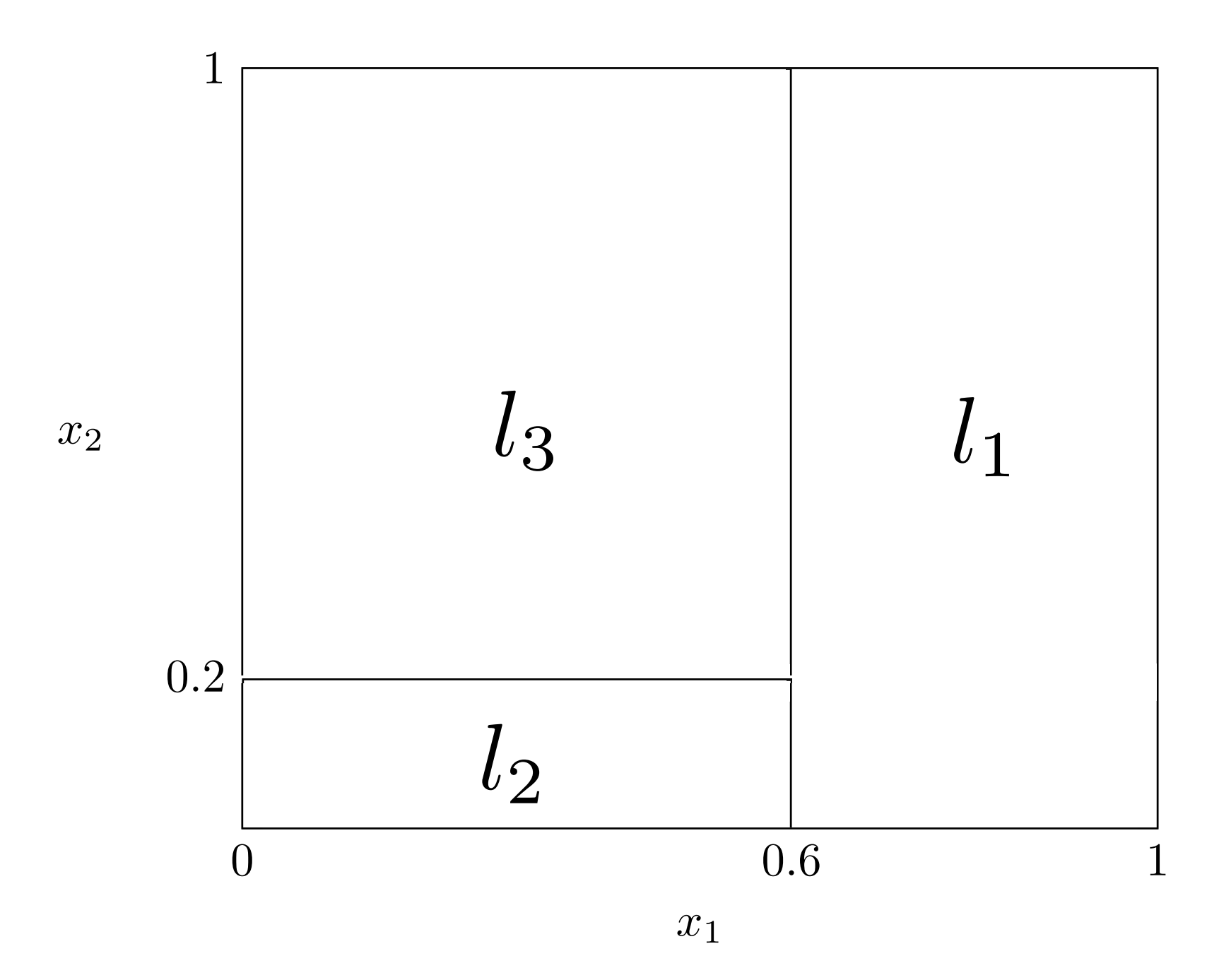}
    \end{subfigure}
    \caption{{(Left) An example of a binary tree. The internal nodes are labeled by their splitting rules and the terminal nodes are labeled with the corresponding parameters $l_i$.\\ (Right) The corresponding partition of the sample space.}}
    \label{fig:CART}
\end{figure}

Building on CART, \cite{athey2016recursive} developed a causal decision tree algorithm with the aim of detecting the heterogeneity of causal effects. In particular, they modified the splitting function to minimize the estimation error of conditional effects. Moreover, \cite{athey2016recursive} introduced honest inference by using a sub-sample to build the tree (training or discovery sample) and a separate sub-sample to perform inference (estimation sample). This sample-splitting approach is transparent and efficient even in high-dimensional settings.

Our proposed NCT differs from the standard causal tree algorithm in two critical aspects: (i) it estimates heterogeneous causal effects -- both treatment and spillover effects -- in the presence of clustered network interference,  and (ii) it possibly models heterogeneity with respect to more than one effect at the same time through a composite splitting criterion. In this Section, we describe and motivate the splitting criteria for our NCT algorithm (Subsection \ref{subsec: split}) and its detailed structure (Subsection \ref{subsec: nctalg}).

\subsection{Splitting Criteria} \label{subsec: split}

The NCT algorithm is built to detect and estimate heterogeneous treatment and spillover effects, in the presence of clustered network interference. Moreover, NCT is able to discover the heterogeneity with respect to more than one estimand. Here, we present the three criteria that rule the splitting procedure of NCT, targeted to single effects or multiple effects.

\subsubsection{Single splitting criteria}
Let $\mathds{P}$ be the space of partitions. Given a causal effect $\tau_{(w,g;w',g')}$, we can use recursive splitting to look for the best partition $\Pi \in \mathds{P}$ with respect to a splitting criterion $Q_{(w,g;w',g')}(\Pi)$. Formally:

\begin{equation}
     \Pi=\text{argmax}_{\Pi \in \mathds{P}}   \, Q_{(w,g;w',g')}(\Pi).
\end{equation}
Given our goal of describing the relationship between the causal effect and the covariate space and detecting subsets that exhibit a high level of heterogeneity, we can define a splitting criterion that 
maximizes accuracy in the prediction of conditional effects $\tau_{(w,g;w',g')}(\vX_{ik})$ in the whole sample V. This translates into the minimization of the expected value of the mean square error (MSE):
\begin{equation}
\label{eq:splitfun}
     Q_{(w,g;w',g')}(\Pi)= - EMSE\Big(\widehat{\tau}_{(w,g;w',g')}(\ell(\vx), \Pi)\Big) 
       =- \E\Big[\Big(\tau_{(w,g;w',g')}(\vx)-\widehat{\tau}_{(w,g;w',g')}(\ell(\vx, \Pi)\Big)^2\Big]
\end{equation}
where the expected value is taken over the sampling distribution. When this splitting criterion is used to select the partition $\Pi$, we maximize the function in \eqref{eq:splitfun} evaluated in the sample used to build the tree, i.e., the \emph{training set}. In this case, in the machine learning literature, the objective function is referred to as the \emph{in-sample splitting function} and we denote this by $Q^{in}_{(w,g;w',g')}(\Pi)$.

As opposed to the EMSE of the observed outcome prediction, the true causal effect $\tau_{w,g;w',g'}(\vx)$ is unknown. However, we can use the training data to estimate the EMSE for the in-sample spitting rule. 
Thanks to the unbiasedness of the estimator $\widehat{\tau}_{(w,g;w',g')}(\ell(\vx), \Pi)$ with respect to the population causal effect $E[\tau_{(w,g;w',g')}(\vXik)|\vXik\in \ell(\vx)]$, following \cite{athey2016recursive} we can estimate the EMSE as follows:
\begin{equation}
\label{eq:splitfun_in}
     Q^{in}_{(w,g;w',g')}(\Pi)= - \widehat{EMSE}\Big(\widehat{\tau}_{(w,g;w',g')}(\ell(\vx), \Pi)\Big)=\frac{1}{N^{tr}}\sum_{k\in  \mathcal{K}^{tr}}\sum_{i=1}^{n_k} \big(\widehat{\tau}_{(w,g;w',g')}(\ell(\vXik, \Pi))\big)^2
\end{equation}
where $\mathcal{K}^{tr}$ is the subset of clusters belonging to the training set and $N^{tr}$ is the sample size.\footnote{In standard CART the training set is a subset of the whole sample together with the testing set, which is used to evaluate the objective function in order to choose the best partition selected in the training set that maximizes out-of-sample prediction accuracy.} Therefore, the maximization of this splitting function results in the maximization of the heterogeneity across leaves. In fact, if two sub-populations $\ell_1$ and $\ell_2$ have a different causal effect $\tau_{(w,g;w',g')}$, i.e., $\tau_{(w,g;w',g')}(\ell_1)\neq \tau_{(w,g;w',g')}(\ell_2) $, a partition $\Pi$ that splits them would yield a higher $Q_{(w,g;w',g')}(\Pi)$ than the  partition $\Pi^c$ that combines the two sub-populations into one leaf $\ell_{1+2}$ (a simple proof can be found in Appendix \ref{app:proofs}).

To avoid using the same information for selecting the partition and for the estimation, \cite{athey2016recursive} propose to estimate the effects in a separate sample from the one used to build the tree. They call this an `honest' causal tree. We denote by $V^{est}$ the estimation set and by $V^{tr}$ the training set (or discovery set) that is used to build the tree (by evaluating the splitting function). The training set and the estimation set are here obtained by taking two random subsets $\mathcal{K}^{tr}$ and $\mathcal{K}^{est}$ of the clusters in $\mathcal{K}$. Hence, $V^{est}=\bigcup_{k\in \mathcal{K}^{est}} V_k$ and $V^{tr}=\bigcup_{k\in \mathcal{K}^{tr}} V_k=V/V^{est}$.  This random split of the sample avoids any dependencies between training and estimation sub-samples. In this `honest' version, the splitting function can be estimated as follows:
\begin{equation}
\label{eq:honest}
   Q_{(w,g;w',g')}^{in,H}(\Pi)=\frac{1}{N^{tr}}\sum_{k\in  \mathcal{K}^{tr}}\sum_{i=1}^{n_k}\big(\hat{\tau}_{(w,g;w',g')}(\ell(\vXik, \Pi))\big)^2 
   -\Big(\frac{1}{N^{tr}}+\frac{1}{N^{est}}\Big)
   \sum_{\ell\in \Pi}\widehat{\Var}\Big(\widehat{\tau}_{(w,g;w',g')}(\ell; \Pi)\Big)
\end{equation}
where $N^{est}=|V^{est}|$.  The proof follows from \cite{athey2016recursive}. In \eqref{eq:honest} we can see that the splitting function is such that splits will be chosen so to maximize the heterogeneity across leaves as well as to minimize the average variance in the estimated effects. The idea is to identify the most heterogeneous partitions while introducing a penalization term that corrects the objective function to minimize the leaf-specific variation in the estimated effect. This penalization term has also the effect of reducing the depth of the tree because leaves with a small number of observations $N(\ell)$ will exhibit a higher variance. Hence, the final depth of the tree will depend on the sample size as well as on the randomization design and the exposure mapping function. In addition to this penalization, we will also add a stopping rule based on a minimum number of observations for each condition $(w,g)$ and $(w',g')$ that we are comparing. This stopping rule is required to avoid having leaves where the effect $\tau_{(w,g;w',g')}$ cannot be estimated because there are no observations with observed treatment $(W_{ik},G_{ik})$ equal to the values $(w,g)$ or $(w',g')$.

\subsubsection{Composite splitting criterion}
We now introduce a composite splitting rule targeted to multiple causal estimands.
When interference is in place targeting strategies might involve both treatment and spillover effects. For example, in settings with limited resources, the treatment should be provided to those who would benefit from it, i.e., with a non-zero treatment effect,  whereas we could save resources by not giving the treatment to those who would benefit from other people being treated, namely units with high spillover effects. For instance, this is the case in marketing interventions where we can provide advertisements only to those who would be affected and who are less likely to get the information from someone else.
Another interesting example can be found in the potential challenges of the COVID-19 vaccine distribution. Those at high risk of getting infected, even if those in close contact were immune, should be targeted. These would be individuals who are often in crowded spaces, either in the workplace or in a social setting, and would likely get infected in these environments. Therefore, these individuals are characterized by a high treatment effect even with all close contacts being treated. On the contrary, those who are in contact with a low number of people and can greatly gain from having one of these contacts vaccinated could be left without the vaccine, at least in the early stages of the distribution. 

For these kinds of targeting strategies involving more than one causal effect, we must partition the population into sub-groups that show a high level of heterogeneity in all estimands of interest. Building a separate tree for each causal effect would provide us with different partitions that cannot be used for the design of multi-effect strategies. Therefore, we propose a composite splitting function that would result in a tree that maximizes heterogeneity in all the causal estimands of interest. This composite objective function is a weighted average of the effect-specific splitting functions:
\begin{equation}
\label{qcom}
     Q_{\mathcal{T}}(\Pi)=\sum_{(w,g;w',g') \in \mathcal{T}} \gamma_{(w,g;w',g')} Q_{(w,g;w',g')}(\Pi) \quad\text{with}\quad \gamma_{(w,g;w',g')}=\frac{\omega_{(w,g;w',g')}}{\big(\hat{\tau}^{(w,g;w',g')}\big)^2 }
\end{equation}
where $\omega_{(w,g;w',g')}\in [0,1]$ is a customized weight for each estimand and
 $\hat{\tau}_{(w,g;w',g')}$ is the estimated effect in the whole sample.
 Each effect $\tau_{(w,g;w',g')}$, where $(w,g;w',g') \in \mathcal{T}$, contributes to the global objective function according to a specific weight $\gamma_{(w,g;w',g')}$. $\gamma_{(w,g;w',g')}$ is proportional to a customized weight $\omega_{(w,g;w',g')}$, which is set by the researcher according to the extent to which the estimand $\tau_{(w,g;w',g')}$ is of interest, and is normalized by the estimated effect in the whole sample to rule out any dependence on the magnitude of the effect. The composite criterion requires that at least two of the four weights are strictly greater than zero. A similar composite objective function can be derived from the splitting functions for the `honest' causal trees:
 \begin{equation}
\label{qcomb_H}
     Q^H_{\mathcal{T}}(\Pi)=\sum_{(w,g;w',g') \in \mathcal{T}} \gamma_{(w,g;w',g')} Q^H_{(w,g;w',g')}(\Pi) \quad\text{with}\quad \gamma_{(w,g;w',g')}=\frac{\omega_{(w,g;w',g')}}{\big(\hat{\tau}^{(w,g;w',g')}\big)^2 }.
\end{equation}

\subsection{Network Causal Tree (NCT) Algorithm} \label{subsec: nctalg}

Compared with the standard HCT algorithm, the main novelties of NCT are the introduction of interference and the possibility of including more than one effect. Specifically, the extent to which each effect $\tau_{(w,g;w',g')}$, with $(w,g;w',g') \in \mathcal{T}$, contributes to the determination of the tree is specified by the weight $w(w,g;w',g')$. Here we describe the key steps of the NCT algorithm, including the recursive partitioning based on the splitting functions and the stopping rules.

\subsubsection{Key steps of the NCT algorithm}
The proposed algorithm takes mainly six elements as inputs: 
\begin{enumerate}
\item the sample $V$, which collects for each unit $ik$ the individual treatment assignment status $\Wik$, the observed outcome $Y_{ik}$ and a vector of characteristics $\vXik$;
\item the network information, which is fully described by the global adjacency matrix $\boldsymbol{A}$, including the cluster-specific blocks $\boldsymbol{A}_k$;
\item the specification of the exposure mapping function $g(\cdot)$ which, together with the adjacency matrix and possibly covariate matrix, will be translated in the computation of the observed network exposure $G_{ik}$ for each unit;
\item the experimental design which will determine the computation of the probabilities $\pi_{ik}(w,g)$ and $\pi_{ik}(w,g, w',g')$; 
\item the weight $\omega_{w,g,w'g'} $for each causal effect;
\item the specification of the two parameters: \textsl{maximum depth}, that is the maximum depth of the tree, and the \textsl{minimum size}, that is the minimum number of units falling in each exposure condition $(w,g)$ in each leaf. 
\end{enumerate}
After some preliminary steps, the algorithm consists of two main steps. The first step is focused on the selection of the partition, i.e. the tree, while the second step is concerned with the estimation of causal effects and returns point estimates and standard errors of the conditional average causal effects, for all the comparisons of interest and within each leaf of the detected partition. We report below the key steps of the NCT algorithm:

		\begin{enumerate} 
		\setcounter{enumi}{-1}
		\item \textbf{Step 0} (Preliminaries): In a preliminary stage the algorithm computes the quantities and tools that will be used in the subsequent steps. In particular:
	\begin{enumerate}
	    \item 
	    Given the adjacency blocks $\boldsymbol{A}_k$ and potentially the covariate matrix $\boldsymbol{X}$, for each unit the \emph{network exposure} variable $G_{ik}$ is computed according to the rule expressed in Assumption \ref{ass: bing}.
		\item The joint exposure probabilities $\pi_{ik}(w,g)$ and $\pi_{ikjk'}(w,g;w',g')$ are computed (as in Subsection \ref{subsec: bernoulli}) or estimated (as in \cite{aronow2017estimating}).
		\item Finally, the algorithm randomly splits the clusters between the training set $V^{est}$ and the estimation set $V^{est}$.\footnote{Following \cite{athey2016recursive} we suggest assigning half of the clusters to the discovery sample and another half to the estimation sample.} 
  	\end{enumerate}
	 \item \textbf{Step 1} (Tree Discovery): the first step of the algorithm sprouts the tree structure of the NCT, that is, it detects the relevant heterogeneous partitions. Note that this step is performed over the discovery set only. In particular, the NCT algorithm works with the clusters belonging to the set $\mathcal{K}^{tr}$ and builds the tree using binary recursive partitioning.
		\begin{enumerate} 
			\item Recursive Partitioning. The algorithm \emph{grows a tree} by maximizing the in-sample splitting criterion at each binary split. At iteration $r-1$ the partition can be represented as follows: $$\Pi^{r-1}=\{\vx\in \mathcal{X}: \bigcap_{m=1}^{r-1} \,x^m\in \mathcal{A}^{h_m}_m\}_{\mathbf{h}\in \{L,R\}^{r-1}}$$ where $x^m \in \{x_{p}\}_{p=1,\dots, P}$ is the feature that was split at iteration $m$ and $\mathcal{A}^{L}_m=\{x^m\leq c_m\}$, $\mathcal{A}^{L}_m=\{x^m\> c_m\}$ for some cutoff point $c_m$. The variable $x^m$ split at iteration $m$ together with the cutoff point $c_m$ compose a \textit{node} of the tree.
At iteration $r$, the partition will be complemented with a split of a variable $x^r \in \{x_{p}\}_{p=1,\dots, P}/\{x^m\}_{m=1\dots,r-1}$ at some cutoff point $c_r$: 
$$\Pi^r=\{\vx\in \mathcal{X}: \bigcap_{m=1}^{r-1} \,x^m\in \mathcal{A}^{h_m}_m \bigcap x^r\in \mathcal{A}^{h_r}_r \}_{\mathbf{h}\in \{0,1\}^r}.$$
Among all the candidate splits $x^r$ and $c_r$, the algorithm will choose the one that maximizes the \textit{in-sample} splitting function in \eqref{eq:splitfun_in} or \eqref{eq:honest}.
\item Stopping Rule. The recursive partitioning stops when at least one \emph{stopping condition} is met: (i) the NCT has reached the specified \textsl{maximum depth}; (ii) the current split $r$ generates at least one leaf $\ell$ where the set of units $N(\ell)_{(w,g)}=\{ik \in V^{tr}: \vXik\in \ell, W_{ik}=w, G_{ik}=g\}$ with a number of observations $|N(\ell)_{(w,g)}|$ lower than the specified \textsl{minimum size}, for at least one exposure condition $(w,g)$.

		\end{enumerate}
This step generates a \emph{network causal tree} which corresponds to a partition $\Pi$ of the feature space $\mathcal{X}$ into $M$ leaves: $\Pi=\{\ell_{1}, \dots ,\ell_{M}\}$, with $\bigcup^{M}_{m:1}\ell_{m}=\mathcal{X}$ and $l(\vx,\Pi): \mathcal{X} \rightarrow \Pi $.
	\item \textbf{Step 2} (Estimation): the second step of the algorithm takes as input the Network Causal Tree $\Pi$ built in Step 1 and computes all the point estimates, the standard errors, and the confidence intervals of the leaf-specific causal effects of interest in all its nodes $\ell(\vx,\Pi)$. This is done using the Horvitz-Thompson estimator in Section \ref{sec:cace}. In the `honest' version, at this stage, the NCT algorithm works with the clusters belonging to the set $\mathcal{K}^{est}$.
      \end{enumerate}
      
	\begin{algorithm}[h]
		\caption{Overview of the NCT algorithm \label{alg:dis}}
		
		\begin{itemize}
		    \item \textbf{Inputs}: i) observed data $\{\Wik, Y_{ik}, \vXik \}_{ik\in V}$; ii) global adjacency matrix $\boldsymbol{A}$, which comprises the cluster-specific blocks $\boldsymbol{A}_k$; iii) experimental Design;  iv) vector of weights $\omega(w,g;w',g')$, where $(w,g;w',g') \in \mathcal{T}$; v) tree parameters: \textsl{maximum depth} and \emph{minimum size}.

            \item \textbf{Outputs}: (1) a partition $\Pi$ if the covariate space, and (2) point estimates,  standard errors, and confidence intervals of the conditional average causal effects:
	
		\begin{enumerate} 
		    \item Step 0 (Preliminaries): compute $G_{ik}$ and both the marginal and joint exposure probabilities $\pi_{ik}(w,g)$ and $\pi_{ikjk'}(w,g;w',g')$. Then, randomly assign clusters to discovery and estimation samples.
		
			\item Step 1 (Tree Discovery): using the discovery sample, build a tree according to the in-sample splitting criterion and stop when either the tree has reached its maximum depth or any additional split would generate leaves, that are not sufficiently representative of the four exposure conditions.
			
			\item Step 2 (Estimation): use the Horvitz-Thompson estimator on the estimation sample to estimate the leaf-specific CACE and their standard errors in each leaf.
			\end{enumerate}

		\end{itemize}

	\end{algorithm}

\section{Simulation Study} \label{sec:simulations}

Our algorithm provides an interpretable method to detect and estimate heterogeneous effects in the presence of clustered network interference. In this section we evaluate, through a set of simulations, the performance of the proposed algorithm with respect to both discovery and estimation. In particular, we investigate its ability to correctly identify the actual heterogeneous sub-populations, comparing the use of single or composite splitting functions, and we assess the performance of the Horvitz-Thompson estimator for leaf-specific treatment and spillover effects.  While the latter performance assessment is quite standard in the literature, the former is critical for the development of interpretable algorithms for heterogeneous causal effects  \citep{bargagli2019causal,lee2020causal}.

We evaluate the performance of the algorithm and the estimator in settings that differ with respect to three main factors: (1) the structure of the heterogeneity, (2) the extent of the effect heterogeneity, and (3) the number of clusters. In Appendix \ref{appendix:monte_carlo}, we also consider \review{three} additional factors: (4) the correlation structure in the covariate matrix, (5) the presence of homophily in the network structure, \review{and (6) a mixture of continuous and discrete covariates}. Regarding the structure of the heterogeneity, we are particularly interested in settings where the structure of the causal tree representing heterogeneity is different for each causal effect. In particular, causal trees differ if they have different nodes corresponding to the split of a feature, that is, if covariates driving the heterogeneity are different, or if they have different terminal leaves where the causal effect is heterogeneous, i.e., non-zero. We call \emph{causal rules} these heterogeneous terminal leaves.

For each simulation scenario, we simulated $M=500$ samples and applied our NCT algorithm to detect sub-populations with heterogeneous causal effects (or causal rules) and to estimate our causal effects of interest. To evaluate the performance of our composite splitting function under different settings splits rely on either effect-specific splitting criteria or on the composite function.

All simulations are performed under Bernoulli trials, that is treatment is randomly assigned independently to each unit with a fixed probability $\pi_{ik}^W=\alpha$. In our simulation study, we also assume that interference only takes place at the neighborhood level and we choose the following definition of network exposure: 
\begin{equation}
\label{eq:G}
    G_{ik}=\mathds{1}\bigg(\bigl(\sum_{j\in \Nik} W_{jk}\bigl) \geq 1 \bigg),
\end{equation}
that is, the network exposure of unit $ik$ is 1 if at least one neighbor is treated. A binary network exposure together with a binary individual treatment results in a joint treatment with four categories---i.e., $(\Wik, \Gik) \; \in \{ (w,g)=(0,0),(1,0),(0,1),(1,1)\} $. The binary definition of network exposure is chosen to allow the growth of deeper trees. Given that the minimum size requirement stops the algorithm when the number of units in a child leaf is not enough to estimate a conditional causal effect, a joint treatment with four categories ensures that this stopping condition is unlikely to be met during the first few splits. 

In addition, the assumption of neighborhood interference allows the computation of the marginal and joint probabilities without the need for intensive estimation procedures. In fact, the approximate algorithm for estimating the marginal and joint probabilities proposed by \cite{aronow2017estimating} is computationally demanding and could not be incorporated into our simulation study. However, in the case of Bernoulli trial and network exposure defined as in \eqref{eq:G}, the probability $\pi_{ik}(w,g)$ can be computed using the formula in \eqref{eq:ber_pi} while the joint probability $\pi_{ikjk'}(w,g, w',g')$ is simply the product of $\pi_{ik}(w,g)$ and $\pi_{jk'}(w',g')$ if the two units $ik$ and $jk'$ are independent. On the contrary, if $\mathcal{N}_{ik}^{wg}$ and $ \mathcal{N}_{jk'}^{wg}\neq 0$ overlap the joint treatment of the units $ik$ and $jk'$ will be dependent, that is, $\pi_{ikjk'}(w,g;w',g')\neq\pi_{ik}(w,g)\pi_{jk'}(w',g')$. In this case, the joint probability can still be readily computed using combinatorics formulas on two overlapping sets.

\subsection{Data generating process}

For each simulation $m=1\dots, M$ we generated $K$ clusters and simulated, within each cluster, Erd\H{o}s-R\'{e}nyi random graphs \citep{erdos1959random} with $n_k=100$ nodes and a fixed probability (0.01) to observe a link. Given the definition of the network exposure, we removed isolated nodes from the analysis to make the Assumption \ref{ass: positivity} hold.
$W_{ik}$ and any of the 10 covariates $X_{ip}$ were sampled from independent Bernoulli distributions with probability 0.5: 
$W_i \sim Ber(0.5)$ and $X_{ip} \sim Ber(0.5)$. 

In the simulation study, we focus on two main effects: (i) the pure treatment effect $\tau_{(1,0;0,0)}$, and (ii) the pure spillover effect $\tau_{(0,1;0,0)}$. To ease notation, we denote by $\tau$ the treatment effect and by $\delta$ the spillover effect. After setting the value of these two conditional effects for each unit with covariates $\vXik$ (depending on the simulation scenarios), we generated the four different potential outcomes: $Y_{ik}(0,0) \sim \mathcal{N}(0,1)$; $Y_{ik}(1,1) \sim \mathcal{N}(0,1)$; $Y_{ik}(1,0) = Y_{ik}(0,0)+ \tau(\vXik)$; and $Y_{ik}(0,1) = Y_{ik}(0,0)+ \delta(\vXik)$. Finally, the observed outcome is given by: 
\begin{equation*}
    Y_{ik} = \sum_{w=0}^1\sum_{g=0}^1 \mathds{1}(W_{ik}=w, G_{ik}=g) Y_{ik}(w,g). 
\end{equation*}

We now detail how we varied the three factors (1), (2), and (3). We simulated two different scenarios with respect to the heterogeneity structure (1). In the first scenario, we have: 
\begin{small}
\begin{equation*}
\hspace{-0.8cm}
    \tau(\vXik)=
    \begin{cases}
    \,\,h \quad \text{if}\quad \vXik\in \ell_{1}=\{X_{ik1}=0,X_{ik2}=0\};\\
    \!\!-h \quad \text{if}\quad \vXik\in \ell_{2}=\{X_{ik1}=1,X_{ik2}=1\};\\
    0 \quad \text{otherwise};
    \end{cases}
\quad
    \delta(\vXik)=
    \begin{cases}
    \,\,h \quad \text{if}\quad \vXik\in \ell_{1}=\{X_{ik1}=0,X_{ik2}=0\};\\
    \!\!-h \quad \text{if}\quad \vXik\in \ell_{2}=\{X_{ik1}=1,X_{ik2}=1\};\\
    0 \quad \text{otherwise}.
    \end{cases}
\end{equation*}
\end{small}
Hence, in this scenario, the heterogeneity driving variables (HDV)---i.e.,  $X_{i1}$ and $X_{i2}$---are the same for both the treatment effect $\tau$ and the spillover effect $\delta$ and the two causal rules overlap. In the second scenario, we introduce a change in the drivers of the heterogeneity in the following way:
\begin{small}
\begin{equation*}
\hspace{-0.8cm}
    \tau(\vXik)=
    \begin{cases}
    \,\,h \quad \text{if}\quad \vXik\in \ell_{\tau1}=\{X_{ik1}=0,X_{ik2}=0\};\\
    3h \quad \text{if}\quad \vXik\in \ell_{\tau2}=\{X_{ik1}=0,X_{ik2}=1\};\\
    0 \quad \text{otherwise};
    \end{cases}
\quad
    \delta(\vXik)=
    \begin{cases}
    \,\,h \quad \text{if}\quad \vXik\in \ell_{\delta1}=\{X_{ik1}=1,X_{ik3}=0\};\\
    3h \quad \text{if}\quad \vXik\in \ell_{\delta2}=\{X_{ik1}=1,X_{ik3}=1\};\\
    0 \quad \text{otherwise}.
    \end{cases}
\end{equation*}
\end{small}
Thus, in the second scenario, the heterogeneity drivers are different for the two causal effects. Specifically,  we have:  $X_{i1}$ and $X_{i2}$ for the treatment effects, and  $X_{i1}$ and $X_{i3}$ for the spillover effects. In addition, we have two causal rules for the treatment effect, namely $\{X_{ik1}=0,X_{ik2}=0\}$ and $\{X_{ik1}=0,X_{ik2}=1\}$ and two different causal rules for the spillover effect, namely $\{X_{ik1}=1,X_{ik3}=0\}$ and $\{X_{ik1}=1,X_{ik3}=1\}$. For each structural scenario, we varied the effect size: $h=(r+0.1)_{r=1}^{10}$. Figure \ref{fig: simultree} graphically represents the two simulation scenarios.
\begin{figure}
  \centering
  \begin{subfigure}{.45\textwidth}
  		\centering
		\begin{tikzpicture}[level distance=80pt, sibling distance=10pt, edge from parent path={(\tikzparentnode) -- (\tikzchildnode)}]
		\tikzset{every tree node/.style={align=center}}
		\Tree [.\node[draw]{$x_1=0$}; \edge node[auto=right,pos=.6]{Yes}; [.
		\node[draw]{$x_2 = 0$}; \edge node[auto=right,pos=.6]{Yes}; \node[circle,draw]{$\tau=h$\\ $\delta=h$}; \edge node[auto=left,pos=.6]{No}; \node[circle,draw]{$\tau=0$\\ $\delta=0$};]
		\edge node[auto=left,pos=.6]{No};[.\node[draw]{$x_2 = 0 $}; \edge node[auto=right,pos=.6]{Yes}; \node[circle,draw]{$\tau=0$\\ $\delta=0$}; \edge node[auto=left,pos=.6]{No}; \node[circle,draw]{\begin{footnotesize}
		$\tau=-h$
		\end{footnotesize}\\\begin{footnotesize}
		$\delta=-h$
		\end{footnotesize}};  ]]
		
		\end{tikzpicture}
		\caption{First scenario.}
  \end{subfigure} \hspace{0.1cm}
  \begin{subfigure}{.45\textwidth}
  		\centering
		\begin{tikzpicture}[level distance=80pt, sibling distance=10pt, edge from parent path={(\tikzparentnode) -- (\tikzchildnode)}]
		\tikzset{every tree node/.style={align=center}}
		\Tree [.\node[draw]{$x_1=0$}; \edge node[auto=right,pos=.6]{Yes}; [.
		\node[draw]{$x_2 = 0$}; \edge node[auto=right,pos=.6]{Yes}; \node[circle,draw]{$\tau=h$\\ $\delta=0$}; \edge node[auto=left,pos=.6]{No}; \node[circle,draw]{$\tau=3h$\\ $\delta=0$};]
		\edge node[auto=left,pos=.6]{No};[.\node[draw]{$x_3 = 0 $}; \edge node[auto=right,pos=.6]{Yes}; \node[circle,draw]{$\tau=0$\\ $\delta=h$}; \edge node[auto=left,pos=.6]{No}; \node[circle,draw]{$\tau=0$\\ $\delta=3h$};  ]]
		\end{tikzpicture}
		\caption{Second scenario.}
  \end{subfigure}
  \caption[Simulation scenarios]{Simulation scenarios.}
  \label{fig: simultree}
\end{figure}

Moreover, we changed the number of clusters keeping their size fixed to $K=(10,20,30)$.\footnote{Note that K/2 clusters will be assigned to the discovery sample and the remaining clusters will be in the estimation set as in \cite{athey2016recursive}.} For each scenario we built three NCTs: one tree implementing the composite splitting rule for the treatment and spillover effects as in \eqref{qcomb_H} (with $w_{(1,0;0,0)}=w_{(0,1;0,0)}=0.5$), one tree using the singular splitting rule for the treatment effect $Q^{in,H}_{(1,0;0,0)}(\Pi)$ as in \eqref{eq:honest}, and a third tree using the singular splitting rule for the spillover effect $Q^{in,H}_{(0,1;0,0)}(\Pi)$ as in \eqref{eq:honest}. 

\subsection{Results}

\review{We evaluate the performance of our NCT with respect to two dimensions: (i) its ability to correctly identify the actual heterogeneous sub-populations, and (ii) its performance in the estimation of the leaf-specific treatment and spillover effects. More details on the performance measures can be found in Section \ref{sec:performance_measures} of the Online Appendix.}

We start by analyzing the ability of the algorithm to correctly detect the heterogeneous subgroups in the first simulation scenario, that is, when the heterogeneity is the same for the two causal effects of interest. Figure \ref{fig:correct_leaves} reports the average number of correctly discovered heterogeneous causal rules with the composite splitting rule or effect-specific splitting rules targeted to the treatment effect or the spillover effect, in the case of 10, 20, and 30 clusters. Given that in this scenario the causal rules defining the heterogeneity are the same for both effects, all the splitting rules, targeted to a single effect or to both effects, have a similar performance and are able to detect the two heterogeneous leaves with a success rate that gets higher as the effect size increases. In addition, as the number of clusters grows the minimum effect size allowing the algorithm to optimally discover all the heterogeneous sub-populations gets lower.

Table \ref{tab:30clusters} reports the results for the first scenario \ref{tab:second_scenario} for the performance of the estimator in the correctly detected leaves in the case of 30 clusters. \review{Results for the cases of 10 and 20 clusters are reported in the Online Appendix.} We only report the results of the estimation procedure on the tree built with the composite splitting rule, as the spitting rule would only affect the identification of the heterogeneous sub-populations but not the estimation of the causal effects once these sub-populations are correctly detected. The estimator is able to estimate the heterogeneous treatment and spillover effects without bias and its precision grows as the number of clusters increases. Interestingly, NCT provides more accurate estimates of the heterogeneous spillover effects than the treatment effects. This is due to the larger number of units with $(W_{ik}=0, G_{ik}=1)$ than those with $(W_{ik}=1, G_{ik}=0)$, by definition of the network exposure (by design, units with $(W_{ik}=0, G_{ik}=1)$ are on average 50\% more than units with $(W_{ik}=1, G_{ik}=0)$). The higher precision in the estimation of the spillover effects is also reflected in the identification of the heterogeneous subgroups, which is more accurate when splits are targeted to the minimization of the MSE of the spillover effect (Figure \ref{fig:correct_leaves}).  

\begin{figure}[H]
\centering
\includegraphics[width=1\textwidth]{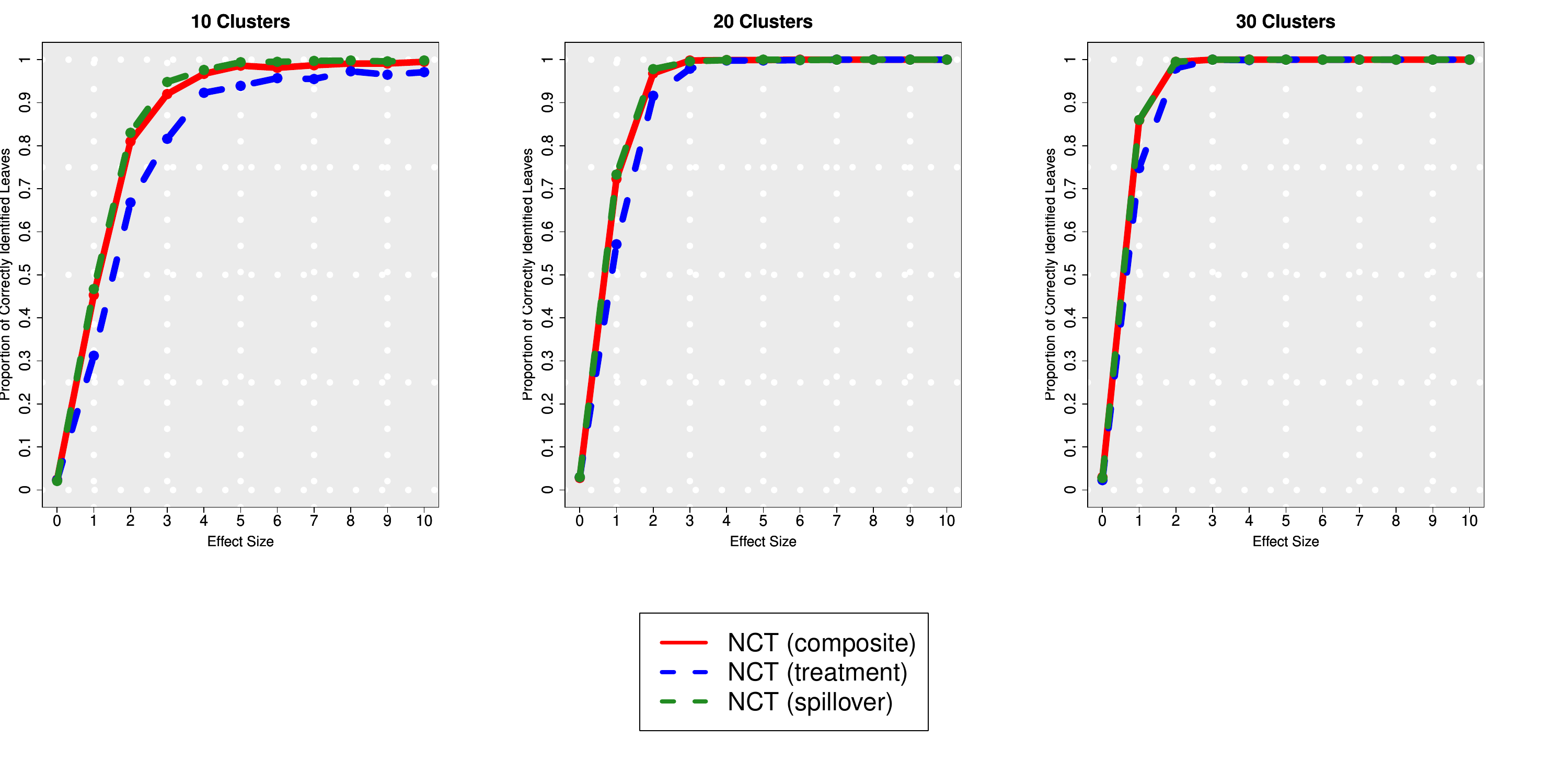}
		\caption{Simulations' results for correctly discovered leaves in the first scenario with 2 correct leaves and 10, 20, and 30 clusters, respectively.}
	\label{fig:correct_leaves}
\end{figure}

\vspace{-1cm}

\vspace{2cm}
\begin{table}[H]
\centering
\scriptsize

\begin{tabular}{cccccccc}
\multicolumn{1}{l}{}                  & \multicolumn{7}{c}{Treatment Effects}                                                                                                                 \\ \cline{2-8} 
\multicolumn{1}{l}{Effect Size}       & $\hat{\tau}_{\ell_1}$   & $\hat{se}(\hat{\tau}_{\ell_1})$   & $\hat{\tau}_{\ell_2}$   & $\hat{se}(\hat{\tau}_{\ell_2})$   & MSE   & Bias   & Coverage \\ \hline
0                                     & 0.140                   & 0.166                             & -0.146                  & 0.170                             & 0.034 & 0.143  & 0.929    \\
\rowcolor{gray!10} 1 & 1.027                   & 0.210                             & -1.024                  & 0.210                             & 0.041 & 0.025  & 0.969    \\
2                                     & 2.021                   & 0.302                             & -2.014                  & 0.304                             & 0.088 & 0.017  & 0.953    \\
\rowcolor{gray!10} 3 & 2.984                   & 0.410                             & -2.989                  & 0.407                             & 0.152 & -0.014 & 0.950    \\
4                                     & 3.982                   & 0.519                             & -3.979                  & 0.518                             & 0.246 & -0.019 & 0.959    \\
\rowcolor{gray!10} 5 & 4.962                   & 0.634                             & -4.978                  & 0.637                             & 0.362 & -0.030 & 0.957    \\
6                                     & 6.021                   & 0.761                             & -6.064                  & 0.758                             & 0.486 & 0.043  & 0.962    \\
\rowcolor{gray!10} 7 & 6.986                   & 0.875                             & -7.030                  & 0.884                             & 0.689 & 0.008  & 0.950    \\
8                                     & 7.940                   & 0.987                             & -8.012                  & 1.015                             & 0.978 & -0.024 & 0.966    \\
\rowcolor{gray!10} 9 & 8.951                   & 1.105                             & -9.104                  & 1.137                             & 1.125 & 0.027  & 0.962    \\
10                                    & 10.023                  & 1.244                             & -10.058                 & 1.270                             & 1.430 & 0.040  & 0.955    \\ \hline
\multicolumn{1}{l}{}                  & \multicolumn{7}{c}{Spillover Effects}                                                                                                                 \\ \cline{2-8} 
\multicolumn{1}{l}{}                  & $\hat{\delta}_{\ell_1}$ & $\hat{se}(\hat{\delta}_{\ell_1})$ & $\hat{\delta}_{\ell_2}$ & $\hat{se}(\hat{\delta}_{\ell_2})$ & MSE   & Bias   & Coverage \\ \hline
0                                     & 0.116                   & 0.141                             & -0.106                  & 0.150                             & 0.020 & 0.110  & 0.929    \\
\rowcolor{gray!10} 1 & 1.020                   & 0.170                             & -1.017                  & 0.171                             & 0.026 & 0.018  & 0.963    \\
2                                     & 2.004                   & 0.225                             & -2.006                  & 0.223                             & 0.040 & 0.005  & 0.973    \\
\rowcolor{gray!10} 3 & 2.990                   & 0.292                             & -2.995                  & 0.292                             & 0.065 & -0.007 & 0.975    \\
4                                     & 3.997                   & 0.368                             & -3.998                  & 0.368                             & 0.101 & -0.003 & 0.982    \\
\rowcolor{gray!10} 5 & 4.986                   & 0.445                             & -4.986                  & 0.444                             & 0.137 & -0.014 & 0.979    \\
6                                     & 6.041                   & 0.528                             & -6.012                  & 0.525                             & 0.185 & 0.026  & 0.986    \\
\rowcolor{gray!10}7  & 6.944                   & 0.603                             & -6.983                  & 0.606                             & 0.247 & -0.037 & 0.976    \\
8                                     & 8.003                   & 0.687                             & -7.957                  & 0.688                             & 0.297 & -0.020 & 0.981    \\
\rowcolor{gray!10} 9 & 9.022                   & 0.770                             & -9.006                  & 0.770                             & 0.391 & 0.014  & 0.981    \\
10                                    & 9.965                   & 0.848                             & -9.961                  & 0.851                             & 0.492 & -0.037 & 0.987    \\ \hline
\end{tabular}

\caption{Simulations' results for the first scenario (30 clusters)}
\label{tab:30clusters}
\end{table}

For the second scenario with different causal rules for each causal effect, we only report the results for simulations with 30 clusters. Figure \ref{fig:correct_leaves_4_rules} depicts the average number of correctly discovered heterogeneous causal rules with the composite splitting rule or effect-specific splitting rules targeted to the treatment effect or the spillover effect. 

When we are interested in building a tree that can represent the heterogeneity of all causal effects simultaneously (left panel), the composite splitting rule NCT is able to correctly identify all the heterogeneous causal rules (four in this example), while the other two NCT, targeted to either the treatment effect or the spillover effect, only detect the two leaves where the corresponding causal effect is heterogeneous. This is even more clear when looking at the other two plots, where we depict the ability to detect just the treatment effect rules (central panel) and the spillover effect rules (right panel). Indeed, when we are interested in subgroups that are heterogeneous with respect to only one causal effect, both the effect-specif spitting rule targeted to that effect or the composite spitting function can be used and perform similarly, while the use of the effect-specif spitting rule targeted to the other effect results in poor detection of the correct causal rules. The results from this second scenario show the clear added value of the composite splitting rule. Indeed, when the HDVs are different for treatment and spillover effects implementing this splitting rule enables the researcher to correctly spot all the true causal rules \textit{simultaneously}. Finally, Table \ref{tab:second_scenario} shows how the Horvitz-Thompson estimator performs well in estimating the conditional treatment and spillover effects in the two corresponding heterogeneous leaves.  

\vspace{1cm}
\begin{figure}[H]
\centering
\includegraphics[width=1\textwidth]{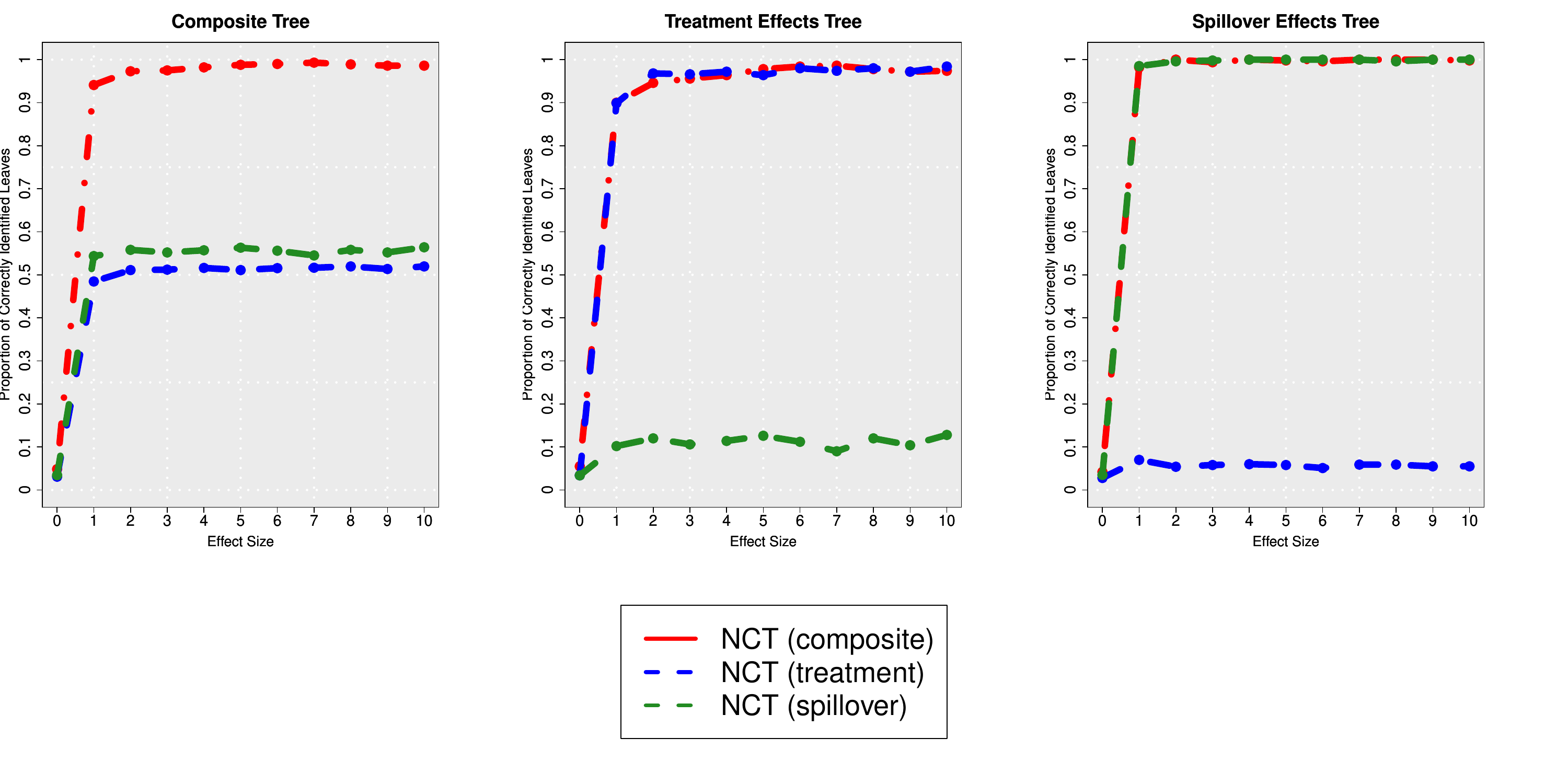}
		\caption{Simulations' results for correctly discovered leaves in the second scenario with 30 clusters.}
	\label{fig:correct_leaves_4_rules}
\end{figure}

\vspace{-0.5cm}

\vspace{2cm}
\begin{table}[H]
\centering
\scriptsize

\begin{tabular}{cccccccc}
\multicolumn{1}{l}{}                  & \multicolumn{7}{c}{Treatment Effects}                                                                                                                 \\ \cline{2-8} 
\multicolumn{1}{l}{Effect Size}       & $\hat{\tau}_{\ell_1}$   & $\hat{se}(\hat{\tau}_{\ell_1})$   & $\hat{\tau}_{\ell_2}$   & $\hat{se}(\hat{\tau}_{\ell_2})$   & MSE   & Bias   & Coverage \\ \hline
0                                     & -0.062                  & 0.171                             & 0.004                   & 0.180                             & 0.034 & -0.029 & 0.957    \\
\rowcolor{gray!10} 1 & 1.008                   & 0.209                             & 3.046                   & 0.414                             & 0.100 & 0.027  & 0.961    \\
2                                     & 1.988                   & 0.295                             & 5.981                   & 0.752                             & 0.279 & -0.016 & 0.957    \\
\rowcolor{gray!10} 3 & 2.998                   & 0.409                             & 9.044                   & 1.129                             & 0.601 & 0.021  & 0.960    \\
4                                     & 4.028                   & 0.524                             & 12.001                  & 1.495                             & 1.041 & 0.014  & 0.955    \\
\rowcolor{gray!10} 5 & 4.976                   & 0.634                             & 15.022                  & 1.849                             & 1.530 & -0.001 & 0.960    \\
6                                     & 5.967                   & 0.748                             & 18.077                  & 2.240                             & 2.665 & 0.022  & 0.949    \\
\rowcolor{gray!10} 7 & 6.987                   & 0.880                             & 20.965                  & 2.595                             & 3.352 & -0.024 & 0.954    \\
8                                     & 8.023                   & 1.001                             & 23.985                  & 2.937                             & 4.349 & 0.004  & 0.958    \\
\rowcolor{gray!10} 9 & 8.963                   & 1.120                             & 27.081                  & 3.328                             & 5.371 & 0.022  & 0.960    \\
10                                    & 9.948                   & 1.245                             & 29.986                  & 3.691                             & 6.460 & -0.033 & 0.961    \\ \hline
\multicolumn{1}{l}{}                  & \multicolumn{7}{c}{Spillover Effects}                                                                                                                 \\ \cline{2-8} 
\multicolumn{1}{l}{}                  & $\hat{\delta}_{\ell_1}$ & $\hat{se}(\hat{\delta}_{\ell_1})$ & $\hat{\delta}_{\ell_2}$ & $\hat{se}(\hat{\delta}_{\ell_2})$ & MSE   & Bias   & Coverage \\ \hline
0                                     & 0.046                   & 0.141                             & 0.064                   & 0.150                             & 0.031 & 0.055  & 0.923    \\
\rowcolor{gray!10} 1 & 1.017                   & 0.171                             & 3.009                   & 0.292                             & 0.043 & 0.013  & 0.976    \\
2                                     & 2.000                   & 0.224                             & 6.016                   & 0.525                             & 0.121 & 0.008  & 0.974    \\
\rowcolor{gray!10} 3 & 2.996                   & 0.295                             & 9.028                   & 0.771                             & 0.231 & 0.012  & 0.974    \\
4                                     & 4.003                   & 0.368                             & 12.077                  & 1.020                             & 0.381 & 0.040  & 0.977    \\
\rowcolor{gray!10} 5 & 4.984                   & 0.446                             & 15.017                  & 1.268                             & 0.592 & 0.000  & 0.981    \\
6                                     & 6.010                   & 0.525                             & 18.063                  & 1.520                             & 0.808 & 0.036  & 0.978    \\
\rowcolor{gray!10}7  & 7.026                   & 0.607                             & 21.020                  & 1.770                             & 1.252 & 0.023  & 0.979    \\
8                                     & 8.067                   & 0.691                             & 23.970                  & 2.019                             & 1.551 & 0.018  & 0.978    \\
\rowcolor{gray!10} 9 & 9.007                   & 0.773                             & 27.017                  & 2.271                             & 1.892 & 0.012  & 0.986    \\
10                                    & 9.945                   & 0.849                             & 29.983                  & 2.527                             & 2.331 & -0.036 & 0.978    \\ \hline
\end{tabular}

\caption{Simulations' results for second scenario (30 clusters)}
\label{tab:second_scenario}

\end{table}

\review{
\section{Empirical Application} \label{sec:application}

Our motivating application concerns the investigation of the effectiveness of intensive training sessions to promote agricultural insurance policies against extreme weather and the data-driven identification of the determinants of the heterogeneous response among farmers, when training sessions were directly received by them (treatment effect) or when training sessions were received by their friends (spillover effect). In this scenario, interference is likely to arise. For instance, households that have undergone intensive information sessions on the insurance policy could disseminate their acquired information within their social circles, thereby indirectly motivating untreated households to consider adopting the policy. As a result, untreated households within the same village might benefit from these sessions through their interactions with households in the same village that have received treatment. As highlighted in Section \ref{sec:motivating_application}, this knowledge plays a pivotal role by fostering a deeper understanding of the policies and creating opportunities for targeted interventions that aim to optimize their cost-effectiveness. Besides the importance of investigating the heterogeneity in the treatment and spillover effects, in the scenario of our application, it is critically important to consider network interference to avoid biased estimates of the treatment effect which may lead to flawed policy evaluations. Furthermore, comprehending spillovers is essential to grasp the degree to which untreated households gain from the intervention, as a result of interacting with treated farmers and being indirectly motivated to adopt the insurance policy.

This section is structured as follows: Subsection \ref{subsec:data} introduces the data used in the application and the notation; Subsection \ref{subsec:results} presents the results for the heterogeneous direct and spillover effect analyses via NCT. In particular, Subsection \ref{subsec:results_one} presents the results for an exposure mapping defined as having at least one friend that attended the intensive information session; Subsection \ref{subsec:results_two} describes the results for an exposure mapping defined as having two or more treated friends; Subsection \ref{subsec:results_delay} presents the results for the interesting case of households that received the information session at a delayed time with respect to the initial roll-out of the treatment. Considering these three cases is extremely important to understand the global impact of the intervention. In particular, the scenario examined in Subsection 7.2.2 allows us to check the robustness of results obtained in the main analysis (presented in Subsection 7.2.1) with respect to a different specification of the network exposure mapping, so as to account for potential misspecification of the interference mechanism. The scenario explored in Subsection 7.2.3 provides useful insights into the effectiveness of the intervention on those households who receive training sessions (both intensive or \emph{simple}) at a delayed time. This investigation is important to assess whether the timing in which households are exposed to the training sessions actually changes the effect of the intervention. 
 
\subsection{Data and Notation}\label{subsec:data}

To address this research question, we utilize data obtained from a factorial randomized experiment conducted to evaluate the effects of an intensive information session on the adoption of a new weather insurance policy among farmers in rural China, first analyzed by \cite{cai2015social}. The dataset comprises $N=4,569$ households, distributed across $K=47$ distinct villages. Households located in the same village are linked according to a village-specific network, while between-villages connections are ruled out. The observed friendship network $G$ is represented in Figure \ref{fig: villages}.
\begin{figure}[t]
\centering
\includegraphics[width=80mm]{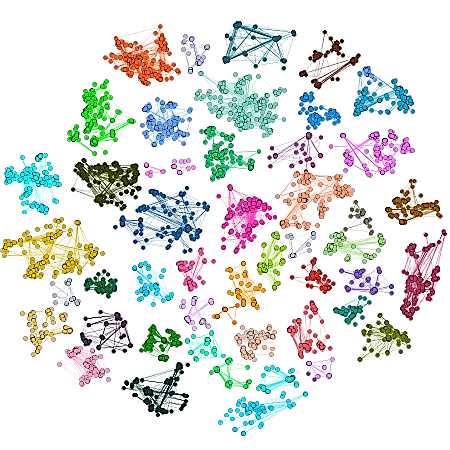}
		\caption{Friendship network between households living in rural villages in China. Colors refer to different villages.}
	\label{fig: villages}
\end{figure}

The experiment was conducted under the hypothesis that improving farmers’ understanding of the new weather insurance would reinforce its take-up. Two types of information sessions were offered: simple sessions that took around 20 minutes and provided an introduction to the insurance contract; and intensive sessions that took around 45 minutes and covered a more in-depth
explanation of the insurance and its expected benefits. In addition to the type of information session, farmers were randomized to receive the session in one of two rounds, scheduled a few days apart. The purpose of this delay was to allow for information sharing by first-round participants in order to investigate the spillover effect across friends. Second-round participants were further randomized to a third factor with three levels
representing the amount of additional information provided on previous purchases by first-round participants: i) no additional information; ii) overall take-up rate in their village; iii) a detailed list of those who purchased the insurance. For simplicity, we do not distinguish between the two types of information provided.  

According to this factorial design, different combinations of the three factors may be compared. For instance, the researcher may be interested in assessing the effect of being exposed to intensive information sessions, with no delays. Here, treated households are those who receive intensive information sessions in the first round, and untreated households are all those households who either receive intensive sessions at the delayed time or receive simple training sessions (this definition of the intervention is employed in the main analysis proposed by \cite{cai2015social}. The detailed characterization of the type of intervention assigned to the farmers' households can be summarized by a complex indicator $S_{ik}$ with 6 categories representing the specific combination of the three factors: (i) $S_{ik}=0$ if the household $i$ living in village $k$ receives a simple information session at round 2 with additional information, on the insurance take-up; (ii) $S_{ik}=1$ if it receives a simple information session at round 2 without additional information; (iii) $S_{ik}=2$ if it receives an intensive information session at round 2 with additional information; (iv) $S_{ik}=3$ if it receives an intensive information session at round 2 without additional information; (v) $S_{ik}=4$ if it receives a simple session at round 1; (vi) $S_{ik}=5$ if it receives an intensive information session at round 1.

In the main analysis we present here, we evaluate the effects of intensive information sessions received at round 1 on insurance take-up. Thus, in this analysis, the treatment variable $\Wik$ equals 1 if household $i$ in village $k$ received an intensive information session on the promoted insurance policy at round 1---i.e., if $S_{ik}=5$---while it equals 0 otherwise---i.e., if $S_{ik}=0, \dots, 4$. According to the clustered network in Figure \ref{fig: villages} and to this definition of treatment, Figure \ref{fig: degree_neightr} represents the distribution of the number of treated network neighbors, together with that of neighbors, referred to as the outdegree.
 \begin{figure}[t]
	\centering
	\begin{subfigure}{.5\textwidth}
	\centering		\includegraphics[width=85mm]{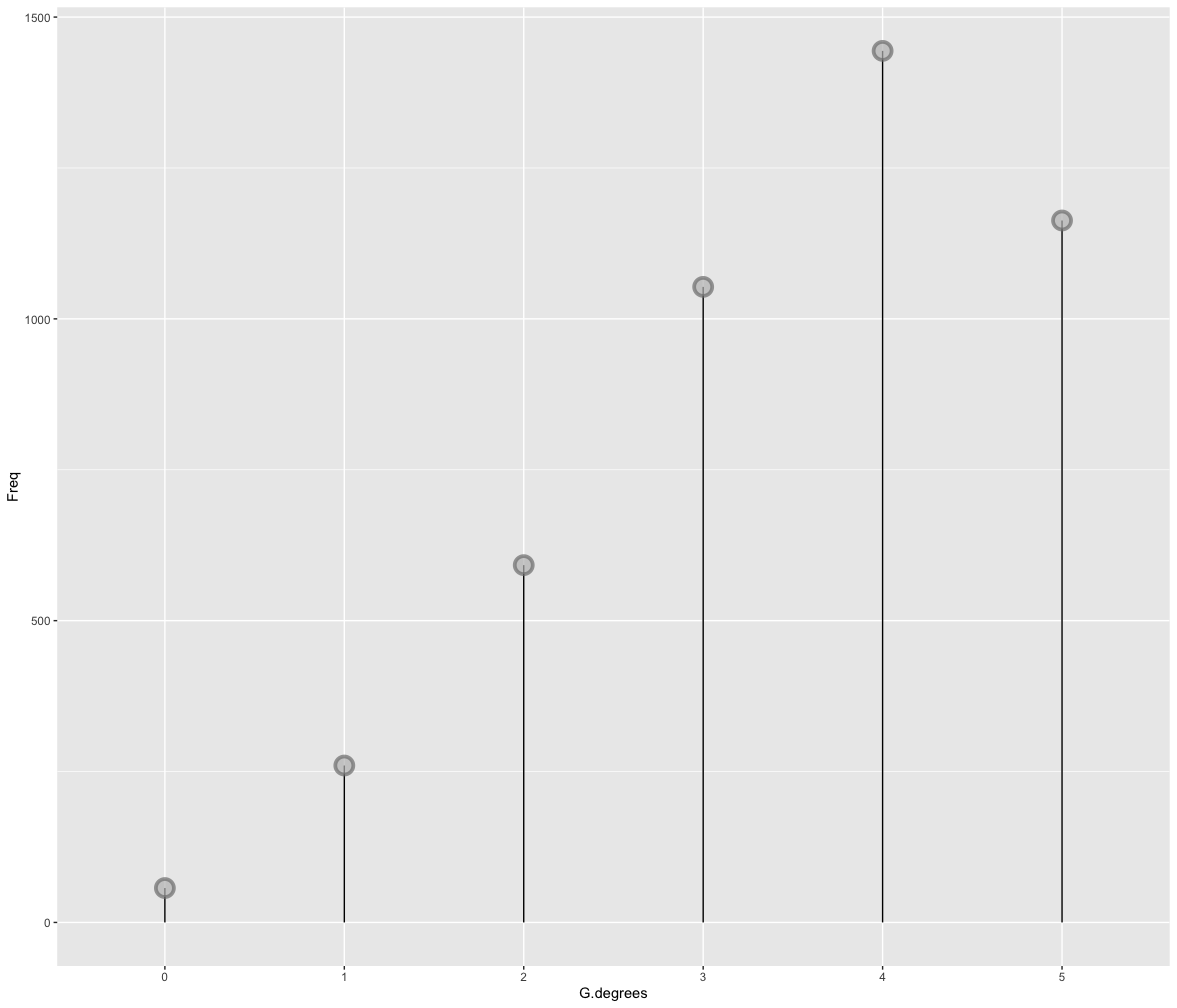}
		\caption{Number of neighbors}
		\label{fig: degree}
 	\end{subfigure}%
	\begin{subfigure}{.5\textwidth}
		\centering
 	\includegraphics[width=85mm]{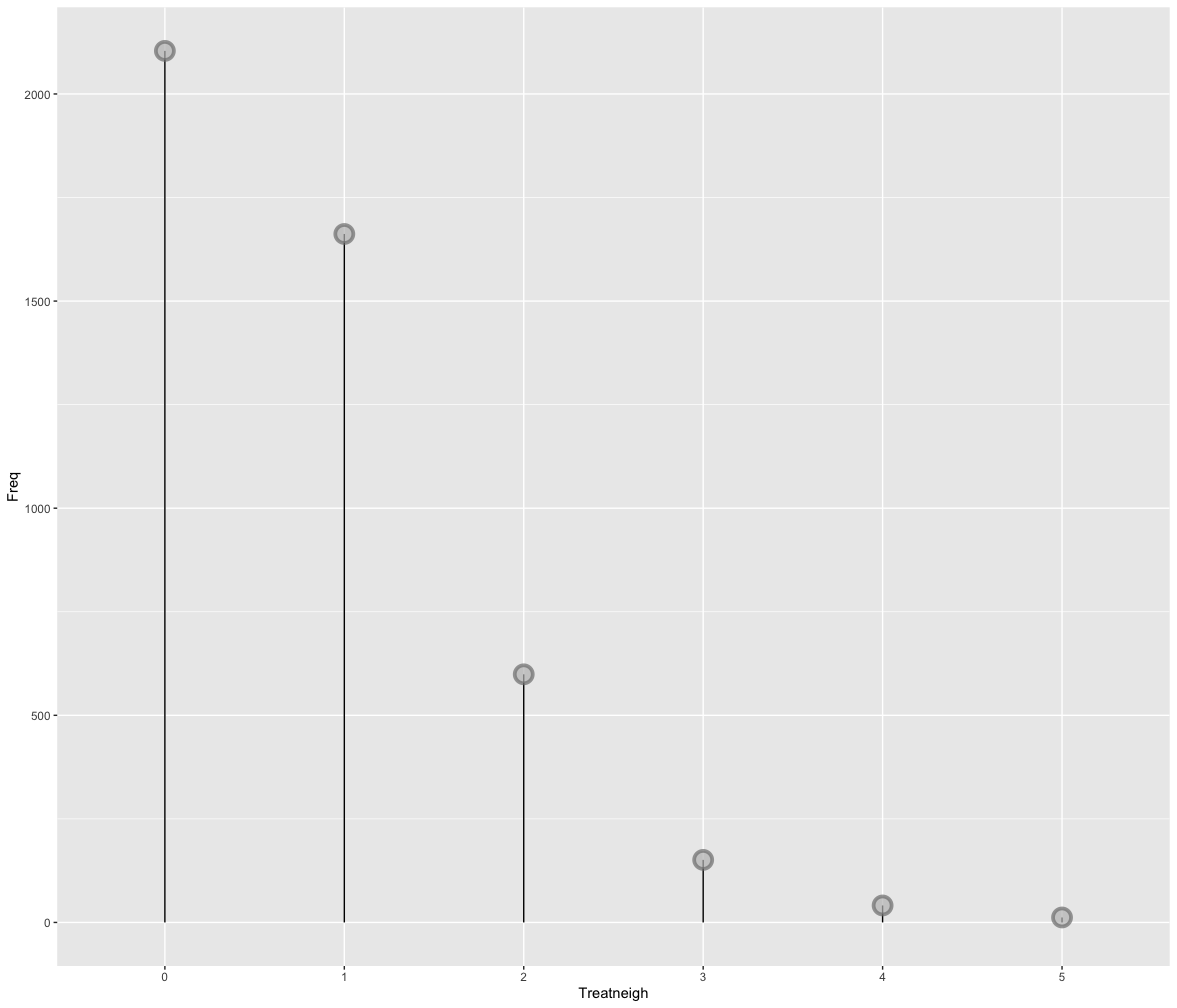}
		\caption{Number of treated neighbors}
		\label{fig: neightr}
	\end{subfigure}
	\caption{Distribution of the outdegree (a) and the number of treated neighbors (b)}
 \label{fig: degree_neightr}
 \end{figure} 

In this setting, interference could take place because the information received in the intensive information session might be transferred to network contacts. We assume that exposure to the network treatments can only occur from network neighbors. That is, we exclude the impact from higher-order neighbors---e.g., friends of friends. This assumption is justifiable by the fact that the decision of a farmer to buy weather insurance might depend only on the information either received directly in the intensive session or indirectly by one of his friends. This is also what has been found in \cite{cai2015social} where the only significant spillover effect is from immediate friends. In particular, we define the network (indirect treatment) exposure according to the following rule:$G_{ik}=\mathds{1}(\bigl(\sum_{j\in \Nik} W_{jk}) \geq 1 \bigr)$. Hence, $G_{ik}$ is equal to 1 if household $i$ in village $k$  has at least one treated friend and 0 otherwise. To assess the robustness of the results, we will further vary the threshold defining the binary network exposure.

We dropped from the analysis the $68$ households without any friends. The total number of remaining households is $4,518$, residing in $47$ villages. Among them, $977$ families were assigned to the intensive training sessions in the first round,  while $3,586$ belonged to the control group---i.e., received the simple session in the first round or any type of information session in the second round. $2,053$ households do not have any treated friends---i.e., $G_ik=0$---while $2,465$ of them undertake friendship relationships with at least one treated household---i.e., $G_ik=1$. The joint distribution of the individual and network exposure is summarized in the following Table \ref{tab: empjointtr}.

\begin{table}[H]
\centering
\caption{Distribution of the joint treatment}
\begin{tabular}{l|cc}
\toprule
  &  $G_{i}=0$  &   $G_{i}=1$   \\
  \midrule
  $W_{i}=0$ & 1734  &  1807 \\
  $W_{i}=1$ & 319 & 658 \\
  \bottomrule
\end{tabular}

\label{tab: empjointtr}
\end{table}
\noindent Figure 
\ref{fig:2villages}

\ref{fig:4groupsindtr} and \ref{fig:4gropusjotr} 
represent the treatment distribution in three villages. In Figure \ref{fig:4groupsindtr}, nodes are colored according to their individual treatment assignment, while in Figure \ref{fig:4gropusjotr} node colors refer to the joint treatment status. Villages are graphically contoured through colored polygons.

 \begin{figure}[h]
	\centering
	\begin{subfigure}{.45\textwidth}
	\centering		\includegraphics[width=85mm]{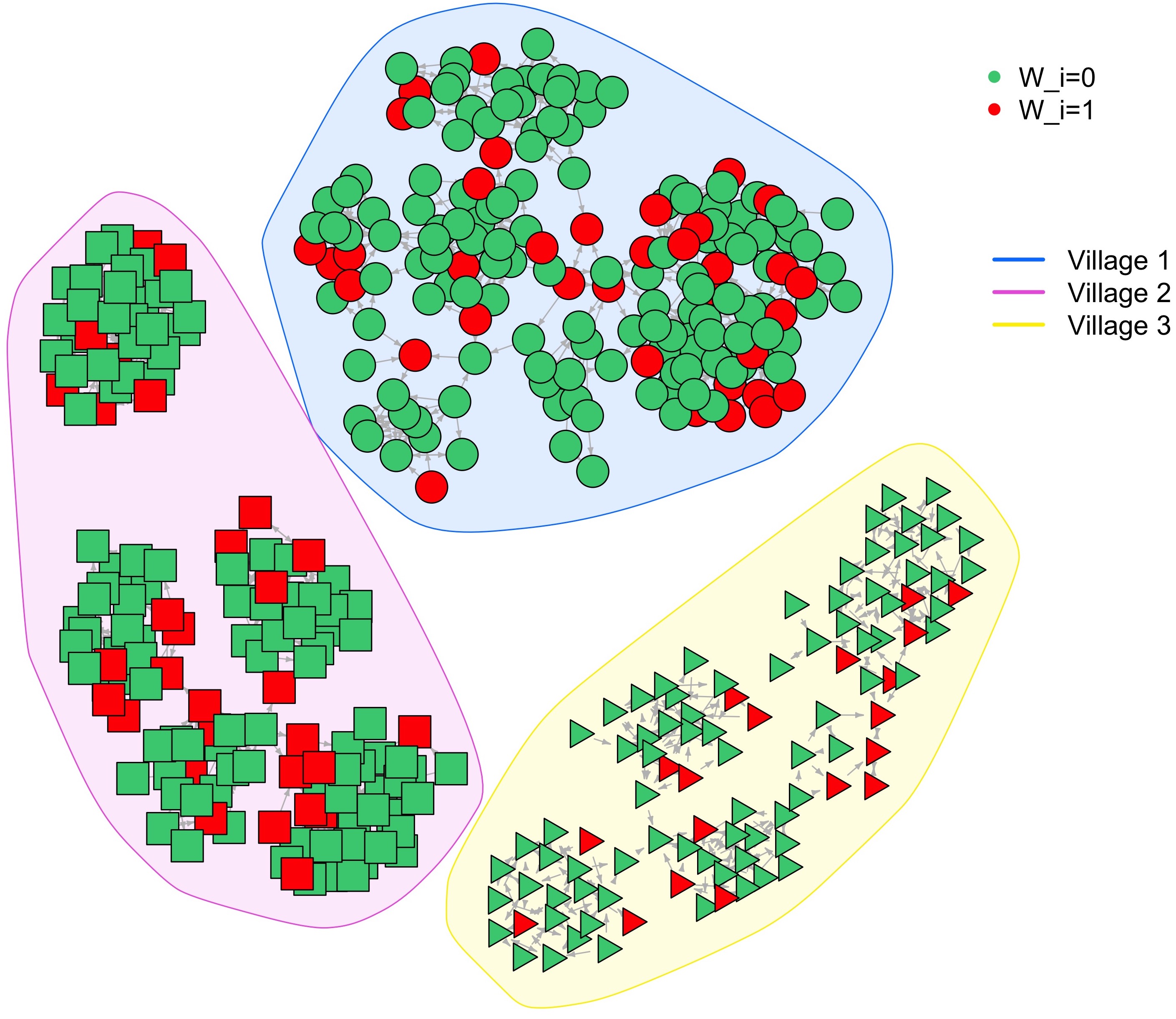}
		\caption{Individual treatment in three villages}
		\label{fig:4groupsindtr}
 	\end{subfigure}%
 	 	\hspace{1cm}
	\begin{subfigure}{.45\textwidth}
		\centering
 	\includegraphics[width=85mm]{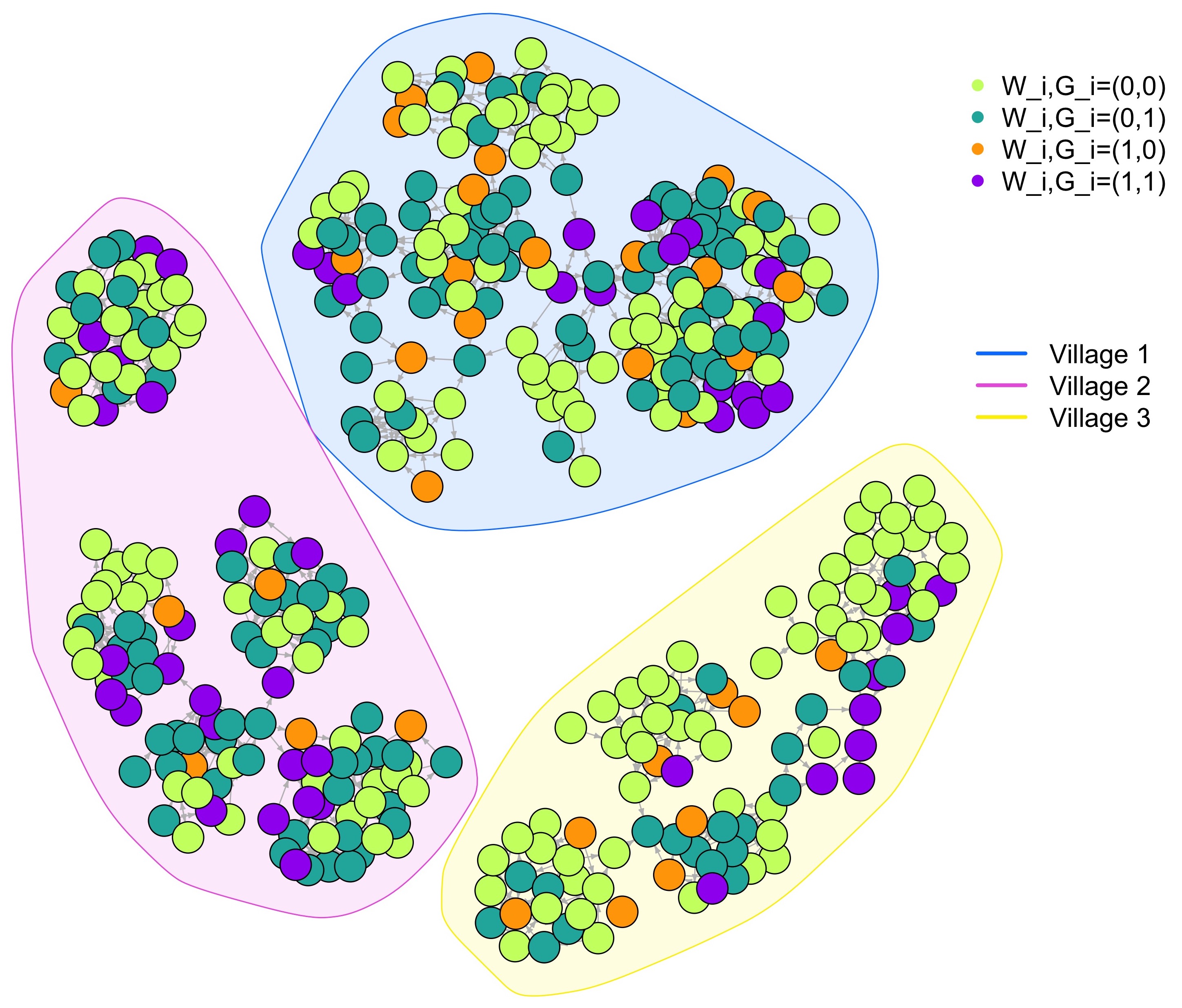}
		\caption{Joint treatment in three villages }
\label{fig:4gropusjotr}
	\end{subfigure}
 \caption{Individual treatment (a) and joint treatment distribution (b) in four villages}
  \label{fig:2villages}
 \end{figure}

The outcome variable $\Yik \in \{0,1\}$ represents the insurance uptake and equals 1 if the household $i$ residing in village $k$ purchased the insurance after the information session. At the end of the experiment, $2,498$ families chose not to accept the proposed insurance policy, while $2,020$ were positively persuaded by the session and accepted the weather insurance.

To evaluate the heterogeneity of treatment and spillover effects, we included in the analysis all the observed characteristics that could plausibly drive the heterogeneity in the effect of the intervention: three dummies representing the household's production area (\texttt{reg1}, \texttt{reg2}, and \texttt{reg3}, respectively); a binary variable which equals 1 if the head of the household is at least $50$ years old (\texttt{age}); a binary variable which is equal to one if the household's head is male (\texttt{male}); a binary variable that takes the value 1 if in the household there are more than four components (\texttt{hsize}); a binary variable being 1 if the household's head has successfully completed high school (\texttt{educ}); a binary variable distinguishing families who are strongly worried about weather phenomena (\texttt{prdis}, which is equal to 1 if the perceived probability of a weather disaster happening in the coming year exceeds $0.30$), which corresponds to the median of that variable in the sample of households; an indicator representing the household's risk aversion (\texttt{averse}); an indicator for whether the household was impacted by any weather disaster in the prior year (\texttt{dis}); an indicator measuring the trust in the Government (\texttt{trust}); a binary variable that equals 1 if the household has received payouts from other insurance products before (\texttt{repay}).

\subsection{Empirical Results}\label{subsec:results}

\subsubsection{One Treated Neighbor}\label{subsec:results_one}

\begin{figure}[!tb]
	\centering
\includegraphics[width=1\textwidth]{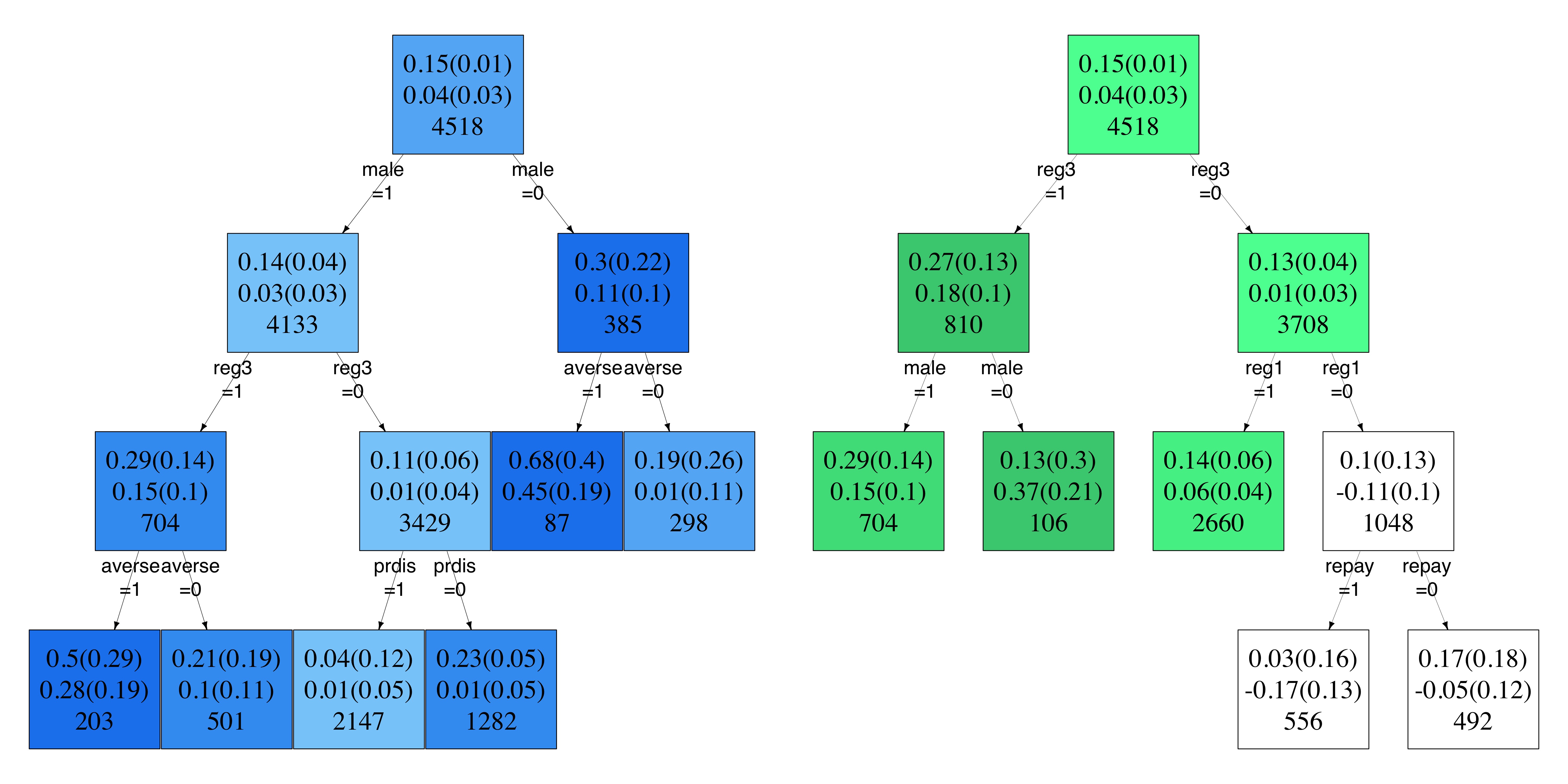}
	\caption{Network causal trees targeted to single effects: NCT targeted to the main treatment effect  $\tau(1,0;0,0)$ (left) and NCT targeted to the main spillover effect  $\tau(0,1;0,0)$ (right). In the left side figure nodes are colored according to the magnitude of the estimated main treatment effect (in shades of blue). In the right side figure nodes are colored according to the magnitude of the estimated main spillover effect (in shades of green). In both figures, white nodes represent negative effects. Each node reports - from top to bottom - the estimate of the main treatment effect, the estimate of the main spillover effect, and the size of the selected sub-population}
		\label{fig:empnctsing}
 \end{figure} 

In this section, we present the most relevant empirical findings, with the individual treatment $\Wik$ defined as receiving the intensive information session at round 1, i.e., $W_{ik}=\mathds{I}(S_{ik}=5)$, and the network exposure $\Gik$ defined as having at least one treated friend. Figures \ref{fig:empnctsing} and \ref{fig:empnctmult} report the trees built targeting single effects and multiple effects, respectively. In each node, we report from top to bottom the estimates of the main treatment and spillover effects and the size of the selected sub-populations. Colors provide intuition on the magnitude either of the main treatment effect $\tau(1,0;0,0)$ or of the main spillover effect $\tau(0,1;0,0)$: blue nodes are those where the estimated main treatment effects are positive, while green nodes are those where the estimated main spillover effects are positive; in both cases, the darker is the color the stronger is the magnitude of the effect of interest. Negative effects are colored in white. 
27 villages are randomly assigned to the training set, while the remaining villages are allocated to the estimation set only.\footnote{The random assignment of the villages to the training and estimation samples has been kept fixed in all the trees.} Furthermore, we have set the maximum depth at 3 to maintain a high level of interpretability, while we have set at 20 the minimum number of units that must be present in the child leaves for each of the four exposure conditions.
We can see that, in the whole population, the treatment has a positive effect on insurance take-up, while the spillover effect is positive but not statistically significant. The most relevant heterogeneity drivers are: the production area (specifically, living in the third area), sex, risk aversion, and the perceived probability of disaster However, the estimated tree slightly changes under different specifications of the spitting rule. Figure \ref{fig:empnctsing} represents the trees targeted to single effects: in particular, the left side figure shows the tree targeted to the treatment effect $\tau(1,0;0,0)$ while the right side figure shows the one targeted to the spillover effect $\tau(0,1;0,0)$, respectively. Note that the colors of the nodes refer to the magnitude of the effect the two trees are targeted to: hence, in the tree on the left, the colors of the nodes represent the magnitude of the treatment effect, while in the tree on the right they are related to the magnitude of the estimated spillovers. 
The NCT targeted to the treatment effect $\tau(1,0;0,0)$ shows that the most important variables driving the heterogeneity of the treatment effect are the production area, risk aversion, sex, and the perceived probability of disaster. The treatment appears to be particularly effective in the sub-population of households where the head is a female and she is risk averse, and in the one where the head of the household is a male who lives in the third region. Moreover, even households with a male head not living in the third production area and worried about possible weather disasters highly benefit from the intervention. We can speculate that households who are more worried about possible disasters benefit less from an intensive information session because even less information would prompt them to purchase weather insurance. 

When looking at the tree targeted to the spillover effect $\tau(0,1;0,0)$, living in the third region is still the main characteristic determining the heterogeneity of the spillover effect. In this partitioning, the households that are more responsive to the intensive session received by at least one of their friends in the first round are those who live in the third production area, and whose head is a female. However, the algorithm shows that also farmers living in the first production area are highly responsive to the intensive training sessions on the policy received by at least one of their friends in the first round. Yet, spillovers are rarely statistically significant in the leaves. In the sub-populations identified by this tree, the estimated treatment effects are also positive, and statistically significant.

\begin{figure}[!t]
	\centering
\includegraphics[width=1\textwidth]{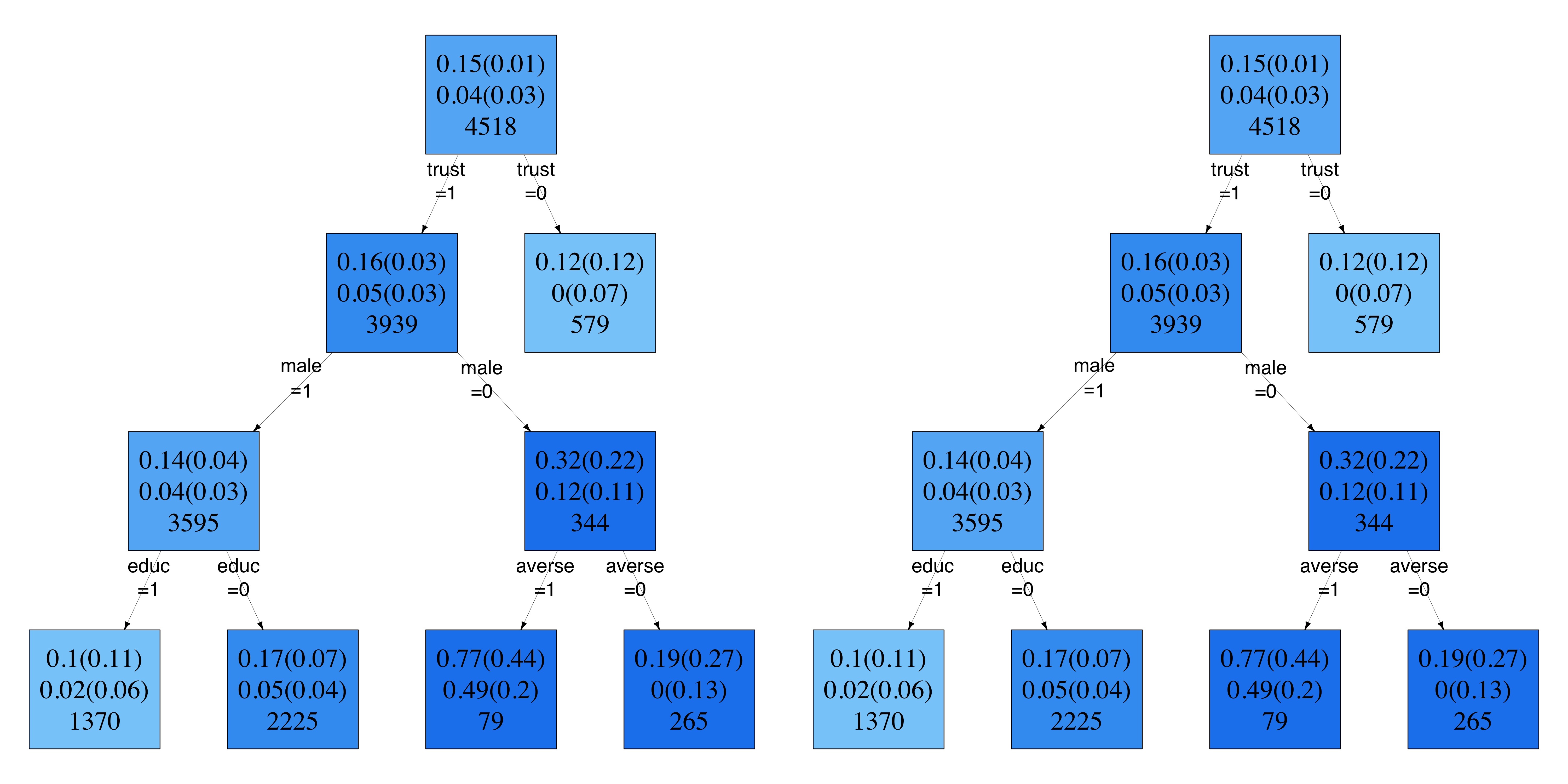}
	  	\caption{Network causal trees targeted to multiple effects:  NCT targeted to the main treatment effect  $\tau(1,0;0,0)$ and the main spillover effect  $\tau(0,1;0,0)$ (left) and NCT targeted to all the effects (right). In both figures, nodes are colored according to the magnitude of the estimated main treatment effect (in shades of blue). White nodes represent negative effects. Each node reports - from top to bottom - the estimate of the main treatment effect, the estimate of the main spillover effect, and the size of the selected sub-population}
	  		\label{fig:empnctmult}
        \vspace{-0.1cm}
 \end{figure}

Figure \ref{fig:empnctmult} depicts the partition selected using a composite splitting function targeted to multiple causal effects. Specifically, the left side figure refers to the network causal tree targeted to both $\tau(1,0;0,0)$ and $\tau(0,1;0,0)$, such that each component contributes to determining the objective function with equal weight (0.5), while the right side figure is related to the tree which has been built assigning equal weight to all the four effects in $\mathcal{T}$. Note that here colors refer to the magnitude of the treatment effects in both figures. In this application, the composite tree which considers both $\tau(1,0;0,0)$ and $\tau(0,1;0,0)$ coincides with the tree based on all effects simultaneously. Here the sub-population where the treatment is more effective is the one including households whose head is a female who trusts the government and who is risk-averse. These households are also those that are more responsive to the intensive session received by at least one of their friends at first. 

We can conclude that intensive training sessions encouraged Chinese rural households to take up insurance policies. The main determinants of the heterogeneity in the treatment effect are the production area, the trust in the government, and both the risk aversion and the sex of the head of the household. In this analysis, spillover effects do not seem to play a key role.

\subsubsection{Two Treated Neighbors}\label{subsec:results_two}
  
Here, we increase the threshold to define the network exposure to 2, i.e., $G_{ik}=\mathds{1}(\bigl(\sum_{j\in \Nik} W_{jk}) \geq 2 \bigr)$, such households who have 2 or more friends who attended the intensive information session at round 1 are considered indirectly exposed to the intensive session. In this setting, the positive impact of directly receiving the intensive session at round 1 (treatment effect) is even stronger than in the previous scenario, where having one treated neighbor is enough to determine an indirect exposure to the intervention. As in the previous scenario, spillover effects are not significant, neither in the whole population nor in the sub-populations identified by the trees.

\begin{figure}[t]
	\centering
\includegraphics[width=1\textwidth]{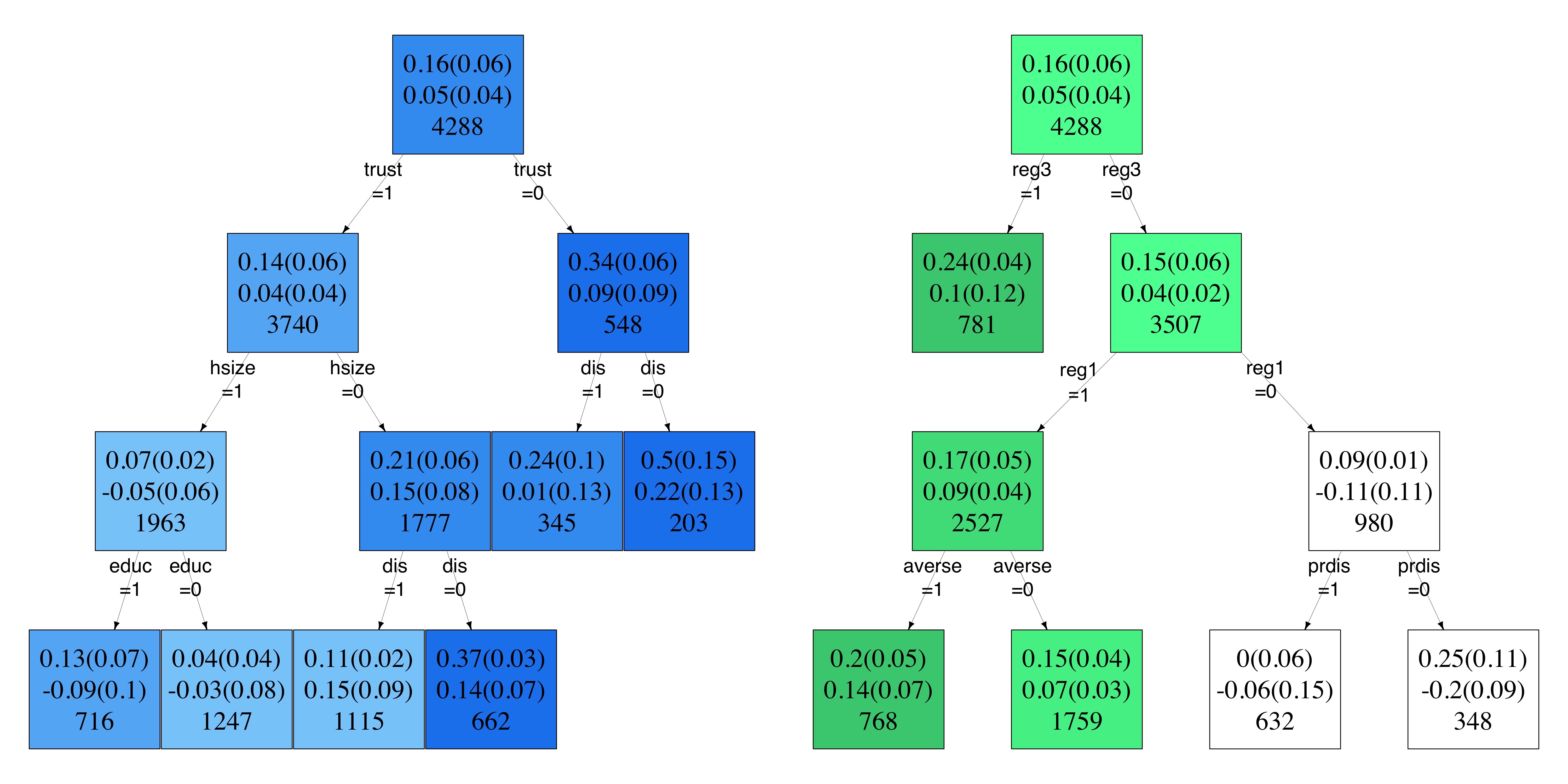}
	\caption{Network causal trees targeted to single effects, with network exposure defined by the cutoff $q=2$: NCT targeted to the main treatment effect $\tau(1,0;0,0)$ (left) and NCT targeted to the main spillover effect  $\tau(0,1;0,0)$ (right). In the left side figure nodes are colored according to the magnitude of the estimated main treatment effect (in shades of blue). In the right side figure nodes are colored according to the magnitude of the estimated main spillover effect (in shades of green). In both figures, white nodes represent negative effects. Each node reports---from top to bottom---the estimate of the main treatment effect, the estimate of the main spillover effect, and the size of the selected sub-population.
    }
		\label{fig:empnctsing_a1}
 \end{figure} 
 
The most important variables driving the heterogeneity of the effects are the region of living, the trust in the Government, the size of the household, and the indicator signaling whether the household was impacted by any weather disaster last year. Figure \ref{fig:empnctsing_a1} represents the trees targeted to $\tau(1,0;0,0)$ and $\tau(0,1;0,0)$, respectively. According to the tree that is targeted to $\tau(1,0;0,0)$ only (left), the most important heterogeneity drivers are the trust in the government, the size of the household, the fact of having been impacted by a weather disaster and the level of education. The treatment appears to be particularly effective in the following sub-populations: i) large and highly educated households who have faith in the government; ii) small households who have faith in the government and whose activity was not affected by a weather disaster in the previous year, and iii) households who do not trust the government and were not impacted by a weather disaster. According to this tree, small households who trust the government and who were not affected by disasters and households who do not trust the government and who were not impacted by a weather disaster last year also significantly benefit from indirect exposure to the intervention (all the other spillovers are positive but not statistically significant).  

When looking at the tree targeted to the main spillover effect (right), the production region, the perceived probability of disaster, and risk aversion determine the heterogeneity of spillover effects. According to this tree, households living in the third region who do not have faith in the government and risk-averse households living in the first production area are those who benefit the most from receiving information about weather insurance from some of their friends  In the leaves identified by this tree, also treatment effects are positive and statistically significant.

Figure \ref{fig:empnctmult_a1} depicts the partition selected using a composite splitting function targeted to multiple causal effects. Specifically, the left side figure refers to the network causal tree targeted to both $\tau(1,0;0,0)$ and $\tau(0,1;0,0)$, such that each component contributes to determining the objective function with equal weight (0.5), while the right side figure represents the tree that has been built assigning equal weight to all the four effects in $\mathcal{T}$. As in the previous scenario, where the network exposure is defined by the cutoff $q=1$, the two trees coincide. The most important heterogeneity drivers are the production area, the sex of the head, and the size of the household. The treatment is particularly effective in small households not in the second region that has a male head. It is also possible to observe that spillovers are positive and significant in households that live either in the first or third region and who have a female head. 

\begin{figure}[t]
	\centering
\includegraphics[width=1\textwidth]{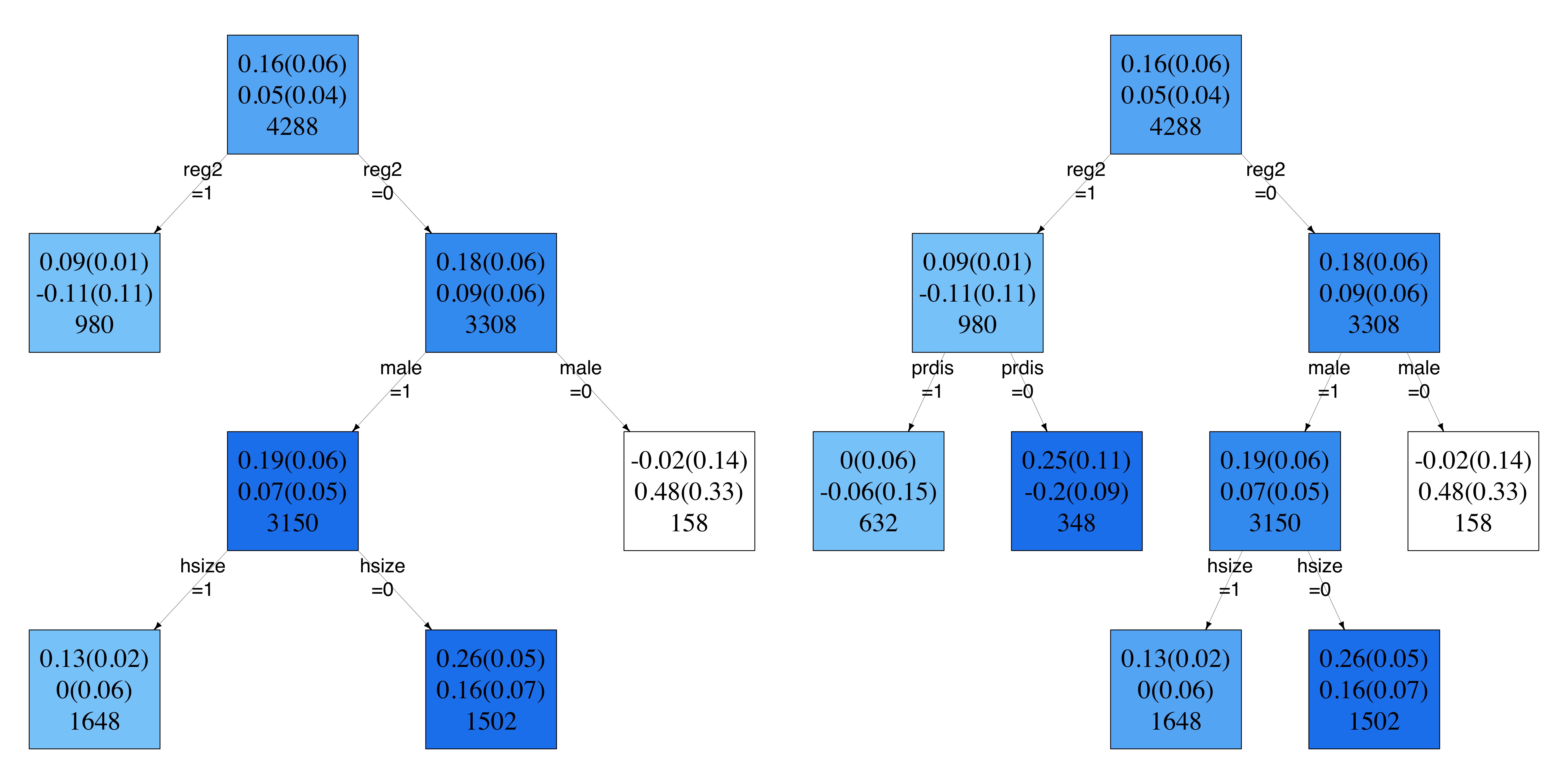}
	  	\caption{Network causal trees targeted to multiple effects, with network exposure defined by the cutoff $q=2$: NCT targeted to the main treatment effect  $\tau(1,0;0,0)$ and the main spillover effect  $\tau(0,1;0,0)$ (left) and NCT targeted to all the effects (right). In both figures, nodes are colored according to the magnitude of the estimated main treatment effect (in shades of blue). In both figures, white nodes represent negative effects. Each node reports---from top to bottom---the estimate of the main treatment effect, the estimate of the main spillover effect, and the size of the selected sub-population.
    }
	  		\label{fig:empnctmult_a1}
 \end{figure}

\subsubsection{Delayed Time of Intervention}\label{subsec:results_delay}

So far, we have presented results related to the assessment of the causal effects of intensive information sessions received at the first round on insurance take-up, that is, the analyses have been conducted with the treatment variable $\Wik$ being equal to 1 if household $ik$ received an intensive information session at round 1---i.e., if $S_{ik}=5$---and equal to 0 otherwise---i.e., if $S_{ik}=0, \dots, 4$. However, given the factorial design and the factorial combinations represented by the variable $S_{ik}$, we can focus on other meaningful comparisons. It is possible to define the treatment variable $\Wik$ and the network exposure variable  $\Gik$ using different combinations of the factorial variable $S_{ik}$. 

Here, we focus on households who received the information session at the second round, i.e., $S_{ik}=\{ 0,1,2,3 \}$, and we define as treated units households who receive the intensive information session, while untreated units correspond to the households who receive simple information sessions, i.e., 
$$
W_{ik}=
\begin{cases}
1 \quad \text{if} \quad S_{ik}\in \{ 2,3 \},\\
0 \quad \text{if} \quad S_{ik}\in \{ 0,1 \}.
\end{cases}
$$
The network exposure variable $\Gik$ is defined, as in the main analysis, as the indicator for having at least one friend who received the intensive information session at the first round, i.e., 
$$G_{ik}=\mathds{1}(\bigl(\sum_{j\in \Nik} \mathds{1}(S_{jk}=5) \geq 1 \bigr)$$.
With these definitions, we can investigate the spillover effect of for having at least one friend who received the intensive information session at the first round among those who did not receive any information until the second round. 

\begin{figure}[t]
	\centering
\includegraphics[width=1\textwidth]{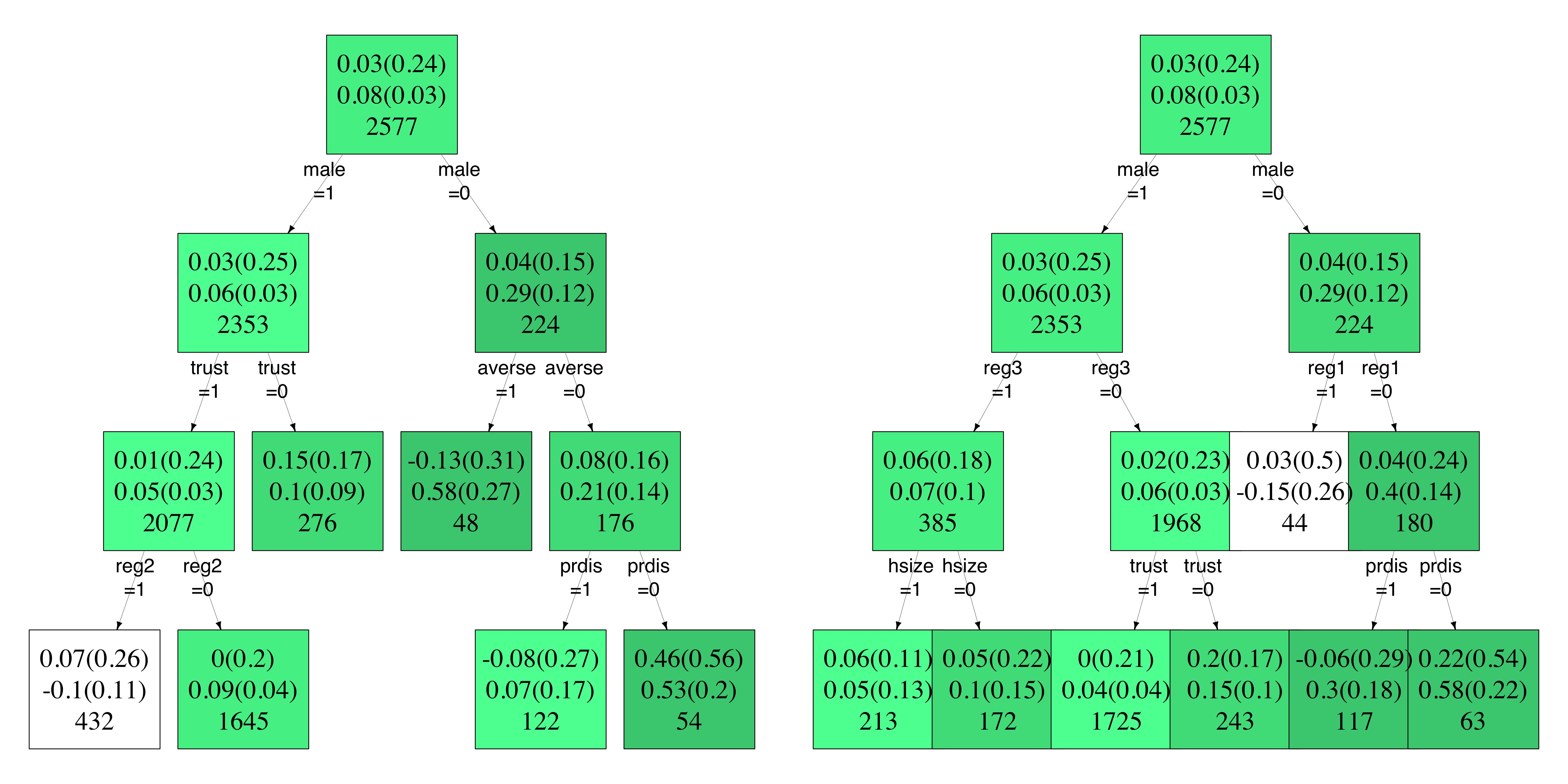}
	\caption{Network causal trees for the spillover effect on those who received the information session at the delayed time. NCT targeted to the main spillover effect  $\tau(0,1;0,0)$ (left) and NCT targeted to both the main treatment effect $\tau(1,0;0,0)$ and the main spillover effect, with equal weight (0.5) $\tau(0,1;0,0)$ (right). In both figures, nodes are colored according to the magnitude of the estimated main spillover effect (in shades of green). In both figures, white nodes represent negative effects.  Each node reports---from top to bottom---the estimate of the main treatment effect, the estimate of the main spillover effect, and the size of the selected sub-population.
 }
		\label{fig:empnct_b0}
 \end{figure}

We found that here the main spillover effect $\tau(0,1;0,0)$ in the whole sample, being here the effect of having at least one friend who received the intensive information session at round 1 on the insurance uptake for those who received a simple session at the second round, is positive and significant ($\widehat{\tau}=0.08$, $sd(\widehat{\tau})=0.03$), suggesting that being indirectly exposed to the intensive session through friends has an impact on those who only receive a simple session at a delayed time. 
On the contrary, in this case the main treatment effect $\tau(1,0;0,0)$ in the whole sample---i.e., the effect of directly receiving an intensive session in the second round versus receiving a simple session still in the second round with no friends who received the intensive session at round 1---is not significant ($\widehat{\tau}=0.03$, $sd(\widehat{\tau})=0.24$). 
Figure \ref{fig:empnct_b0} shows both the tree that has been built targeting the main spillover effect $\tau(0,1;0,0)$ and the network causal tree targeted to both $\tau(1,0;0,0)$ and $\tau(0,1;0,0)$, such that each component contributes in determining the objective function with equal weight (0.5). Note that in these figures the color of the leaves refers to the sign and the magnitude of the main spillover effect $\tau(0,1;0,0)$.

Results show that the main heterogeneity drivers of the spillover effect on those who received the simple session at round 2 are the trust in the government, the risk aversion, and the sex of the head, while when the treatment effect is included in the composite rule, the production area plays the most important role. Specifically, the left side tree included in Figure \ref{fig:empnct_b0}---which represents the tree targeted to the main spillover effect only---shows that households whose head is a female who is either risk-averse or not-risk-averse but with a low perceived probability of an upcoming disaster are those who, if they received simple information sessions at the delayed time, would benefit the most from interacting with friends who received intensive information sessions at the first round. In the partitions identified by the tree, the main treatment effect is not significant. When looking at the tree targeted to both treatment and spillover effect (right side tree of Figure \ref{fig:empnct_b0}), the sex of the head and the production area are the main drivers that determine the heterogeneity of the (composite) effect. Secondary drivers are the size of the household, the trust in the government, and risk aversion. In the partitions identified by this tree, spillover effects are positive: they are particularly high in the sub-population of households that have a female head and who are not in the first production region. As observed in the tree targeted to the main spillover effect only, in most of the sub-populations identified by this tree the estimated treatment effects are positive but not significant.

}

\section{Conclusions} \label{sec:conclusions}


In this paper, we have introduced a new algorithm to estimate heterogeneous causal treatment and spillover effects in the presence of clustered network interference. The proposed network causal tree bridges the gap between two streams of literature: causal inference under interference and tree-based methods for the discovery of heterogeneous sub-populations. We build upon the seminal algorithm proposed by \cite{athey2016recursive} to account for clustered network interference through a rework of the splitting function. Leaf-specific causal effects are then estimated using the Horvitz-Thompson  estimator proposed by \cite{aronow2017estimating}. 

The proposed NCT algorithm has enhanced interpretability and shows an excellent performance in a set of Monte Carlo simulations. In particular, the algorithm is able both to spot the relevant sources of heterogeneity in the data and to consistently estimate the conditional treatment and spillover effects. Moreover, we have introduced a composite splitting function that allows the researchers to simultaneously detect the sub-populations where both treatment and spillover effects are heterogeneous. The identification of these multi-effect heterogeneous subgroups is crucial for the design of targeting strategies that involve multiple effects. 
 
Our simulation study shows that the use of such a composite splitting rule is able to correctly detect all the heterogeneous sub-populations defined by both treatment and spillover effects, as well as the ones defined by one effect only. Therefore, the selected partition of the population can be used to design strategies whose objective function incorporates multiple effects, but can also be used a posteriori to target subgroups maximizing either a treatment or a spillover effect. 
\review{However, further research is needed to use the selected partition to actually design targeting strategies involving both treatment and spillover effects. In addition, these strategies should in principle rely also on the heterogeneity of spillover effects with respect to the characteristics of the senders (treated individuals), and not just those of the receivers (individuals whose outcome is affected by the treatment of others). To this end, new causal estimands could be defined and their heterogeneity should be investigated and incorporated in the design of complex targeting rules, aimed at maximizing the effect of the intervention in the population, while saving resources.}

\review{When applied to real-world data, the NCT algorithm provided useful insights on the effectiveness of intensive training sessions among Chinese rural households on the uptake of a weather insurance policy. We found that the main characteristics that drive the heterogeneity of the effects are the production area, the sex of the farmers, the trust in the government, and risk aversion. Results also suggest that the intervention is not statistically significant in those households who receive the information sessions at a delayed time: however, in those households, the spillover effect is positive and significant and its heterogeneity is mostly driven by the sex of the owner, the trust in the government and the risk aversion.}

The proposed algorithm may suffer from the limitations common to (single) tree-based methods: instability to the random allocation of units in the training sample, the potential impact of outlier observations in the node-specific estimations, and the difficulty of modeling very smooth outcomes.   However, (single) tree-based algorithms, thanks to their high interpretability, are suitable for the discovery of heterogeneous sub-populations, which play a crucial role in policy design. Moreover, we did not find instability in either the detection or the estimation of heterogeneous effects in the Monte Carlo simulations. 
\review{The proposed algorithm could be extended to tree-based ensemble methods, following \cite{wager2018estimation} and \cite{athey2019generalized}. By averaging the estimates from many trees, this extension could enhance estimation precision, particularly in the presence of a highly smoothed outcome, at the cost of reduced interpretability (see \cite{lee2020causal} for a discussion on the trade-off between accuracy and interpretability).} In addition, our approach might be rearranged to deal with settings where the network cannot be partitioned into well-defined and pre-specified clusters. However, in some network structures, clusters could be detected by implementing a network-based community detection algorithm \citep{fortunato2010community}, and the NCT algorithm could be applied to the detected communities, while the estimator should take into account the uncertainty in group membership. 

Here we assume that the network exposure variable is discrete, with the performance of the algorithm being affected by the number of categories resulting from the exposure mapping function. In our simulation study and application, we have used a binary neighborhood exposure, which allowed us to grow deeper trees, reduce the number of possible causal effects, and have enough observations for each exposure condition to maintain the variance in a reasonable range. However, alternative specifications could be used. For instance, network exposure could be defined as the proportion of treated neighbors, perhaps categorized into a few bins. \review{
We acknowledge the possibility of a misspecified network exposure variable due to a misspecification of the exposure mapping function that would make Assumption \ref{ass:cni} hold. Even though it has been shown that, when the exposure mapping function is misspecified, the Horvitz-Thompson estimator would still be unbiased for the `individual average potential outcome' marginalized over the distribution of the true network exposure (or the whole treatment vector) given the misspecified network exposure \citep{aronow2017estimating,savje2023causal},  
the existence of some level of heterogeneity in the marginalizing distribution would give rise to spurious heterogeneity in the estimated causal effects. This is true whenever we simplify the definition of an exposure for the purpose of heterogeneity investigation. However, when the misspecification is on the network exposure, in the case of randomized experiments any heterogeneity in the conditional distribution of the true exposure would be due to the dependence of the true network exposure from network properties and a possible correlation between these and other covariates. Then, if we include network properties in our network causal tree, the only spurious heterogeneity would be found in these network properties. In the motivating application, we also explore variations in the definition of the exposure mapping by changing the threshold of the required number of treated neighbors required to determine an indirect exposure to the intervention.}
A more complex definition of network exposure, possibly resulting in a continuous variable, would require some methodological adjustments in the estimation strategy. We leave this extension to future work.

\pagebreak
\bibliography{references.bib}

\newpage
\appendix
{\bf \huge Online Appendix}
    \pagenumbering{arabic}
    \setcounter{page}{1}

\section{Proofs}
\label{app:proofs}

\begin{repproposition}{prop:unbiaseness}[Unbiasedness]
\begin{equation}
    E[\widehat{\mu}_{(w,g)}(\ell(\vx))]=\mu_{(w,g)}(\ell(\vx)). \nonumber
\end{equation}
\end{repproposition}

\begin{proof}
\begin{equation}
\begin{aligned}
\E[\widehat{\mu}_{(w,g)}(\ell(\vx)]&=\E\biggl[
\frac{1}{N(\ell(\vx))}\sum_{k=1}^K\sum_{i=1}^{n_k}\frac{Y_{ik}}{\pi_{ik}(w,g)}\mathds{1}(W_{ik}=w, G_{ik}=g, \mathbf{X}_{ik}\in \ell(\vx)) 
\biggr]\\
&=\frac{1}{N(\ell(\vx))}\sum_{k=1}^K\sum_{i=1}^{n_k}\E\big[\I(\Wik=w, \Gik=g, \mathbf{X}_{ik}\in \ell(\vx))\big]\frac{\Yik(w,g)}{\pi_{ik}(w,g)}\\
&=\frac{1}{N(\ell(\vx))}\sum_{k=1}^K\sum_{i=1}^{n_k}\E\big[\I(\Wik=w, \Gik=g)\big]\I(\mathbf{X}_{ik}\in \ell(\vx))\frac{\Yik(w,g)}{\pi_{ik}(w,g)}\\
&= \mu_{(w,g)(\ell(\vx)}
\end{aligned}
\end{equation}
where the expectation is over the randomization distribution of $\Wik$ and the induced distribution on $\Gik$, the second equality holds by Assumptions \ref{ass:cni}, \ref{ass:consistency},  and \ref{ass: bing} , the third equality is due to the randomized experiment such that the distribution of $\Wik$ and $\Gik$ depends neither on covariates nor on the potential outcome (Assumption \ref{ass:ita}), and the fourth equality is due to the fact that in a randomized experiment, and under Assumption \ref{ass: positivity}, $\pi_{ik}(w,g)$ is known and equal to  $\E\big[\I(\Wik=w, \Gik=g)\big]$.
\end{proof}


\begin{repproposition}{prop: consistency}[Consistency of the estimator]
Consider the asymptotic regime where the number of clusters K go to infinity, i.e., $K \longrightarrow \infty$, while the cluster size remains bounded, i.e., $n_k \leq B$ for some constant B. In addition, assume that $|Y_{ik}(w,g)|/\pi_{ik}(w,g)\leq C<1$, $\forall i,k, w,g$. Then as $K \longrightarrow \infty$
$$\widehat{\tau}_{(w,g;w',g')}(\ell(\vx)) \overset{p}{\longrightarrow} \tau_{(w,g;w',g')}(\ell(\vx)).$$
\end{repproposition}

\begin{proof}
As in Proposition \ref{prop:unbiaseness} $\widehat{\mu}_{(w,g))}(\ell(\vx))$ is unbiased. Hence, for consistency to hold we need to prove that the variance goes to 0 as N goes to infinity.

Following \cite{aronow2017estimating}, one can show that the variance of $\widehat{\mu}_{(w,g)}(\ell(\vx))$ is given by:
\begin{align}
    \Var\Big(\widehat{\mu}_{(w,g)}(\ell(\vx))\Big)=&\frac{1}{N(\ell(\vx))^2}\sum_{k=1}^K\sum_{i=1}^{n_k} \I(\vXik \in \ell(\vx))\pi_{ik}(w,g)[1-\pi_{ik}(w,g)]\biggl[\frac{\Yik(w,g)}{\pi_{ik}(w,g)}\biggr]^2 \nonumber \\&+ \frac{1}{N(\ell(\vx))^2}\sum_{k=1}^K\sum_{i=1}^{n_k}\sum_{j\neq i} \I(\vXik \in \ell(\vx),\mathbf{X}_{jk} \in \ell(\vx)) \nonumber \\ &\times[\pi_{ikjk}(w,g)-\pi_{ik}(w,g)\pi_{jk}(w,g])\frac{\Yik(w,g)}{\pi_{ik}(w,g)}\frac{\Yjk(w,g)}{\pi_{jk}(w,g)}
    \end{align}
Since $|Y_{ik}(w,g)|/\pi_{ik}(w,g)\leq C<1$ and given that in each cluster the sample size belonging to leaf $\ell(\vx)$ is bounded, i.e., $n_k(\ell(\vx))\leq B(\ell(\vx))\leq B$, we have:
$$  (K\times B)^2\Var\Big(\widehat{\mu}_{(w,g)}(\ell(\vx))\Big)\leq \,\, C^2\times K \times B+ C^2\times K \times B^2 
$$
Consistency of $\widehat{\mu}_{(w,g)}(\ell(\vx))$ is therefore ensured since $\Var\Big(\widehat{\mu}_{(w,g)}(\ell(\vx))\Big)\longrightarrow 0$ as $K \longrightarrow \infty$.
Consistency of $\widehat{\tau}_{(w,g;w',g')}(\ell(\vx))$ follows by Slutsky's Theorem.
\end{proof}

\begin{proposition}
The partition $\Pi^{\star}$ such that
$$\Pi^{\star}=\text{argmax}_{\Pi\in \mathds{P}}Q_{(w,g;w',g')}(\Pi)= \frac{1}{N}\sum_{k=1}^K\sum_{i=1}^{n_k} \big(\widehat{\tau}_{(w,g;w',g')}(\ell(\vXik, \Pi))\big)^2$$
maximizes the heterogeneity across leaves.
\end{proposition}

\begin{proof}
Let $\ell_1$ and $\ell_2$ be two sub-populations with a different causal effect $\tau_{(w,g;w',g')}$, i.e.,  $\tau_{(w,g;w',g')}(\ell_1)\neq\tau_{(w,g;w',g')}(\ell_2)$. Let $\Pi$ be a partition that splits $\ell_1$ and $\ell_2$ into two leaves and let $\Pi^c$ be the partition that combines the two sub-populations into one node $\ell_{1+2}$. Then we have that $Q_{(w,g;w',g')}(\Pi)>Q_{(w,g;w',g')}(\Pi^c)$.
The proofs follows from Jensen's inequality. In fact,
for partition $\Pi$ the splitting function can be written as follows:
\begin{align}
Q_{(w,g;w',g')}(\Pi)&=\frac{1}{|\ell_1|+|\ell_2|}\sum_{ik\in \ell_1\cup\ell_2} \big(\widehat{\tau}_{(w,g;w',g')}(\ell(\vXik, \Pi))\big)^2 \nonumber\\&=\frac{1}{|\ell_1|+|\ell_2|}\Big(\sum_{ik\in \ell_1} \big(\widehat{\tau}_{(w,g;w',g')}(\ell_1)\big)^2+\sum_{ik\in \ell_2} \big(\widehat{\tau}_{(w,g;w',g')}(\ell_2)\big)^2\Big) \nonumber
\end{align}
For partition $\Pi^c$ we have:
\begin{align}
Q_{(w,g;w',g')}(\Pi^c)&=\frac{1}{|\ell_1|+|\ell_2|}\sum_{ik\in \ell_1\cup\ell_2} \big(\widehat{\tau}_{(w,g;w',g')}(\ell(\vXik, \Pi^c))\big)^2=\frac{1}{|\ell_1|+|\ell_2|}\sum_{ik\in \ell_1\cup\ell_2} \big(\widehat{\tau}_{(w,g;w',g')}(\ell_{1+2})\big)^2\nonumber\\&=\frac{1}{|\ell_1|+|\ell_2|}\sum_{ik\in \ell_1\cup\ell_2} \Big[\frac{1}{|\ell_1|+|\ell_2|}\Big(\sum_{ik\in \ell_1} \widehat{\tau}_{(w,g;w',g')}(\ell_1)+\sum_{ik\in \ell_2} \widehat{\tau}_{(w,g;w',g')}(\ell_2)\Big)\Big]^2\nonumber\\&=\Big[\frac{1}{|\ell_1|+|\ell_2|}\Big(\sum_{ik\in \ell_1} \widehat{\tau}_{(w,g;w',g')}(\ell_1)+\sum_{ik\in \ell_2} \widehat{\tau}_{(w,g;w',g')}(\ell_2)\Big)\Big]^2 \nonumber
\end{align}
where the second-last equality holds because of the properties of the Horvitz-Thompson estimator. 
Thanks to Jensen's inequality
\begin{align*}\frac{1}{|\ell_1|+|\ell_2|}\Big(\sum_{ik\in \ell_1} \big(\widehat{\tau}_{(w,g;w',g')}(\ell_1)\big)^2+\sum_{ik\in \ell_2} \big(\widehat{\tau}_{(w,g;w',g')}(\ell_2)\big)^2\Big)\geq\\ \Big[\frac{1}{|\ell_1|+|\ell_2|}\Big(\sum_{ik\in \ell_1} \widehat{\tau}_{(w,g;w',g')}(\ell_1)+\sum_{ik\in \ell_2} \widehat{\tau}_{(w,g;w',g')}(\ell_2)\Big)\Big]^2 \end{align*}
Hence, $Q_{(w,g;w',g')}(\Pi) \geq Q_{(w,g;w',g')}(\Pi^c)$.
\end{proof}

\section{Further Details of the Variance Estimator of Leaf-Specific CACE}
\label{appendix_varcace}
If in a generic leaf $\ell(\vx)$ there are some pairs of units ($i$,$j$) whose joint probability of the exposure condition $(w,g)$ is zero (namely, $\pi_{ikjk}(w,g)=0$), the variance of $\widehat{\mu_{(w,g)}}(\ell(\vx))$ can be estimated following a result from \cite{aronow2017estimating}. Such estimator, denoted by $\widehat{\mathbb{V}}^{c}\Big(\widehat{\mu}_{(w,g)}(\ell(\vx))\Big)$, is the sum of two components: (i) the estimated variance of leaf-specific potential outcomes, $\widehat{\Var}\Big(\widehat{\mu}_{(w,g)}(\ell(\vx)\Big))$ in \eqref{eq:var} for the case when $\pi_{ikjk}(w,g)>0 \,\, \forall  i,j,k$, and (ii) a correction term $\widehat{A}_{(w,g)}(\ell(\vx))$:
\begin{equation}
\widehat{\mathbb{V}}^{c}\Big(\widehat{\mu}_{(w,g)}(\ell(\vx))\Big)=\widehat{\Var}\Big(\widehat{\mu}_{(w,g)}(\ell(\vx))\Big) + \widehat{A}_{(w,g)}(\ell(\vx))
\end{equation}
where
\begin{eqnarray}
\widehat{A}_{(w,g)}(\ell(\vx))&=&\frac{1}{N(\ell(\vx))^2}\sum_{k=1}^K\sum_{i=1}^{n_k} \sum_{j \ne i: \pi_{ikjk}(w,g)=0}\Bigg[ \frac{\mathds{1}(W_{ik}=w, G_{ik}=g,\vXik \in \ell(\vx))Y^{2}_{ik}}{2\pi_{ik}(w,g)} + \nonumber \\
&+& \frac{\mathds{1}(W_{jk}=w, G_{jk}=g, \vXik \in \ell(\vx))Y^{2}_{jk}}{2\pi_{jk}(w,g)} \Bigg].
\end{eqnarray}
Note that the correction term $\widehat{A}_{(w,g)}(\ell(\vx))$ is zero if the leaf does not have pairs of units such that $\pi_{ikjk}(w,g,w,g)=0$. Furthermore, as in \cite{aronow2017estimating},  $\widehat{\mathbb{V}}^{c}\Big(\widehat{\mu}_{(w,g)}(\ell(\vx))\Big)$ is a conservative estimator of the leaf-specific variance, as the following holds:
\begin{equation}
\mathbb{E}\Bigg[\widehat{\mathbb{V}}^{c}\Big(\widehat{\mu}_{(w,g)}(\ell(\vx))\Big) \Bigg]=\mathbb{E}\Bigg[\widehat{\mathbb{V}}\Big(\widehat{\mu}_{(w,g)}(\ell(\vx))\Big) + \widehat{A}_{(w,g)}(\ell(\vx))\Bigg]\ge \mathbb{V}\Big(\widehat{\mu}_{(w,g)}(\ell(\vx))\Big).   
\end{equation}
We now explicit the covariance $\widehat{\mathbb{C}}^{c}\Big(\widehat{\mu}_{(w,g)}(\ell(\vx)),\widehat{\mu}_{(w',g')}(\ell(\vx)) \Big)$ in the case we have pairs of units $(i,j)$, whose joint probability of experiencing the conditions $(w,g)$ and $(w',g')$, respectively, is zero, that is, $\pi_{ikjk}(w,g;w',g')=0$:

\begin{small}
\begin{eqnarray}
\widehat{\mathbb{C}}\Big(\widehat{\mu}_{w,g}(\ell(\vx)),\widehat{\mu}_{w',g'}(\ell(\vx))\Big)&=& \frac{1}{N(\ell(\vx))^2}
\sum_{k=1}^K\sum_{i=1}^{n_k}\sum_{j \neq i: \pi_{ikjk}(w,g;w',g')>0} \nonumber\\
&&\frac{\I(\Wik=w,\Gik=g, \mathbf{X}_{ik} \in \ell(\vx))\I(\Wjk=w',\Gjk=g', \mathbf{X}_{jk} \in \ell(\vx))}{\pi_{ikjk}(w,g;w',g')} \nonumber\\
&\times&[\pi_{ikjk}(w,g;w',g')-\pi_{ik}(w,g)\pi_{jk}(w',g')]\frac{\Yik}{\pi_{ik}(w,g)}\frac{\Yjk}{\pi_{jk}(w',g')} \nonumber  \\
 &-& \frac{1}{N(\ell(\vx))^2}\sum_{k=1}^K\sum_{i=1}^{n_k}\sum_{j \neq i: \pi_{ikjk}(w,g;w',g')=0}\Big[\frac{\I(\Wik=w,\Gik=g, \mathbf{X}_{ik} \in \ell(\vx))\Yik^2}{2\pi_{ik}(w,g)} \nonumber \\ &+&\frac{\I(\Wik=w',\Gik=g', \mathbf{X}_{ik} \in \ell(\vx))\Yik^2}{2\pi_{ik}(w',g')}\Big]. 
\end{eqnarray}
\end{small}

\movedtext{
\section{Performance Measures for Monte Carlo Simulations}\label{sec:performance_measures}
In all the scenarios the performance of the NCT algorithm is evaluated using the following measures, averaged over the $M$ generated data sets: 
\begin{itemize}
\item Average number of correctly discovered heterogeneous causal rules corresponding to the leaves of the generated NCTs (reported with respect to the effect sizes in Figures \ref{fig:correct_leaves} and \ref{fig:correct_leaves_4_rules}) in the discovery sample;
\item Average conditional treatment and spillover effects estimated for each correctly detected heterogeneous terminal leaf  (reported as $\hat{\tau}$ and $\hat{\delta}$ in Tables \ref{tab:10clusters}, \ref{tab:20clusters}, \ref{tab:30clusters} and \ref{tab:second_scenario});
\item Average standard errors estimated for each correctly detected heterogeneous terminal leaf (reported as $\hat{se}(\hat{\tau})$ and $\hat{se}(\hat{\delta})$ in Tables \ref{tab:10clusters}, \ref{tab:20clusters}, \ref{tab:30clusters} and \ref{tab:second_scenario});
\item Monte Carlo bias in the estimation sample:
$$\text{Bias}_m(V^{est})=
\frac{1}{N^{\text{est}}}\sum_{k\in \mathcal{K}^{est}}\sum_{i=1}^{n_k}\Big(\tau_{(w,g,w'g')}(\vXik)-\widehat{\tau}_{(w,g,w'g')}(\ell(\vXik,, \Pi_m),V^{est})\Big),
$$
$$
\text{Bias}(V^{est})=\frac{1}{M}\sum_{m=1}^M \text{Bias}_m(V^{est})
$$
where $\Pi_m$ is the partition selected in simulation $m$;
\item  Monte Carlo MSE in the estimation sample:
$$
\text{MSE}_m(V^{est})=
\frac{1}{N^{\text{est}}}\sum_{k\in \mathcal{K}^{est}}\sum_{i=1}^{n_k}\Big(\tau_{(w,g,w'g')}(\vXik)-\widehat{\tau}_{(w,g,w'g')}(\ell(\vXik), \Pi_m,V^{est})\Big)^2,
$$
$$
\text{MSE}(V^{est})=\frac{1}{M}\sum_{m=1}^M \text{MSE}_m(V^{est});
$$
\item  Coverage, computed as the average proportion of units for whom where the estimated $95\%$ confidence interval of the causal effect in the assigned leaf  includes the true value:
$$\text{C}_m(V^{est})=
\frac{1}{N^{\text{est}}}\sum_{k\in \mathcal{K}^{est}}\sum_{i=1}^{n_k}\I\Big(\tau_{(w,g,w'g')}(\vXik) \in \widehat{\text{CI}}_{95}\Big(\widehat{\tau}_{(w,g,w'g')}(\ell(\vXik, \Pi_m),V^{est})\Big) \Big),
$$
$$
\text{C}(V^{est})=\frac{1}{M}\sum_{m=1}^M \text{C}_m(V^{est}).
$$
\end{itemize}
}

\section{Estimation Results for Monte Carlo Simulations}


\vspace{2cm}
\begin{table}[H]
\centering
\scriptsize
\begin{tabular}{cccccccc}
\multicolumn{1}{l}{}                  & \multicolumn{7}{c}{Treatment Effects}                                                                                                                 \\ \cline{2-8} 
\multicolumn{1}{l}{Effect Size}       & $\hat{\tau}_{\ell_1}$   & $\hat{se}(\hat{\tau}_{\ell_1})$   & $\hat{\tau}_{\ell_2}$   & $\hat{se}(\hat{\tau}_{\ell_2})$   & MSE   & Bias   & Coverage \\ \hline
0                                     & 0.373                   & 0.321                             & -0.564                  & 0.299                             & 0.305 & 0.452  & 0.667    \\
\rowcolor{gray!10} 1 & 1.096                   & 0.365                             & -1.131                  & 0.371                             & 0.134 & 0.113  & 0.969    \\
2                                     & 2.022                   & 0.513                             & -2.014                  & 0.511                             & 0.239 & 0.018  & 0.945    \\
\rowcolor{gray!10} 3 & 3.027                   & 0.696                             & -3.001                  & 0.703                             & 0.463 & 0.014  & 0.951    \\
4                                     & 4.071                   & 0.905                             & -3.961                  & 0.885                             & 0.737 & 0.016  & 0.947    \\
\rowcolor{gray!10} 5 & 4.961                   & 1.091                             & -4.977                  & 1.085                             & 1.131 & -0.031 & 0.947    \\
6                                     & 6.005                   & 1.299                             & -5.951                  & 1.303                             & 1.476 & -0.022 & 0.952    \\
\rowcolor{gray!10} 7 & 7.023                   & 1.504                             & -7.099                  & 1.521                             & 2.006 & 0.061  & 0.954    \\
8                                     & 8.066                   & 1.707                             & -7.979                  & 1.693                             & 2.511 & 0.022  & 0.952    \\
\rowcolor{gray!10} 9 & 8.929                   & 1.903                             & -9.119                  & 1.911                             & 3.157 & 0.024  & 0.960    \\
10                                    & 10.126                  & 2.167                             & -10.082                 & 2.146                             & 4.164 & 0.104  & 0.958    \\ \hline
\multicolumn{1}{l}{}                  & \multicolumn{7}{c}{Spillover Effects}                                                                                                                 \\ \cline{2-8} 
\multicolumn{1}{l}{}                  & $\hat{\delta}_{\ell_1}$ & $\hat{se}(\hat{\delta}_{\ell_1})$ & $\hat{\delta}_{\ell_2}$ & $\hat{se}(\hat{\delta}_{\ell_2})$ & MSE   & Bias   & Coverage \\ \hline
0                                     & 0.408                   & 0.299                             & -0.410                  & 0.260                             & 0.232 & 0.409  & 0.750    \\
\rowcolor{gray!10} 1 & 1.079                   & 0.289                             & -1.090                  & 0.292                             & 0.074 & 0.085  & 0.971    \\
2                                     & 2.004                   & 0.380                             & -2.012                  & 0.379                             & 0.118 & 0.008  & 0.966    \\
\rowcolor{gray!10} 3 & 3.041                   & 0.503                             & -2.992                  & 0.494                             & 0.189 & 0.016  & 0.968    \\
4                                     & 3.996                   & 0.629                             & -3.998                  & 0.627                             & 0.290 & -0.003 & 0.975    \\
\rowcolor{gray!10} 5 & 5.005                   & 0.757                             & -5.015                  & 0.765                             & 0.439 & 0.010  & 0.967    \\
6                                     & 6.018                   & 0.891                             & -5.966                  & 0.892                             & 0.571 & -0.008 & 0.977    \\
\rowcolor{gray!10}7  & 6.993                   & 1.026                             & -7.046                  & 1.031                             & 0.770 & 0.020  & 0.968    \\
8                                     & 7.929                   & 1.159                             & -7.980                  & 1.167                             & 0.999 & -0.045 & 0.974    \\
\rowcolor{gray!10} 9 & 9.082                   & 1.313                             & -8.920                  & 1.300                             & 1.225 & 0.001  & 0.971    \\
10                                    & 10.050                  & 1.449                             & -9.981                  & 1.441                             & 1.495 & 0.015  & 0.975    \\ \hline
\end{tabular}
\caption{Simulations' results for the first scenario (10 clusters)}
\label{tab:10clusters}
\end{table}


\vspace{2cm}
\begin{table}[H]
\centering
\scriptsize
\begin{tabular}{cccccccc}
\multicolumn{1}{l}{}                  & \multicolumn{7}{c}{Treatment Effects}                                                                                                                 \\ \cline{2-8} 
\multicolumn{1}{l}{Effect Size}       & $\hat{\tau}_{\ell_1}$   & $\hat{se}(\hat{\tau}_{\ell_1})$   & $\hat{\tau}_{\ell_2}$   & $\hat{se}(\hat{\tau}_{\ell_2})$   & MSE   & Bias   & Coverage \\ \hline
0                                     & 0.311                   & 0.210                             & -0.213                  & 0.213                             & 0.089 & 0.247  & 0.885    \\
\rowcolor{gray!10} 1 & 1.054                   & 0.259                             & -1.057                  & 0.261                             & 0.069 & 0.055  & 0.961    \\
2                                     & 1.998                   & 0.368                             & -2.006                  & 0.366                             & 0.136 & 0.002  & 0.941    \\
\rowcolor{gray!10} 3 & 2.993                   & 0.495                             & -3.040                  & 0.503                             & 0.239 & 0.016  & 0.953    \\
4                                     & 4.034                   & 0.643                             & -4.016                  & 0.645                             & 0.372 & 0.025  & 0.964    \\
\rowcolor{gray!10} 5 & 5.007                   & 0.784                             & -4.995                  & 0.774                             & 0.579 & 0.001  & 0.950    \\
6                                     & 6.072                   & 0.936                             & -6.027                  & 0.930                             & 0.798 & 0.050  & 0.945    \\
\rowcolor{gray!10} 7 & 6.977                   & 1.074                             & -7.007                  & 1.086                             & 1.109 & -0.008 & 0.938    \\
8                                     & 7.939                   & 1.213                             & -7.955                  & 1.213                             & 1.237 & -0.053 & 0.949    \\
\rowcolor{gray!10} 9 & 8.967                   & 1.362                             & -9.019                  & 1.376                             & 1.798 & -0.007 & 0.943    \\
10                                    & 9.974                   & 1.540                             & -9.947                  & 1.502                             & 2.108 & -0.039 & 0.959    \\ \hline
\multicolumn{1}{l}{}                  & \multicolumn{7}{c}{Spillover Effects}                                                                                                                 \\ \cline{2-8} 
\multicolumn{1}{l}{}                  & $\hat{\delta}_{\ell_1}$ & $\hat{se}(\hat{\delta}_{\ell_1})$ & $\hat{\delta}_{\ell_2}$ & $\hat{se}(\hat{\delta}_{\ell_2})$ & MSE   & Bias   & Coverage \\ \hline
0                                     & 0.255                   & 0.184                             & -0.108                  & 0.173                             & 0.043 & 0.164  & 0.923    \\
\rowcolor{gray!10} 1 & 1.050                   & 0.207                             & -1.036                  & 0.208                             & 0.037 & 0.044  & 0.966    \\
2                                     & 2.001                   & 0.273                             & -2.003                  & 0.272                             & 0.057 & 0.002  & 0.969    \\
\rowcolor{gray!10} 3 & 2.999                   & 0.359                             & -3.013                  & 0.357                             & 0.098 & 0.006  & 0.976    \\
4                                     & 3.998                   & 0.448                             & -4.004                  & 0.449                             & 0.156 & 0.001  & 0.972    \\
\rowcolor{gray!10} 5 & 4.996                   & 0.543                             & -5.014                  & 0.542                             & 0.191 & 0.005  & 0.985    \\
6                                     & 5.985                   & 0.641                             & -6.001                  & 0.642                             & 0.269 & -0.007 & 0.985    \\
\rowcolor{gray!10}7  & 6.963                   & 0.736                             & -7.011                  & 0.740                             & 0.363 & -0.013 & 0.977    \\
8                                     & 8.030                   & 0.839                             & -7.997                  & 0.837                             & 0.503 & 0.013  & 0.983    \\
\rowcolor{gray!10} 9 & 9.015                   & 0.938                             & -9.021                  & 0.938                             & 0.606 & 0.018  & 0.980    \\
10                                    & 10.060                  & 1.042                             & -9.985                  & 1.038                             & 0.688 & 0.022  & 0.984    \\ \hline
\end{tabular}
\caption{Simulations' results for the first scenario (20 clusters)}
\label{tab:20clusters}
\end{table}

\section{Additional Monte Carlo Simulations} \label{appendix:monte_carlo}

We included in the simulation study two additional sets of simulations where (4) we introduce correlation between the covariates, and (5) we replace the Erd\H{o}s-R\'{e}nyi model for network formation with an exponential random graph (ERGM) model introducing homophily within the clusters. These two additional simulations are conducted with 30 clusters and under the first scenario introduced in Section \ref{sec:simulations}, where the heterogeneity is the same for the two causal effects of interest.

\subsection{Correlated Covariates}

In Figure \ref{fig:correlated} we report the number of correctly detected leaves under low and high correlation (0.25 and 0.5), while in Tables \ref{tab:correlated_0.25} and \ref{tab:correlated_0.50} we report the estimated treatment and spillover effects with their standard error in the two heterogeneous leaves, together with the MSE, bias and coverage of the average treatment and spillover effects in the sample. 
From Figure \ref{fig:correlated} one can see that the correlation between covariates compromises the ability of the algorithm to correctly identify the heterogeneous subgroups. This is due to the fact that, as the covariates become more \textit{similar} to each other, it becomes harder for the algorithm to detect the true HDV. Such a problem is common to all tree-based algorithms. Hence, we argue that one should carefully check the correlation patterns between the variables to get a sense of the reliability of the discovered subgroups.

\begin{figure}[H]
\centering
\includegraphics[width=1\textwidth]{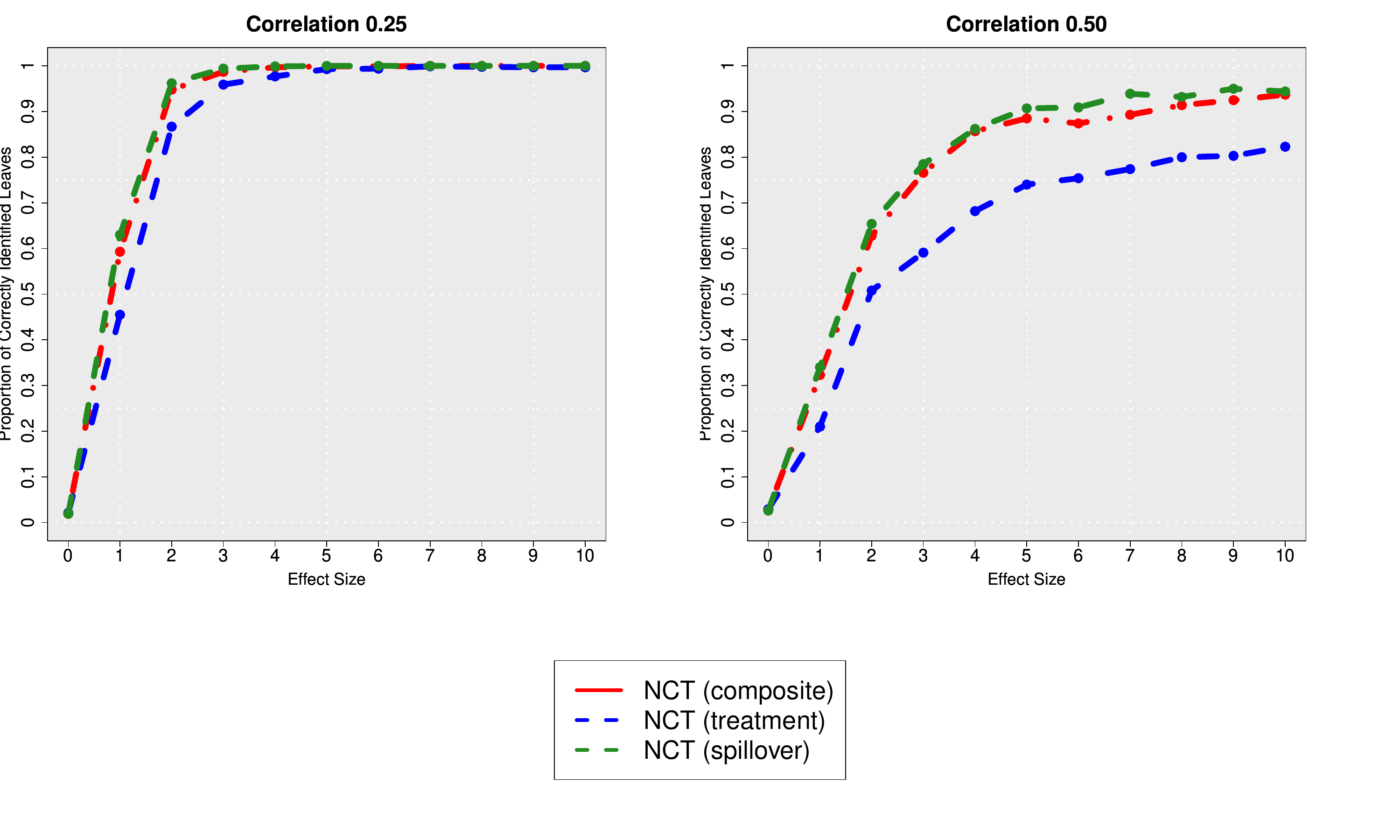}
		\caption{Simulations' results for correctly discovered leaves in the first scenario with correlated covariates}
	\label{fig:correlated}
\end{figure}


\begin{table}[H]
\scriptsize
\centering

\begin{tabular}{cccccccc}
\multicolumn{1}{l}{}                  & \multicolumn{7}{c}{Treatment Effects}                                                                                                                 \\ \cline{2-8} 
\multicolumn{1}{l}{Effect Size}       & $\hat{\tau}_{\ell_1}$   & $\hat{se}(\hat{\tau}_{\ell_1})$   & $\hat{\tau}_{\ell_2}$   & $\hat{se}(\hat{\tau}_{\ell_2})$   & MSE   & Bias   & Coverage \\ \hline
0                                     & -0.103                  & 0.163                             & -0.058                  & 0.155                             & 0.031 & -0.081 & 1.000    \\
\rowcolor{gray!10} 1 & 1.028                   & 0.195                             & -1.027                  & 0.201                             & 0.037 & 0.001  & 0.960    \\
2                                     & 1.997                   & 0.278                             & -2.008                  & 0.281                             & 0.072 & -0.005 & 0.963    \\
\rowcolor{gray!10} 3 & 2.994                   & 0.379                             & -3.007                  & 0.380                             & 0.128 & -0.007 & 0.963    \\
4                                     & 3.990                   & 0.487                             & -3.985                  & 0.482                             & 0.191 & 0.003  & 0.968    \\
\rowcolor{gray!10} 5 & 5.026                   & 0.594                             & -5.001                  & 0.594                             & 0.342 & 0.012  & 0.954    \\
6                                     & 6.057                   & 0.709                             & -6.018                  & 0.706                             & 0.425 & 0.020  & 0.963    \\
\rowcolor{gray!10} 7 & 6.966                   & 0.818                             & -6.991                  & 0.816                             & 0.657 & -0.012 & 0.942    \\
8                                     & 7.984                   & 0.929                             & -7.981                  & 0.925                             & 0.703 & 0.001  & 0.960    \\
\rowcolor{gray!10} 9 & 8.935                   & 1.027                             & -9.014                  & 1.047                             & 0.904 & -0.040 & 0.959    \\
10                                    & 9.995                   & 1.157                             & -10.070                 & 1.179                             & 1.225 & -0.037 & 0.960    \\ \hline
\multicolumn{1}{l}{}                  & \multicolumn{7}{c}{Spillover Effects}                                                                                                                 \\ \cline{2-8} 
\multicolumn{1}{l}{}                  & $\hat{\delta}_{\ell_1}$ & $\hat{se}(\hat{\delta}_{\ell_1})$ & $\hat{\delta}_{\ell_2}$ & $\hat{se}(\hat{\delta}_{\ell_2})$ & MSE   & Bias   & Coverage \\ \hline
0                                     & -0.060                  & 0.140                             & -0.081                  & 0.138                             & 0.028 & -0.071 & 0.947    \\
\rowcolor{gray!10} 1 & 1.027                   & 0.158                             & -1.023                  & 0.161                             & 0.023 & 0.002  & 0.966    \\
2                                     & 1.996                   & 0.207                             & -2.014                  & 0.208                             & 0.034 & -0.009 & 0.978    \\
\rowcolor{gray!10} 3 & 3.002                   & 0.272                             & -3.004                  & 0.273                             & 0.055 & -0.001 & 0.982    \\
4                                     & 4.019                   & 0.343                             & -4.005                  & 0.342                             & 0.081 & 0.007  & 0.986    \\
\rowcolor{gray!10} 5 & 4.978                   & 0.413                             & -4.989                  & 0.414                             & 0.121 & -0.005 & 0.978    \\
6                                     & 6.002                   & 0.489                             & -5.974                  & 0.488                             & 0.165 & 0.014  & 0.980    \\
\rowcolor{gray!10}7  & 6.966                   & 0.562                             & -7.014                  & 0.563                             & 0.225 & -0.024 & 0.978    \\
8                                     & 7.971                   & 0.638                             & -7.954                  & 0.639                             & 0.291 & 0.009  & 0.975    \\
\rowcolor{gray!10} 9 & 9.021                   & 0.716                             & -8.992                  & 0.716                             & 0.349 & 0.014  & 0.982    \\
10                                    & 10.000                  & 0.792                             & -9.986                  & 0.790                             & 0.425 & 0.007  & 0.979    \\ \hline
\end{tabular}

\caption{Simulations' results for the first scenario with correlated covariates (0.25)}
\label{tab:correlated_0.25}
\end{table}


\begin{table}[H]
\scriptsize
\centering

\begin{tabular}{cccccccc}
\multicolumn{1}{l}{}                  & \multicolumn{7}{c}{Treatment Effects}                                                                                                                 \\ \cline{2-8} 
\multicolumn{1}{l}{Effect Size}       & $\hat{\tau}_{\ell_1}$   & $\hat{se}(\hat{\tau}_{\ell_1})$   & $\hat{\tau}_{\ell_2}$   & $\hat{se}(\hat{\tau}_{\ell_2})$   & MSE   & Bias   & Coverage \\ \hline
0                                     & 0.114                   & 0.158                             & -0.100                  & 0.157                             & 0.021 & 0.007  & 0.964    \\
\rowcolor{gray!10} 1 & 1.105                   & 0.189                             & -1.100                  & 0.188                             & 0.042 & 0.003  & 0.941    \\
2                                     & 2.124                   & 0.266                             & -2.125                  & 0.275                             & 0.066 & -0.001 & 0.960    \\
\rowcolor{gray!10} 3 & 3.099                   & 0.365                             & -3.136                  & 0.366                             & 0.112 & -0.018 & 0.958    \\
4                                     & 4.098                   & 0.465                             & -4.091                  & 0.458                             & 0.187 & 0.003  & 0.958    \\
\rowcolor{gray!10} 5 & 5.099                   & 0.565                             & -5.142                  & 0.567                             & 0.297 & -0.022 & 0.948    \\
6                                     & 6.119                   & 0.669                             & -6.124                  & 0.678                             & 0.427 & -0.003 & 0.963    \\
\rowcolor{gray!10} 7 & 7.127                   & 0.775                             & -7.142                  & 0.779                             & 0.490 & -0.007 & 0.961    \\
8                                     & 8.128                   & 0.879                             & -8.071                  & 0.877                             & 0.702 & 0.029  & 0.956    \\
\rowcolor{gray!10} 9 & 9.089                   & 0.976                             & -9.109                  & 0.980                             & 0.840 & -0.010 & 0.951    \\
10                                    & 10.030                  & 1.081                             & -10.146                 & 1.086                             & 1.072 & -0.058 & 0.956    \\ \hline
\multicolumn{1}{l}{}                  & \multicolumn{7}{c}{Spillover Effects}                                                                                                                 \\ \cline{2-8} 
\multicolumn{1}{l}{}                  & $\hat{\delta}_{\ell_1}$ & $\hat{se}(\hat{\delta}_{\ell_1})$ & $\hat{\delta}_{\ell_2}$ & $\hat{se}(\hat{\delta}_{\ell_2})$ & MSE   & Bias   & Coverage \\ \hline
0                                     & 0.072                   & 0.138                             & -0.130                  & 0.133                             & 0.016 & -0.029 & 0.964    \\
\rowcolor{gray!10} 1 & 1.116                   & 0.151                             & -1.100                  & 0.151                             & 0.022 & 0.008  & 0.975    \\
2                                     & 2.121                   & 0.200                             & -2.113                  & 0.200                             & 0.031 & 0.004  & 0.982    \\
\rowcolor{gray!10} 3 & 3.080                   & 0.261                             & -3.113                  & 0.260                             & 0.050 & -0.016 & 0.979    \\
4                                     & 4.110                   & 0.327                             & -4.104                  & 0.325                             & 0.067 & 0.003  & 0.986    \\
\rowcolor{gray!10} 5 & 5.107                   & 0.394                             & -5.117                  & 0.395                             & 0.109 & -0.005 & 0.977    \\
6                                     & 6.101                   & 0.463                             & -6.129                  & 0.466                             & 0.145 & -0.014 & 0.983    \\
\rowcolor{gray!10}7  & 7.143                   & 0.536                             & -7.091                  & 0.534                             & 0.193 & 0.026  & 0.981    \\
8                                     & 8.068                   & 0.604                             & -8.121                  & 0.605                             & 0.232 & -0.027 & 0.975    \\
\rowcolor{gray!10} 9 & 9.095                   & 0.677                             & -9.109                  & 0.676                             & 0.292 & -0.007 & 0.987    \\
10                                    & 10.114                  & 0.747                             & -10.033                 & 0.744                             & 0.352 & 0.041  & 0.987    \\ \hline
\end{tabular}

\caption{Simulations' results for the first scenario with correlated covariates (0.50)}
\label{tab:correlated_0.50}
\end{table}

Nevertheless, for both correlation levels (0.25 and 0.50) the estimator seems to perform well within correctly detected leaves (see Tables \ref{tab:correlated_0.25} and \ref{tab:correlated_0.50}).

\subsection{Network Homophily Within the Clusters}

Table \ref{tab:homophily} shows the results in the case of network homophily within the clusters. In this case, we find larger standard errors than the original scenario reported in \ref{tab:30clusters} without homophily. As a consequence, the Monte Carlo MSE is also slightly larger. 

Moreover, as we can see from Figure \ref{fig:homophily} there is a decrease in the ability of the algorithm to discover the true leaves. Indeed, if one compares this Figure with the right panel of Figure \ref{fig:correct_leaves}, one can see how the correct discovery of the true leaves is \textit{slower} in the case with homophily network. This is due to the fact that the standard errors are larger than in the original first scenario.

\begin{figure}[H]
\centering
\includegraphics[width=0.5\textwidth]{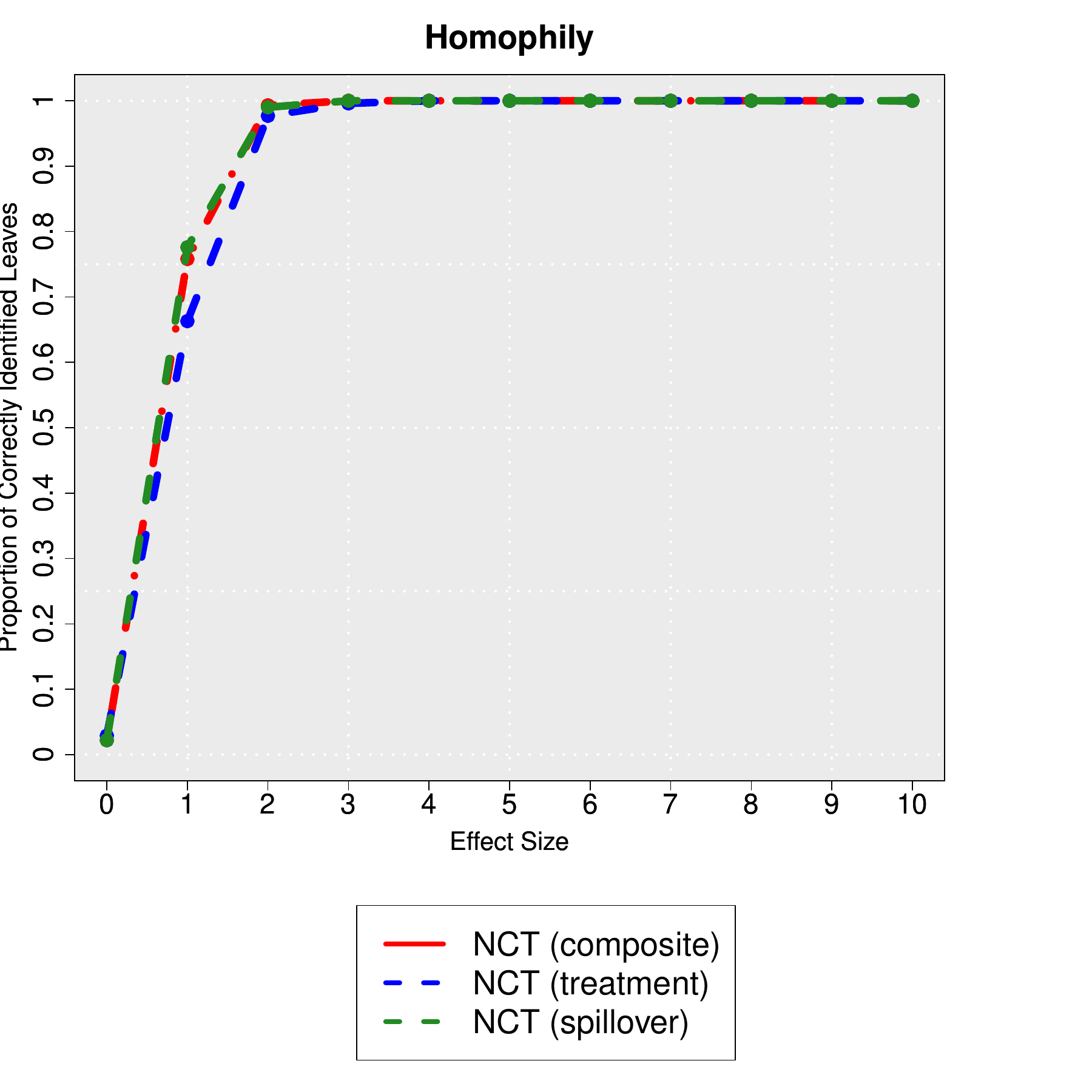}
		\caption{Simulations' results for correctly discovered leaves in the first scenario with homophily network}
	\label{fig:homophily}
\end{figure}


\begin{table}[H]
\scriptsize
\centering

\begin{tabular}{cccccccc}
\multicolumn{1}{l}{}                  & \multicolumn{7}{c}{Treatment Effects}                                                                                                                 \\ \cline{2-8} 
\multicolumn{1}{l}{Effect Size}       & $\hat{\tau}_{\ell_1}$   & $\hat{se}(\hat{\tau}_{\ell_1})$   & $\hat{\tau}_{\ell_2}$   & $\hat{se}(\hat{\tau}_{\ell_2})$   & MSE   & Bias   & Coverage \\ \hline
0                                     & -0.047                  & 0.174                             & 0.100                   & 0.183                             & 0.072 & 0.027  & 0.773    \\
\rowcolor{gray!10} 1 & 1.036                   & 0.219                             & -1.035                  & 0.216                             & 0.044 & 0.000  & 0.959    \\
2                                     & 2.007                   & 0.309                             & -1.993                  & 0.309                             & 0.077 & 0.007  & 0.971    \\
\rowcolor{gray!10} 3 & 2.999                   & 0.418                             & -3.013                  & 0.421                             & 0.151 & -0.007 & 0.972    \\
4                                     & 4.009                   & 0.533                             & -4.024                  & 0.537                             & 0.236 & -0.007 & 0.972    \\
\rowcolor{gray!10} 5 & 4.974                   & 0.656                             & -4.993                  & 0.660                             & 0.343 & -0.009 & 0.964    \\
6                                     & 6.067                   & 0.781                             & -6.017                  & 0.781                             & 0.490 & 0.025  & 0.975    \\
\rowcolor{gray!10} 7 & 7.000                   & 0.899                             & -7.004                  & 0.907                             & 0.706 & -0.002 & 0.969    \\
8                                     & 8.017                   & 1.026                             & -7.965                  & 1.020                             & 0.830 & 0.026  & 0.968    \\
\rowcolor{gray!10} 9 & 8.977                   & 1.145                             & -9.004                  & 1.158                             & 1.030 & -0.014 & 0.965    \\
10                                    & 9.953                   & 1.271                             & -10.058                 & 1.273                             & 1.267 & -0.052 & 0.966    \\ \hline
\multicolumn{1}{l}{}                  & \multicolumn{7}{c}{Spillover Effects}                                                                                                                 \\ \cline{2-8} 
\multicolumn{1}{l}{}                  & $\hat{\delta}_{\ell_1}$ & $\hat{se}(\hat{\delta}_{\ell_1})$ & $\hat{\delta}_{\ell_2}$ & $\hat{se}(\hat{\delta}_{\ell_2})$ & MSE   & Bias   & Coverage \\ \hline
0                                     & 0.061                   & 0.156                             & 0.065                   & 0.166                             & 0.035 & 0.063  & 0.909    \\
\rowcolor{gray!10} 1 & 1.027                   & 0.187                             & -1.039                  & 0.187                             & 0.029 & -0.006 & 0.964    \\
2                                     & 1.997                   & 0.253                             & -1.993                  & 0.253                             & 0.047 & 0.002  & 0.981    \\
\rowcolor{gray!10} 3 & 2.995                   & 0.335                             & -3.021                  & 0.337                             & 0.094 & -0.013 & 0.965    \\
4                                     & 4.018                   & 0.425                             & -3.979                  & 0.424                             & 0.132 & 0.020  & 0.979    \\
\rowcolor{gray!10} 5 & 4.952                   & 0.515                             & -4.989                  & 0.516                             & 0.196 & -0.019 & 0.968    \\
6                                     & 6.015                   & 0.610                             & -6.024                  & 0.613                             & 0.285 & -0.005 & 0.970    \\
\rowcolor{gray!10}7  & 7.001                   & 0.707                             & -7.030                  & 0.707                             & 0.341 & -0.015 & 0.984    \\
8                                     & 7.995                   & 0.803                             & -7.999                  & 0.802                             & 0.440 & -0.002 & 0.980    \\
\rowcolor{gray!10} 9 & 9.011                   & 0.899                             & -8.994                  & 0.897                             & 0.548 & 0.009  & 0.983    \\
10                                    & 10.049                  & 0.996                             & -9.922                  & 0.991                             & 0.748 & 0.064  & 0.974    \\ \hline
\end{tabular}

\caption{Simulations' results for the first scenario with network homophily  within the clusters (30 clusters)}
\label{tab:homophily}
\end{table}

\subsection{Mixture of Different Types of Covariates}

\review{Tree-based algorithms are very appealing due to their ability to handle diverse data types, encompassing continuous, categorical, ordinal, and binary variables. This unique characteristic eliminates the need for data transformations. In particular, trees excel in accommodating continuous predictors effortlessly, obviating the necessity of generating dummy variables. However, one of the known pitfalls of the original CART algorithm \citep{friedman1984classification} is its tendency towards selecting continuous covariates over categorical and binary ones \cite[see][for a discussion]{loh1997split,loh2002regression, hothorn2006unbiased}. As detailed in Section \ref{sec: ctrees}, we propose a different criterion function for recursive binary splitting. It is, thus, important to test whether our proposed NCT algorithm suffers from the same possible pitfall of CART. 

To test, we construct a simulation scenario in which we have a mixture of continuous and binary covariates. We keep the same structure of the first scenario in Section \ref{sec:simulations} with 30 clusters. Out of the 10 covariates $X_{ip}$, the two HDV are sampled from Bernoulli distributions with probability 0.5: $X_{i1}, X_{i2}  \sim Ber(0.5)$; the other 8 covariates were instead sampled from a multivariate normal distribution $(X_{i3},..., X_{i10}) \sim \mathcal{N}(0, \Upsigma)$, where $\Upsigma$ is an $8 \times 8$ identity matrix. 

Figure \ref{fig:mixture} reports the results for the proportion of correctly identified leaves. Table \ref{tab:mixture} reports the results for the estimation results. This figure, vis à vis with Figure \ref{fig:correct_leaves} in the main text, shows a slight decrease in the rate of convergence towards the full discovery of all the correct leaves. This is paired with an increase in the MSE (as compared to Table \ref{tab:30clusters} in the main text). However, while the overall performance slightly deteriorates, NCT seems to be robust to the pitfalls of CART in the presence of a mixture of different types of covariates. NCT is, indeed, able to discover the correct, binary HDVs, while depicting a robust performance in estimation (with unbiased results and nominal coverage rate).}

\begin{figure}[H]
\centering
\includegraphics[width=0.5\textwidth]{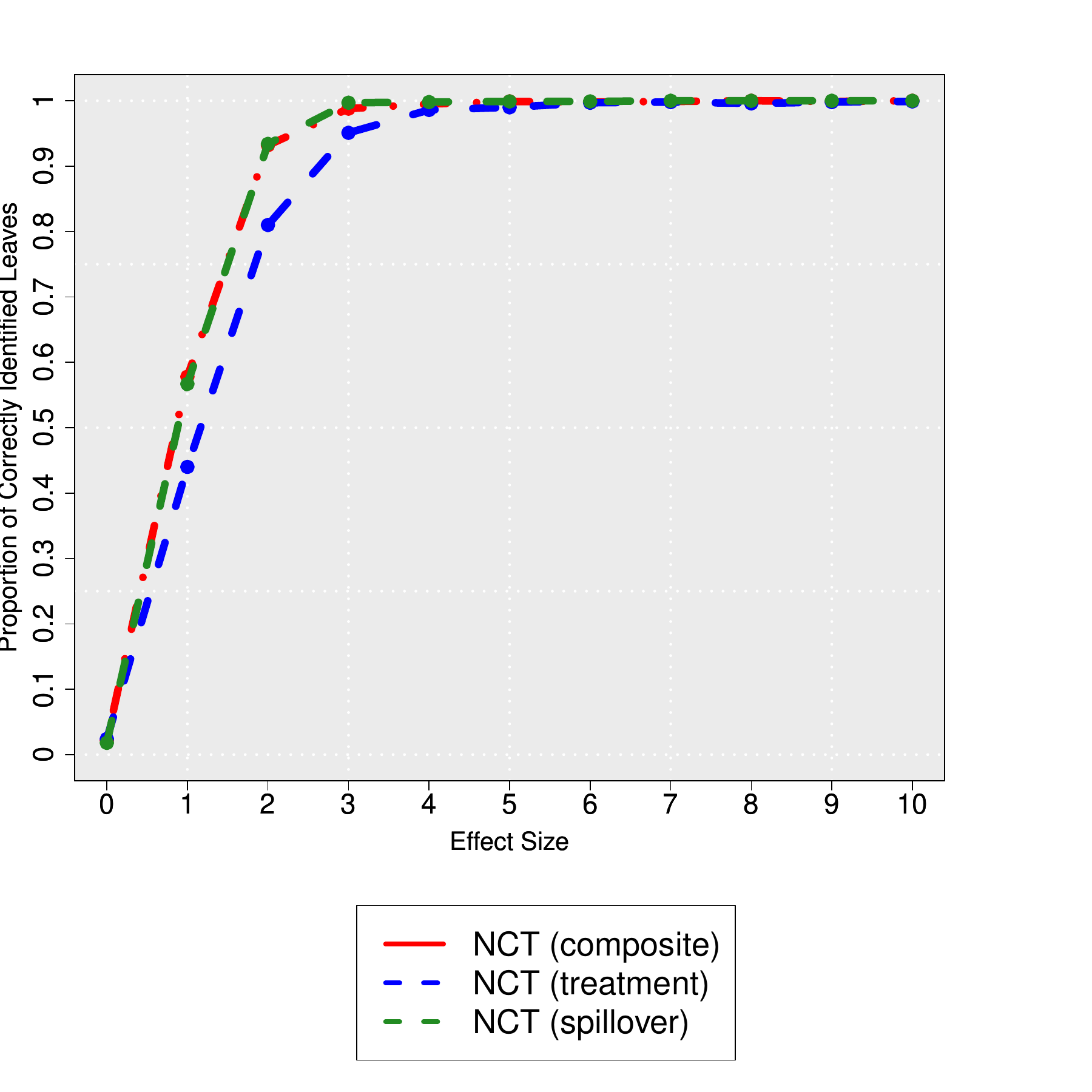}
		\caption{Simulations' results for correctly discovered leaves in the first scenario with mixture covariates}
	\label{fig:mixture}
\end{figure}


\begin{table}[H]
\scriptsize
\centering

\begin{tabular}{cccccccc}
\multicolumn{1}{l}{}                  & \multicolumn{7}{c}{Treatment Effects}                                                                                                                 \\ \cline{2-8} 
\multicolumn{1}{l}{Effect Size}       & $\hat{\tau}_{\ell_1}$   & $\hat{se}(\hat{\tau}_{\ell_1})$   & $\hat{\tau}_{\ell_2}$   & $\hat{se}(\hat{\tau}_{\ell_2})$   & MSE   & Bias   & Coverage \\ \hline
0                                     & 0.149                   & 0.165                             & -0.172                  & 0.189                             & 0.034 & 0.162  & 1.000    \\
\rowcolor{gray!10} 1 & 1.018                   & 0.212                             & -1.053                  & 0.212                             & 0.040 & 0.035  & 0.964    \\
2                                     & 1.984                   & 0.298                             & -1.999                  & 0.302                             & 0.084 & -0.008 & 0.953    \\
\rowcolor{gray!10} 3 & 2.987                   & 0.404                             & -3.004                  & 0.409                             & 0.161 & -0.005 & 0.955    \\
4                                     & 3.980                   & 0.519                             & -3.987                  & 0.520                             & 0.252 & -0.017 & 0.954    \\
\rowcolor{gray!10} 5 & 5.001                   & 0.639                             & -4.981                  & 0.635                             & 0.369 & -0.009 & 0.953    \\
6                                     & 6.014                   & 0.759                             & -5.987                  & 0.760                             & 0.513 & 0.001  & 0.960    \\
\rowcolor{gray!10} 7 & 6.981                   & 0.873                             & -7.012                  & 0.875                             & 0.698 & -0.003 & 0.957    \\
8                                     & 7.987                   & 0.999                             & -8.042                  & 1.020                             & 0.940 & 0.015  & 0.955    \\
\rowcolor{gray!10} 9 & 8.959                   & 1.113                             & -9.021                  & 1.135                             & 1.156 & -0.010 & 0.955    \\
10                                    & 10.024                  & 1.240                             & -9.991                  & 1.240                             & 1.319 & 0.007  & 0.956    \\ \hline
\multicolumn{1}{l}{}                  & \multicolumn{7}{c}{Spillover Effects}                                                                                                                 \\ \cline{2-8} 
\multicolumn{1}{l}{}                  & $\hat{\delta}_{\ell_1}$ & $\hat{se}(\hat{\delta}_{\ell_1})$ & $\hat{\delta}_{\ell_2}$ & $\hat{se}(\hat{\delta}_{\ell_2})$ & MSE   & Bias   & Coverage \\ \hline
0                                     & 0.182                   & 0.157                             & -0.110                  & 0.149                             & 0.028 & 0.137  & 0.955    \\
\rowcolor{gray!10} 1 & 1.006                   & 0.167                             & -1.029                  & 0.169                             & 0.024 & 0.017  & 0.974    \\
2                                     & 2.001                   & 0.225                             & -1.987                  & 0.224                             & 0.039 & -0.006 & 0.971    \\
\rowcolor{gray!10} 3 & 2.991                   & 0.292                             & -3.002                  & 0.292                             & 0.064 & -0.003 & 0.974    \\
4                                     & 4.001                   & 0.368                             & -4.004                  & 0.368                             & 0.099 & 0.003  & 0.980    \\
\rowcolor{gray!10} 5 & 4.983                   & 0.445                             & -4.990                  & 0.444                             & 0.135 & -0.013 & 0.979    \\
6                                     & 5.999                   & 0.524                             & -6.022                  & 0.525                             & 0.191 & 0.011  & 0.980    \\
\rowcolor{gray!10}7  & 6.982                   & 0.605                             & -7.015                  & 0.606                             & 0.249 & -0.001 & 0.988    \\
8                                     & 7.973                   & 0.686                             & -8.024                  & 0.689                             & 0.315 & -0.001 & 0.983    \\
\rowcolor{gray!10} 9 & 9.003                   & 0.769                             & -8.980                  & 0.770                             & 0.403 & -0.008 & 0.981    \\
10                                    & 9.999                   & 0.853                             & -9.947                  & 0.849                             & 0.460 & -0.027 & 0.984    \\ \hline
\end{tabular}

\caption{Simulations' results for the first scenario with mixed covariates.}
\label{tab:mixture}
\end{table}



\end{document}